\documentclass[bj,numbers]{imsartarxiv}

\RequirePackage{amsthm,amsmath,amsfonts,amssymb}
\RequirePackage{graphicx}

\usepackage{array}

\usepackage{float}
\usepackage{natbib}
\usepackage{subfigure}
\usepackage{algorithm}
\usepackage{algorithmicx}
\usepackage{algpseudocode}
\usepackage{soul}
\usepackage{nomencl}
\makenomenclature
\usepackage{multicol}
\usepackage{threeparttable}

\newcommand{\ccond}{{\rm \bf COND}}

\numberwithin{equation}{section}

\renewcommand*\contentsname{Outline}

\newcolumntype{C}[1]{>{\centering\arraybackslash}m{#1}}
\newcolumntype{R}[1]{>{\raggedleft\arraybackslash}m{#1}}

\newtheorem{proposition}{Proposition}

\newtheorem{theorem}{Theorem}
\newtheorem{condition}{Condition}

\newtheorem{corollary}{Corollary}

\newtheorem{definition}{Definition}

\theoremstyle{plain}
\newtheorem{lemma}{Lemma}
\newtheorem{remark}{Remark}
\newtheorem{example}{Example}

\DeclareMathOperator*{\argmin}{arg\,min}
\DeclareMathOperator*{\argmax}{arg\,max}

\newcommand{\haar}{\ensuremath{\textsc{haar}}}
\newcommand{\tr}{\ensuremath{\textsc{tr}}}

\newcommand{\ownparagraph}[1]{\ \\ \noindent {\bf #1.\ }}

\newcommand{\gauss}{\ensuremath{\textsc{gauss}}}
\newcommand{\gaussvar}{\ensuremath{\textsc{g}}}

\newcommand{\Scond}{\ensuremath{{S}_{\textsc{cond}}}}

\newcommand{\ui}{\ensuremath{{\textsc{ui}}}}
\newcommand{\seqrip}{\ensuremath{{\textsc{seq-rip}}}}

\newcommand{\rip}{\ensuremath{{\textsc{rip}}}}
\newcommand{\cond}{\ensuremath{{\textsc{cond}}}}

\newcommand{\commentout}[1]{}

\newcommand{\cU}{\ensuremath{\mathcal U}}

\newcommand{\reals}{\ensuremath{{\mathbb R}}}
\newcommand{\naturals}{\ensuremath{{\mathbb N}}}
\newcommand{\meanspace}{\ensuremath{{\mathtt M}}}
\newcommand{\meanspacenul}{\ensuremath{{\mathtt M}_p}}
\newcommand{\canspace}{\ensuremath{\mathtt B}}

\newcommand{\g}[1]{\ensuremath{#1}}
\newcommand{\gj}[2]{\ensuremath{#1}_{(#2)}}
\newcommand{\si}[2]{\ensuremath{#1}_{#2}}
\newcommand{\sij}[3]{\ensuremath{#1}_{#2(#3)}}
\newcommand{\sn}[2]{\ensuremath{#1}^{#2}}

\renewcommand{\vec}[1]{\ensuremath{{\bm #1}}}
\newcommand{\vm}{\ensuremath{\vec{\mu}}}
\newcommand{\vb}{\ensuremath{\vec{\beta}}}
\newcommand{\hvm}{\ensuremath{\hat{\vec{\mu}}}}
\newcommand{\vmc}{\ensuremath{{\vec{\mu}}^{\circ}}}

\newcommand{\nulhyp}{\ensuremath{\mathcal{P}}}
\newcommand{\althyp}{\ensuremath{\mathcal{Q}}}
\newcommand{\genhyp}{\ensuremath{\mathcal{Q}^{\textsc{gen}}}}
\newcommand{\cQ}{\ensuremath{\mathcal{Q}}}
\newcommand{\nuli}{\ensuremath{p}}
\newcommand{\alti}{\ensuremath{q}}

\newcommand{\pseqrip}{\ensuremath{p_{ \leftsquigarrow q_{\breve{\vm}_{|i-1}}(U_i)}
(\si{U}{i})}}
\newcommand{\Pseqrip}{\ensuremath{P_{ \leftsquigarrow q_{\breve{\vm}_{|i-1}}(U_i)}
(\si{U}{i})}}

\newcommand{\Jp}{\ensuremath{J}}

\newcommand{\peter}[1]{\textcolor{blue}{Peter says: #1}}

\makeatletter
\providecommand{\leftsquigarrow}{%
  \mathrel{\mathpalette\reflect@squig\relax}%
}
\newcommand{\reflect@squig}[2]{%
  \reflectbox{$\m@th#1\rightsquigarrow$}%
}
\makeatother

\startlocaldefs

\endlocaldefs

\begin{document}

\begin{frontmatter}
\title{E-Values for exponential families: \\ the general case}
\runtitle{E-values for exponential families: the general case}

\begin{aug}

\author[A]{\fnms{Yunda}~\snm{Hao}\ead[label=e1]{yundahao@cuhk.edu.hk}}
\author[B,C]{\fnms{Peter}~\snm{Grünwald}
\ead[label=e2]{pdg@cwi.nl}}

\address[A]{
Department of Statistics and Data Science, the Chinese University of Hong Kong,
Hong Kong, China\printead[presep={,\ }]{e1}}

\address[B]{CWI, Science Park 123, 1098 XG Amsterdam, The Netherlands\printead[presep={,\ }]{e2}}

\address[C]{Mathematical Institute, Leiden University, Leiden, The Netherlands}
\end{aug}



\begin{abstract}
We analyze common types of e-variables and e-processes for composite exponential family nulls: the optimal e-variable based on the reverse information projection (RIPr), a conditional (COND) e-variable, and the universal inference (UI) and sequen\-tialized RIPr e-processes. Whereas earlier derivations of the RIPr e-variable, for parametric and nonparametric nulls alike, were restricted to cases in which it reduces to a simple-vs.-simple likelihood, we manage to derive it also in `anti-simple'  cases in which it cannot be so reduced. We characterize the RIPr for simple and Bayes-mixture based alternatives, either precisely (for Gaussian nulls and alternatives) or in an approximate sense (general exponential family nulls). We also provide conditions under which the RIPr e-variable is (again exactly vs. asymptotically) equal to the COND e-variable, and we determine, up to $o(1)$, the e-power of the four e-statistics as a function of sample size.
For $d$-dimensional null and alternative, the e-power of UI tends to be smaller by a term of $(d/2) \log n + O(1)$  than that of the COND e-variable, which is the clear winner.
\end{abstract}



\end{frontmatter}



\section{Introduction}\label{sec:introduction}
Interest in e-values --- a term coined only in 2019 --- has exploded in recent years. Key publications include \cite{wasserman2020universal,GrunwaldHK19,VovkW21,Shafer:2021}; see the introduction \citep{ramdas2023savi} for many more references.
E-values are the values taken on by e-variables and e-processes. E-variables allow for effortless null hypothesis testing under optional continuation to a future study; e-processes additionally allow for optional stopping within a study. 

Here we consider composite null hypotheses {$\nulhyp$}\label{notation:null_hypo} for i.i.d. data {$U_1, U_2, \ldots$}\label{notation:U} with each $U_i$ taking values in some set ${\cal U}$. Recall that an e-variable for sample size $n$ is any nonnegative measurable statistic $S^{(n)}= S^{(n)}(U^n)$ of the data {$U^n = (U_1, \ldots, U_n)$}\label{notation:U_n} such that, under every distribution  {$P \in \nulhyp$}\label{notation:null_P}, we have ${\mathbb E}_P[S^{(n)}] \leq 1.$
For a simple alternative {$\althyp = \{Q \}$}\label{notation:alter_hypo}, expressing {$U_1, U_2, \ldots \sim Q$}\label{notation:Q}, the {\em e-power\/} of e-variable $S^{(n)}$ is given by 
\begin{equation}\label{eq:above}
{\mathbb E}_Q[ \log S^{(n)}]. 
\end{equation}
We consider the e-variable (known variously as GRO (growth rate optimal) or {\em num\'eraire\/}) that maximizes this  e-power over all e-variables on $U^n$. \cite{GrunwaldHK19,ramdas2023savi} provide ample justification to define an analogue of power in this manner. A (perhaps {\em the\/}) central result in e-value theory \citep{GrunwaldHK19,LardyHG24, larsson2024numeraireevariablereverseinformation} says that, under 
mild regularity conditions, this maximizing e-variable is given  by
\begin{equation}\label{eq:ripr}
    S_{Q, \rip}^{(n)} := \frac{q(\sn{U}{n})}{
p_{ \leftsquigarrow q(\sn{U}{n})}(\sn{U}{n})},
\end{equation}
where lower-case letters denote probability densities, and {$p_{ \leftsquigarrow q(\sn{U}{n})}$}\label{notation:RIPr} is the 
{\em reverse information projection (RIPr)\/} of $Q(U^n)$ 
on the set {$\textsc{conv}(\{P(U^n): P \in \nulhyp\})$. Here $\textsc{conv}$}\label{notation:conv} is the convex hull, and for general distributions $R$ for $U_1, U_2, \ldots$, we let $R(U^n)$ denote their marginal on the first $n$ outcomes (see Section~\ref{sec:exponential} for precise definitions of all quantities in this introductory section).  We refer to $S_{Q,\rip}^{(n)}$ alternatively as the RIPr e-variable or simply the {\em optimal\/} e-variable.
Once we have identified the optimal e-variable for simple alternative $\althyp$, there is a standard and easy way to find good e-variables for composite $\althyp$ as well: we represent $\althyp$ by a single $Q^*$ that `learns' in the sense that, for any $Q \in \althyp$, if $U_1, U_2$ are i.i.d. $\sim Q$,  then $q^*(U_i \mid U^{i-1})$ serves as a consistent estimator of $q(U_1)$. This is usually achieved by letting $Q^*$ be a Bayes mixture ('the method of mixtures') or a `prequential (predictive-sequential) plug-in' distribution over $\althyp$  \cite{ramdas2023savi}. For such $Q^*$ we ideally still employ
(\ref{eq:ripr}) with $Q$ set to $Q^*$, but it may be hard to compute. It is often easier to consider its
`sequentialized' (hence potentially sub-optimal) version\label{notation:seqrippy}
\begin{align} \label{eq:seqripnew}
 S_{Q^*,\seqrip}^{(n)}  = \prod_{i=1}^n \frac{q^*(\si{U}{i} \mid \sn{U}{i-1})}{p_{ \leftsquigarrow q^*(\si{U}{i} \mid \sn{U}{i-1})}(\si{U}{i})},  
\end{align}
where for each $\sn{U}{i-1} \in {\cal U}^{i-1}$, $p_{\leftsquigarrow q^*(\si{U}{i} \mid \sn{u}{i-1})}$ is the `local' (in time) RIPr of $Q^*(\si{U}{i}  \mid \sn{U}{i-1} = \sn{u}{i-1})$ on the set
$\textsc{conv}(\nulhyp(U_i))$, i.e. the convex hull of $\nulhyp$ restricted to a single outcome. 

\ownparagraph{A generic issue}
Over the past five years, substantial progress has been made in calculating  $S_{Q,\rip}$ and/or $S_{Q,\seqrip}$ for a wide variety of nulls $\nulhyp$, and with $Q \in \althyp$ or $Q=Q^*$ as above. However, to the best of our knowledge, {\em all\/} these `solved' cases have in common that the RIPrs involved are {\em degenerate}, being achieved by a {\em single\/} element of the null. As a consequence, in all such cases, the resulting e-variable can be written as either 
\begin{equation}\label{eq:simple}
     S_{Q,\rip}^{(n)}= \frac{q(V^n)}{p(V^n)}
 \text{\ (`super-simple case')  \ or\ \ } S_{Q,\seqrip}^{(n)}  = \prod_{i=1}^n \frac{q(\si{U}{i} \mid \sn{U}{i-1})}{p_{i}(U_i)} \text{\ (`simple case')},
\end{equation}
i.e. as a product of simple-vs-simple likelihood ratios, for some fixed $P \in \nulhyp$ (with density $p$) viz. a sequence of elements $P_i$, each in $\nulhyp$, with densities $p_i$ (here $V^n= f(U^n)$ for a measurable function $f$, usually but not always \cite{perez2024estatistics} the identity, and with some abuse of notation we denote the density of $V^n$ also by $q$ resp. $p$). 
We refer to such situations as {\em super-simple\/} and {\em simple}, respectively. The super-simple case includes tests for monotone likelihood ratio families and tests with nuisance parameters expressing a group invariance such as the t-test \cite{perez2024estatistics,GrunwaldHK19,larsson2024numeraireevariablereverseinformation}. The simple case is ubiquitous in nonparametric settings such as testing for the mean of i.i.d. random variables \cite{ramdas2023savi,waudby2024estimating} in which the $\seqrip$ approach is quite natural because $\nulhyp(U_{(1)})$ is itself convex (such e-variables look quite different on the surface, but are really of $S_{Q,\seqrip}$ form \cite[Section 5.2-5.4]{larsson2024numeraireevariablereverseinformation}).

Yet still, many (if not most) testing problems encountered in daily practice are neither simple or super-simple in the sense of (\ref{eq:simple}). In `simple yet not super-simple' settings, reasonable e-variables have been proposed but it has been unknown how much e-power is lost when replacing the optimal-yet-hard-to-compute $S_{Q,\rip}$ by the suboptimal $S_{Q,\seqrip}$. These cases include finite-dimensional linear regression and $k$-sample Bernoulli (contingency table) tests \cite{TurnerLG24,HaoGLLA23}. These are special cases of a general subclass of tests with exponential family nulls that satisfy 
a {\em simplicity condition\/} on  the  covariance matrices corresponding to $\althyp$ and $\nulhyp$,  first identified by \citep{GrunwaldLHBJ24}. For other multivariate exponential family nulls, this condition does not hold, and  no powerful e-values were known to date. 

\ownparagraph{This paper} In this paper we comprehensively analyze e-value and e-power for exponential family nulls, which we may parameterize as 
\begin{equation}\label{eq:expfam}
\nulhyp
 = \{P_{\vec{\mu}}: \vec{\mu} \in \meanspacenul \}\end{equation}
where {$\meanspacenul  \subset \reals^{d}, d 
\geq 1$}\label{notation:mean_param} denotes the mean-value parameter space. 
We identify a counterpart of the covariance {\em simplicity condition\/} on $\nulhyp$ and $\althyp$ mentioned above, denoted the {\em anti-simple case}, under which neither $S_{Q,\rip}$ {nor} $S_{Q,\seqrip}$ can be written as in (\ref{eq:simple}), and we characterize the (nondegenerate) RIPr against simple $\althyp$ for such cases. We show that both in the simple- and anti-simple case, $S_{Q,\rip}^{(n)}$ is well-approximated by another type of {\em conditional\/} e-variable $S_{Q,\cond}$, which is given by
\begin{equation}\label{eq:Scond}
    S_{Q,\cond}^{(n)} :=  \frac{q(\sn{U}{n} \mid Z)}{
p(\sn{U}{n} \mid Z)} := \frac{q(\sn{U}{n} \mid Z)}{
p_W(\sn{U}{n} \mid Z)}, \text{\ with\ } p_W(\sn{U}{n}) := \int p_{\vm}(\sn{U}{n}) d W(\vm) 
\end{equation}
where $Z$ is a minimal sufficient statistic for $\nulhyp$ and $p_W$ the Bayes marginal density based on prior $W$ on $\meanspacenul$. By sufficiency the denominator is identical for any choice of prior $W$, including all degenerate priors putting mass on a single $\vm$, making the left-hand side well-defined and allowing us to drop the subscript $W$.
While exponential family nulls may still feel like a relatively simple setting, they already require development of novel concepts and proof techniques:
our first key insight is that, under the anti-simple condition on their covariances matrices, if $\nulhyp$ is a Gaussian location family and $\althyp$ is also Gaussian, then if we take prior $W_0$ to be a specific a Gaussian with variance $O(1/n)$, we obtain that the marginal distribution of $Z$ under $P_{W_0}$ is identical to its distribution under $Q$ (hence the models are in a sense indistinguishable if we only consider the sufficient statistic $Z$). Combined with the fact that (\ref{eq:Scond}) must hold for this prior,  we show that the RIPr is given by this $p_{W_0}$ and $S_{Q,\cond}$ is equal to $S_{Q,\rip}$ in this Gaussian case. 
We then extend this insight, using a  local central limit theorem with explicit bounds on the error terms \cite{BhattacharyaR76} to analyze the distribution of $Z$ and a nonstandard (because $W_0$ depends on $n$) Laplace approximation of $p_{W_0}$, to find that, still in the anti-simple case, $S_{Q,\cond}$ and $S_{Q,\rip}$ 
achieve asymptotically the same e-power for general exponential family nulls as well. 
Once we have these results in hand, and adding in further results which substantially generalize findings of \cite{grunwald2005asymptotic,kotlowski2011maximum},
we also find precise  asymptotic (up to $o(1)$) expressions for the e-power  of $S^{(n)}_{Q,\seqrip}$, thus quantifying how much is lost by using it instead of $S^{(n)}_{Q,\rip}$. We do the same for yet another popular e-variable, 
\begin{align} \label{eq:uinewb}
S_{Q,\ui}^{(n)}  
:= \frac
{q(\sn{U}{n})}{\sup_{\vm \in \meanspace} p_{\vm}(\sn{U}{n})} = 
\frac
{q(\sn{U}{n})}{p_{\hat{\vm}_{|n}}(\sn{U}{n})},
\end{align}
where the second equality holds whenever the MLE (maximum likelihood estimator) $\hat{\vm}_{|n}$ is well-defined. $S_{Q,\ui}$ is one possible (and in fact quite standard) instantiation of the {\em universal inference} method pioneered by \cite{wasserman2020universal}; see also e.g. \cite{tse2022note}.
By the growth optimality of $S_{Q,\rip}^{(n)}$ we have, for general $Q$, 
\begin{align}\label{eq:basicepower}
    {\mathbb E}_{ Q} [
    \log S_{Q,\rip}^{(n)}
    ] & \geq \max \left\{ 
    {\mathbb E}_{Q} [
    \log S_{Q,\seqrip}^{(n)}
    ],
    {\mathbb E}_{ Q} [
    \log S_{Q,\ui}^{(n)}
    ], 
    {\mathbb E}_{ Q} [
    \log S_{Q,\cond}^{(n)}
    ] \right\}.
    \end{align}
Theorem~\ref{thm:simpleH1gauss} (Gaussian case) and~\ref{thm:simpleH1general} (general exponential family case) refine this statement into precise expressions of e-power for the four e-variables under consideration, not just under the `true' alternative $\althyp= \{Q\}$, but also in the {\em misspecified case\/} when  data are sampled i.i.d. from a distribution $R \neq Q$.

Theorem~\ref{thm:compositeH1gauss} (Gaussian case) and~\ref{thm:compositeH1general} (general case) extend these results to a special type of composite alternatives, namely, a second  exponential family with the same sufficient statistic as the null. Many practical parametric testing problems are of this form, as recalled in and below Example~\ref{ex:twosample}. 
As our most remarkable result (Theorem~\ref{thm:compositeH1general} Part 4), we find that,  when using the method of mixtures, equipping $\althyp$ with a prior $W_1$,  then, under regularity conditions, the RIPr $p_{ \leftsquigarrow q(\sn{U}{n})}(\sn{U}{n})$
is, in an approximate sense, given by the Bayes marginal $P_{W_1}$ for the {\em same\/} prior $W_1$, irrespective of whether we are in the simple or the anti-simple case. 

Additional theoretical insights include that, in the Gaussian case of Theorem~\ref{thm:compositeH1gauss}, the e-variable $S_{\cond}$ coincides with the e-variable obtained by equipping $\nulhyp$ and $\althyp$ with the improper right Haar prior, that is suggested by P\'erez-Ortiz et al. \cite{perez2024estatistics} (Section~\ref{sec:gausscomposite}, (\ref{eq:haariscond})); and that the conditions needed for well-behavedness of the plug-in method and   universal inference are `dual' to each other:  
compare Condition~\ref{cond:plugin}, Section~\ref{sec:generalcomposite}
with Condition~\ref{cond:uinew}, Section~\ref{sec:generalsimple}.

Summarizing (and for simplicity leaving $O(1)$ terms unspecified) our main findings for the --- practically more relevant --- composite $\althyp$ case, we obtain the following relations. These are obtained from  Corollary~\ref{cor:GaussEpowerb} (Section~\ref{sec:gausscomposite}) and Corollary~\ref{cor:generalEpowerb} (Section~\ref{sec:generalcomposite}). Here we used  particular, convenient versions of the plug-in and mixture  method to deal with composite $\althyp$, with precise definitions (including the definition of `strict', `simple', `anti-simple' and $d_{qp}$) in Section~\ref{sec:gauss} and~\ref{sec:general}. We find that, under appropriate (yet mild) regularity conditions on $\nulhyp$ and $\althyp$, that for all $Q \in \althyp$,
    \begin{align}
&  {\mathbb E}_{Q} \left[
    \log {S_{\rip}^{(n)}}/{S_{\cond}^{(n)}} 
    \right] & = \ & o(1). \label{eq:thefirst}\\
&   {\mathbb E}_{ Q} \left[
    \log 
     {S_{\cond}^{(n)}}/{ S_{\ui}^{(n)}}  
    \right] &= \ & \frac{d}{2} \log n + O(1).\label{eq:thefirstB} \\ 
&     {\mathbb E}_{Q} \left[
    \log 
    {S_{\seqrip}^{(n)}}/{S_{\ui}^{(n)}}
    \right] 
   &= \ & \frac{d_{qp}}{2} \log n + O(1)
   \text{\ with\ } 0 < d_{qp} < d, \text{\rm \ in the strict simple case}. 
   \\
 &    {\mathbb E}_{Q} \left[
    \log
    {S_{\seqrip}^{(n)}}/{ S_{\cond}^{(n)}} 
    \right]  &\leq\ & - n \epsilon \label{eq:thelast}
    \text{\rm \ for some $\epsilon > 0$, all large $n$, in the strict anti-simple case}.
    \end{align}
 where $d$ is the dimensionality of the exponential family and  $d_{qp}$ is a notion of `effective dimension' whose exact size depends on $Q$. Note in particular 
that $S_{\seqrip}$ is not competitive in the anti-simple case. 
The resemblance of (\ref{eq:thefirst}) to the ubiquitous BIC model selection criterion is no coincidence: both are derived via a Laplace approximation of a Bayesian marginal likelihood. The occurrence of  $d_{qp} < d$ arises due to the use of plug-in methods --- it also appears in the results on {\em prequential\/} model selection by \cite{grunwald2005asymptotic,kotlowski2011maximum}, whose   techniques for  analyzing log-loss of sequential plug-in estimators we employ and (vastly) generalize. 
Given that, in contrast to $S_{\rip}$, we always know how to calculate $S_{\cond}$, it appears that $S_{\cond}$ is in some sense the clear winner --- a fact that came as a surprise to us, especially since its definition requires neither a prior on the alternative (as in the method of mixtures) nor an estimator (as in the plug-in method): in a sense, it's not ``learning''! However, $S_{\cond}$ does carry one big disadvantage: in contrast to $S_{\ui}$ and $S_{\seqrip}$, usually (i.e. for most exponentially family nulls), $S_{\cond}$ does not define an e-process --- we discuss this further in the final Section~\ref{sec:conclusion}.  
Finally, we emphasize that while some of our results on e-power are asymptotic, the four types of e-variables we employ, and hence the Type-I error guarantees they lead to, are invariably nonasymptotic, i.e. valid at each sample size.

\ownparagraph{Contents} In the remainder of this introductory section, we first (Section~\ref{sec:exponential}) formally define all technical notions referred to in the above. For ease of reference we also provide a list of notations. In Section~\ref{sec:gauss}, we provide our theorems for  multivariate Gaussians, in Section~\ref{sec:general}, we provide the analogous results for general exponential family nulls. Section~\ref{sec:simulations} instantiates these results to two examples for which we provide finite-sample simulations. Section~\ref{sec:gaussproofs} provides the proofs of the Gaussian results, Section~\ref{sec:generalproofs} provides the proofs of the general exponential family results --- these are given in terms of various lemmas that are interesting in their own right and whose proofs are delegated to appendices.
Finally, Section~\ref{sec:conclusion} discusses potential future work. 

\begin{multicols}{2}
\printnomenclature
\end{multicols}


\begin{table*}[tp]
  \centering
  \caption{Summary of notations.}
  \label{notation_part1}
  \small
  \scriptsize
  \footnotesize
  \begin{threeparttable}
    \scalebox{1}{
    \begin{tabular}{l|l|l}
    \hline
    \textbf{Notation}& \textbf{Description}& \textbf{Page}\cr
    \hline
    \hline
    $U \ ;\ {\cal U}$  & $U$ is a random vector representing a single outcome, taking values in ${\cal U}$ & \pageref{notation:U}
    
    \cr
    $U^n$ & $U^n = (U_1, U_2, \ldots, U_n)$, an i.i.d. data sequence
    &\pageref{notation:U_n}
\cr
$X$ & $:= t(U)$, the sufficient statistic for a single outcome: $X=(\gj{X}{1},\ldots,\gj{X}{d}) = (t_1(U),\ldots, t_d(U))$  & \pageref{notation:X}

\cr
$X^n$ & $:= (X_1, \ldots, X_n) = (t(U_1),\ldots, t(U_n))$: the sequence of sufficient statistics in the observed sample &
\pageref{notation:X_n}
\cr
$\mathcal{P}$, $\mathcal{Q}$ & null and alternative hypothesis, respectively & \pageref{notation:alter_hypo}
\cr
$\mathtt{M}_p, \mathtt{M}_q$& mean-value parameter space of $\mathcal{P}, \mathcal{Q}$, respectively & \pageref{notation:mean_paramPP}, \pageref{notation:mean_param_Q}
\cr
$Q, q$, ($Q_{\vm}, q_{\vm}$) & a distribution $(Q)$ and its density $(q)$ in the alternative hypothesis (with mean $\vm$)& \pageref{notation:Q}, \pageref{notation:Q_mu}

\cr
$P, p$, ($P_{\vm}, p_{\vm}$)& a distribution $(P)$ and its density $(p)$ in the null hypothesis (with mean $\vm$)& \pageref{notation:null_P}, \pageref{eq:expfam}

\cr
$p_{W}$, $q_{W}$ & the Bayes marginal density based on prior $W$ on $\mathtt{M}_p$ and $\mathtt{M}_q$, respectively. & \pageref{eq:Scond}, \pageref{notation:qw}

\cr
$R$ & a distribution outside alternative hypothesis and null hypothesis& \pageref{eq:generalizedKL}

\cr
$\textsc{conv}(\mathcal{G})$ & convex hull of set $\mathcal{G}$&
\pageref{notation:conv}
\cr
$p_{\leftsquigarrow q(\sn{U}{n})}$ & reverse information projection (RIPr) of $Q(U^n)$ on the set $\textsc{conv}(\{P(U^n): P \in \mathcal{P} \})$& \pageref{notation:RIPr}
\cr
$p_{\leftsquigarrow q_{\vm}(U_i)}$ &  reverse information projection (RIPr) of $Q_{\vm}(U_i)$ on the set $\textsc{conv}(\{P(U): P \in \mathcal{P} \})$& \pageref{eq:sgauss}

\cr
$S_{Q, \rip}^{(n)}$& $:= \frac{q(U^n)}{p_{\leftsquigarrow q(\sn{U}{n})}}$, growth rate optimal e-variable under $Q(\sn{U}{n})$ with density $q$ & \pageref{eq:ripr}
\cr
$S_{Q,\seqrip}^{(n)}$&  $:= \prod_{i=1}^n \frac{q(\si{U}{i} \mid \sn{U}{i-1})}{p_{ \leftsquigarrow q(\si{U}{i} \mid \sn{U}{i-1})}(\si{U}{i})}$, sequentialized-RIPr e-variable under $Q(\sn{U}{n})$ with density $q$ & \pageref{notation:seqrippy}
\cr
$S_{Q, \cond}^{(n)}, S_{\cond}^{(n)}$ & $:= \frac{q(U^n \vert Z)}{p(U^n \vert Z)}$, $Z$ is a minimal sufficient statistic for $\mathcal{P}$& \pageref{eq:Scond}

\cr
$S_{Q, \ui}^{(n)}$& $:= \frac{q(U^n)}{\sup_{\vm\in \mathtt{M}_p} p_\vm(U^n)}$& \pageref{eq:uinewb}

\cr

$\hat{\vm}_{\vert n}$& $:= \frac{\sum_{i=1}^n x_i}{n}$, i.e. maximum likelihood estimator for $U^n$ on $\mathtt{M}_p$& \pageref{notation:MLE}

\cr

$\breve{\vm}_{\vert n}$ & $:= \frac{x_0 n_0 + \sum_{i=1}^n x_i}{n + n_0}$ for some given $x_0 \in \mathtt{M}_p$ and $n_0 > 0$& \pageref{eq:preq}

\cr

$S_{\breve{\vec\mu}, \seqrip}^{(n)} $ &$:= \prod\limits_{i=1}^n \frac{q_{\breve{\vec{\mu}}_{\vert i-1}}(U_i)}{p_{\leftsquigarrow q_{\breve{\vm}_{\vert i-1}}(U_i)}(U_i)}$& \pageref{eq:sgauss}

\cr

$S_{\breve{\vm},\ui}^{(n)}\  ; \ S_{W_1,\ui}^{(n)}$ & $:= \frac{\prod\limits_{i=1}^n q_{\breve{\vm}_{\vert i-1}}(U_i)}{
\sup_{\vm\in \mathtt{M}_p} p_\vm(U^n)}$ \text{\ and\ } 
$\frac{q_{W_1}(\sn{U}{n})}{\sup_{\vm\in \mathtt{M}_p} p_\vm(U^n) }$, resp.
&
\pageref{eq:sgauss},
\pageref{eq:sgaussb}
\cr

$S_{W_1,\rip}^{(n)}$ &$:= \frac{q_{W_1}(\sn{U}{n})}{
p_{ \leftsquigarrow q_{W_1}(\sn{U}{n})}(\sn{U}{n})}$&
\pageref{eq:sgaussb}
\cr

$d_{ab}$ & trace of $\Sigma_a \Sigma_b^{-1}$, where $\Sigma_a, \Sigma_b$ are positive definite $d \times d$ matrices& \pageref{notation:d_ab}

\cr

$D_{\gauss}(B)$ & $:= \frac{1}{2} \left( - \log \det(B) - \left( d - \tr(B) \right) \right)$, B is an invertible $d\times d$ matrix, $\tr(B)$ is its trace& \pageref{eq:gaussrules}

\cr

$\Sigma_p(\vm), \Sigma_q(\vm)$&  covariance matrix of $X$ under $P_\vm$, $Q_\vm$& \pageref{notation:cov_p}, \pageref{notation:cov_p}

\cr

$D_{\Sigma_r}(\Sigma_q \vert\vert \Sigma_p)$& $:= D_{\gauss}(\Sigma_r \Sigma_p^{-1}) - D_{\gauss}(\Sigma_r \Sigma_q^{-1})$& \pageref{eq:triplegauss}

\cr

$D(R \vert\vert P)$& Kullback-Leibler (KL) divergence between distributions $R$ and $P$ on sample space ${\cal U}$ &
\pageref{notation:KL_RP}
\cr
$D_R(Q \vert\vert P)$& $= D(R \vert\vert P) - D(R \vert\vert Q)$& \pageref{eq:generalizedKL}

\cr
$D(R(\sn{U}{n}) \vert\vert P(\sn{U}{n}))$& Kullback-Leibler (KL) divergence between $R$ and $P$ on $\mathcal{U}^n$& \pageref{notation:KL_RP_on_Un}

\cr

    \hline

    
    \end{tabular}}
    \end{threeparttable}
\end{table*}




\subsection{Technical preliminaries}\label{sec:exponential}
In all our results, the null hypothesis $\nulhyp$ is a regular $d$-dimensional exponential family defined on underlying random element $\g{U}$ taking values in some set $\cU$ and sufficient statistic vector  {$\g{X} = (\gj{X}{1}, \ldots, \gj{X}{d})^\top$}\label{notation:X}. We can write $\gj{X}{j}= t_j(\g{U})$ for given  functions $t_1, \ldots, t_{d}$.
For the Gaussian case, Section~\ref{sec:gauss}, we can take $X=U$, but in general (see e.g. Example~\ref{ex:twosample}), distinguishing between $X$ and $U$ is crucial.
Here and in the sequel we freely use standard properties of exponential families without explicitly referring to their proofs, for which we refer to, e.g. \citep{brown1986fundamentals,BarndorffNielsen78,efron_2022}.
We parameterize $\nulhyp$ in terms of the mean-value parameter space $\meanspacenul$, 
\label{notation:mean_paramPP} 
so that we can write $\nulhyp$ as (\ref{eq:expfam}) with $\meanspacenul  \subset \reals^{d}$
and ${\mathbb E}_{P_{\vec{\mu}}}[X]= \vec{\mu}$. We denote the $d \times d$ covariance matrix of $X$ under $P_{\vec{\mu}}$ as {$\Sigma_p(\vm)$}\label{notation:cov_p}. Recall that $\Sigma_p(\vm)$ is continuous and positive  definite for all $\vm \in \meanspace_p$, and, since the family is regular, $\meanspace_p$ is a convex open set. 
We fix some measure $\nu$ that is mutually absolutely continuous with some, and hence all, $P \in \nulhyp$. We denote the density of $\g{U}$ under  $P_{\vec{\mu}}$ relative to $\nu$  as $p_{\vec{\mu}}$. 

$\nulhyp$ is extended to multiple outcomes by the i.i.d. assumption.  It thus becomes a set of distributions for random process $\si{U}{1}, \si{U}{2}, \ldots$, with for $i \geq 1$, $\si{U}{i}$ an i.i.d. copy of $\g{U}$, and $\si{X}{i} =
(\sij{X}{i}{1}, \ldots, \sij{X}{i}{d})$ with $\sij{X}{i}{j} = t_j(\si{U}{i})$. We abbreviate {$\sn{X}{n} = (\si{X}{1},\ldots, \si{X}{n})$}\label{notation:X_n} and similarly $\sn{U}{n} = (\si{U}{1},\ldots, \si{U}{n})$ and write $P_{\vm}(U^n)$ (density $p_{\vm}(U^n)$) for  the marginal distribution of $U^n$ under $P_{\vm}$.
We use the notation $\nulhyp$ to refer both to the set of distributions for a single outcome $U$ and for the random processes $(U_i)_{i \in \naturals}$, as in each instance it will be clear what is meant. When referring to the set of marginal distributions for the first $n$ outcomes, we use the notation $\nulhyp(U^n) := \{ P_{\vm}(U^n): \vm \in \meanspace_p \}.$

\ownparagraph{MLE and the UI e-variable}
Since we assume $\nulhyp$ to be regular, the maximum likelihood estimator (MLE) in the mean-value parameterization, {$\hvm_{|n}$}\label{notation:MLE}, based on data $U^n$ exists, is unique and equal to $n^{-1} \sum_{i=1}^n X_i$ whenever the latter quantity lies in the set $\meanspacenul$.
With slight abuse of notation, we shall extend the definition of  $\hvm_{|n}$ and simply set it to be equal to  $n^{-1} \sum_{i=1}^n X_i$ even if the latter quantity is not contained in $\meanspace_p$; this can happen, for example, with the Bernoulli distribution if all $X_i$ are equal to $1$, or all are equal to $0$. For such cases we set $p_{{\hvm}_{|n}}(U^n) := \sup_{\vm \in \meanspacenul} p_{\vm}(U^n)$.
In this way, the UI e-variable $S_{Q,\ui}$ can always be written in the rightmost form of (\ref{eq:uinewb}). 

\ownparagraph{The alternative $\althyp$ and the COND e-variable}
The {\em alternative\/} hypothesis $\althyp$ will either (Theorem~\ref{thm:simpleH1gauss} and~\ref{thm:simpleH1general}) be a singleton $\althyp = \{Q \}$ or (Theorem~\ref{thm:compositeH1gauss} and~\ref{thm:compositeH1general}) will itself be another regular exponential family, defined on $\cU$, with the same sufficient statistic $X$ as $\nulhyp$. In both cases we assume that all elements $Q\in \althyp$ are mutually absolutely continuous with $\nu$ and denote by $q$ the density of $Q$ relative to $\nu$.  For the  case that $\althyp$ is an exponential family, we extend all notation in the obvious way: {$\meanspace_q$}\label{notation:mean_param_Q} denotes the mean-value parameter space, $\canspace_q$ is a canonical parameter space, {$\Sigma_q(\vm)$}\label{notation:cov_q} is the $d \times d$ covariance matrix corresponding to {$Q_{\vm}$}\label{notation:Q_mu}, and so on.
Again, $\althyp$ is  extended to sequences of outcomes by the i.i.d. assumption.

In the definition (\ref{eq:Scond}) of $S_{Q,\cond}$, we invariably take 
 $Z= n^{-1/2} \sum_{i=1}^n (X_{i,1}, \ldots, X_{i, d})^\top= n^{1/2} \hat{\vm}_{|n}$.
 $p(\cdot \mid Z)$ is set to be the density of $P \mid Z$, the regular  conditional distribution of $\sn{U}{n}$ given $Z$, which is identical for all $P \in \nulhyp$ because $Z$ is a sufficient statistic and the family is regular. Similarly, we set $q(\cdot \mid Z)$ to be the density of $Q \mid Z$, the regular conditional distribution of $\sn{U}{n}$ given $Z$. Again, this distribution is identical for all $Q \in \althyp$ because either $\althyp=\{Q \}$ is simple or because $Z$ is sufficient for $\althyp$ and $\althyp$ is regular. 
 Both densities $p(\cdot \mid Z)$ and $q(\cdot \mid Z)$    are taken relative to a common underlying (random, since it depends on $Z$) measure $\nu'(\cdot \mid Z)$ which we fix once and for all. We may for example take $\nu' \mid Z$ equal to $(1/2) (P \mid Z) + (1/2) (Q\mid Z)$, but our theorems hold under any other choice as well, since the corresponding $S_{Q,\cond}$s must be almost surely equal  under all $P \in \nulhyp$ and $Q \in \althyp$. 
 To see that 
$S_{Q,\cond}^{(n)}$ indeed is an e-variable, note that for all $P \in \nulhyp$, ${\mathbb E}_P[S_{Q,\cond}^{(n)} ]$ may be rewritten as 
 \begin{align}\label{eq:laura}
{\mathbb E}_{Z \sim P} [ {\mathbb E}_{U^n\sim P}[S_{Q,\cond}^{(n)} |Z ] ] =
  {\mathbb E}_{Z \sim P} \left[ \int p(U^n \mid Z) \cdot \frac{q(U^n\mid Z)}{p(U^n \mid Z)} {\rm d} \nu'(U^n \mid Z)  \right] 
  =
 {\mathbb E}[1 ] = 1.
\end{align}

\ownparagraph{KL Divergence, RIPr and seq-RIPr}
The {\em KL (Kullback-Leibler) divergence\/} $D(\cdot \| \cdot)$ plays a central role in our analysis. For $R,P$ distributions for i.i.d. random process $U_1, U_2, \ldots$ as above, we write {$D(R(V) \| P(V))$}\label{notation:KL_RP} to denote the KL divergence between the $R$- and $P$-marginal distributions for random vector $V$, respectively; for example,  {$D(R(U^n) \| P(U^n))$}\label{notation:KL_RP_on_Un}. Whenever we write $D(R \| P)$, this is meant to abbreviate $D(R(U) \| P(U))$. We further define (as is common, see e.g. \cite{kotlowski2011maximum}) the {\em generalized KL divergence\/} between $Q$ and $P$ under sampling distribution $R$ (also assumed to have a density relative to $\nu$) as  
\begin{equation}\label{eq:generalizedKL}
D_R(Q \|P) := {\bf E}_{U \sim R}\left[ \log \frac{q(U)}{p(U)}
\right] \overset{(a)}{=}
D(R \| P) - D(R \| Q),
\end{equation}
where (a) holds whenever either  $D(R \| P)$ or $D(R \| Q)$ is finite. Yet $D_R(Q \| P)$ is still well-defined in some cases in which both $D(R \| P)$ or $D(R \| Q)$ are $\infty$ \citep{LardyHG24}. 

Now, let $Q^*$ be an arbitrary distribution for random process $\si{U}{1}, \si{U}{2}$ such that its conditional densities $q^*(\si{x}{1}), q^*(\si{x}{2} \mid \si{x}{1}), \ldots$ relative to the chosen background measure $\nu$ are well-defined. 
Gr\"unwald et al. \cite{GrunwaldHK19} define the GRO (growth-rate optimal) e-variable for a sample of size $n$, relative to $Q^*$, to be the e-variable $S^{(n)}$ that, among all e-variables that can be written as a function of data $\sn{U}{n}$, maximizes {\em growth-rate} (\ref{eq:above}) (with $Q$ set to $Q^*$), also known as  {\em e-power\/} \cite{wang2023ebacktesting}, ${\mathbb E}_{\sn{U}{n} \sim Q}[\log S^{(n)}]$, and show that it is given by (\ref{eq:ripr}). 
More precisely, \cite[Theorem 1]{GrunwaldHK19} implies that a sub-probability density $p_{ \leftsquigarrow q^*(\sn{U}{n})}$ exists, such that (\ref{eq:ripr}) gives the growth-optimal e-variable maximizing (\ref{eq:above}) (with $Q$ set to $Q^*$) whenever $D(Q^*(U^n) \| P_{\vm}(U^n)) < \infty$ for all $\vm \in \meanspace_p$ (this result has later been vastly generalized \cite{LardyHG24,larsson2024numeraireevariablereverseinformation} but for us the initial, simple version suffices). \cite[Theorem 1]{GrunwaldHK19} further implies that {\em if\/}
there exists a prior $W$ on $\meanspace_p$ such that the Bayes marginal $p_{W} (\sn{U}{n})$ as defined in (\ref{eq:Scond}) 
achieves $\min D(Q^* \|P_W)$, with the minimum over all distributions $W$ on $\meanspace_p$, {\em then\/} 
$p_{ \leftsquigarrow q^*(\sn{U}{n})}(\sn{U}{n})= p_W(\sn{U}{n})$. In the theorems below we encounter various cases in which such a prior does exist, sometimes degenerate ($W(\{\vm\}=1$ for a single $\vm \in \meanspace_p$) and sometimes Gaussian. 

Applying this RIPr existence result for the case $n=1$ with $Q^*$ substituted by  $Q^*(U_i \mid U^{i-1}= u^{i-1})$ shows that, whenever
$D({Q^*}(U_i \mid U^{i-1}= u^{i-1}) \| P_{\vm}(U_i)) < \infty$ for all $\vm \in \meanspace_p$, the local (in time) RIPr 
$p_{\leftsquigarrow q^*(\si{U}{i} \mid \sn{u}{i-1})}$ of $Q^*(\si{U}{i}  \mid \sn{U}{i-1} = \sn{u}{i-1})$ on the set
$\textsc{conv}(\nulhyp(U_i))$ is well-defined for all $i$, and then so is $S^{(n)}_{Q^*,\seqrip}$ as in (\ref{eq:seqripnew}).



\section{The Gaussian location family}\label{sec:gauss}
\subsection{$\althyp$ simple, $\nulhyp$ multivariate Gaussian location}
\label{sec:gausssimple}
Let the null hypothesis $\nulhyp$ be the Gaussian location family for  $\g{X} = U=  (\gj{X}{1},\ldots, \gj{X}{d})$ with nondegenerate $d \times d$ covariance matrix $\Sigma_{\nuli}$. Fix a particular mean vector $\vec{\mu}^* \in \meanspacenul =  \reals^{d}$. We take as alternative $\althyp= \{Q \}$ with $Q$ a Gaussian for $X$ with mean $\vec{\mu}^*$ but with nondegenerate covariance matrix $\Sigma_{\alti} \neq \Sigma_{\nuli}$.

To prepare for our results, we develop simplified expressions for the generalized KL divergence (\ref{eq:generalizedKL}) that hold in this Gaussian case. 
First, for invertible $d \times d$ matrix $B$, we let 
\begin{equation}\label{eq:gaussrules}
D_{\gauss}(B) := \frac{1}{2} \left( - \log \det(B) - \left( d - \tr(B) \right)   \right),
\end{equation}
where $\tr (B)$ is the trace of $B$ and $\det(B)$ is the determinant of $B$. The subscript derives from the fact that $D_{\gauss}(\Sigma_{\alti} \Sigma^{-1}_{\nuli})$ is the KL divergence between two $d$-dimensional Gaussians that share the same (arbitrary) mean vector and have covariances $\Sigma_{\alti}$ and $\Sigma_{\nuli}$ respectively --- which also tells us that for general positive definite and symmetric  $\Sigma_{\alti}, \Sigma_{\nuli}$,
$D_{\gauss}(\Sigma_{\alti} \Sigma_{\nuli}^{-1}) \geq 0$ with equality iff $\Sigma_{\alti} = \Sigma_{\nuli}$. 
The following characterization, derived from standard properties of determinant and trace,  will prove useful below: letting $\lambda_1, \ldots, \lambda_{d}$ be the eigenvalues of $\Sigma_{\nuli}^{-1/2} \Sigma_{\alti} \Sigma_{\nuli}^{-1/2}$, we have
\begin{align}
    \label{eq:HVAchar}
    D_{\gauss}(\Sigma_{\alti} \Sigma_{\nuli}^{-1}) = D_{\gauss}(\Sigma_{\nuli}^{-1/2} \Sigma_{\alti} \Sigma_{\nuli}^{-1/2}) 
    = \frac{1}{2} \left( \sum_{j=1}^{d} \left( - \log \lambda_j - (1- \lambda_j) \right) \right).
\end{align}

{It is helpful to introduce the following nonstandard notation: for positive definite $d \times d$ matrices $\Sigma_a, \Sigma_b$, we set  {$d_{ab} := \tr(\Sigma_a \Sigma_b^{-1})$}\label{notation:d_ab}.} Using this definition, we further define:
\begin{align}\label{eq:triplegauss}
D_{\Sigma_r}(\Sigma_q \| \Sigma_p) :=  
D_{\gauss}(\Sigma_r \Sigma_p^{-1}) - D_{\gauss}(\Sigma_r \Sigma_q^{-1}) = - \frac{1}{2} \log \det (\Sigma_q \Sigma_p^{-1}) + \frac{d_{rp} - d_{rq}}{2} 
\end{align}
as a generalization (in the sense that $D_{\Sigma_q}(\Sigma_q \| \Sigma_p) = D_{\gauss}(\Sigma_q \Sigma_p^{-1})$\;), of (\ref{eq:gaussrules}). 
Then, for any distribution $R$ on $X$ with mean $\vm^*$ and covariance $\Sigma_r$, we have (as  shown in Section~\ref{sec:gaussproofs} as part of the proof of the theorem below), that with $D_R$ the generalized KL divergence as in (\ref{eq:generalizedKL}),
\begin{equation}
    \label{eq:forevergauss}
D_R(Q_{\vm^*} \| P_{\vm^*}) = D_{ \Sigma_r}(\Sigma_q \| \Sigma_p).
\end{equation}
The prime interest of the theorems below occurs if we analyze e-power under the alternative $Q$, i.e. for the case that $R=Q$. Yet, all results hold more generally under any other sampling distribution $R$
with the same mean $\vm^*$ as $Q$, so we will state and prove our results for such general $R$.
We shall refer to the case $R \neq Q$ as  
{\em misspecified}, in contrast to the {\em well-specified\/} case $R=Q$. Allowing $R\neq Q$ mainly adds strength to Theorem~\ref{thm:compositeH1gauss} for composite alternatives further on; in Theorem~\ref{thm:simpleH1gauss} directly below, if a statistician employs $S_Q$ it implies she has knowledge of $\vm^*$, so the added generality of allowing general $R$ with the same $\vm^*$ is somewhat limited.

In (\ref{eq:gausscond}) below and later in (\ref{eq:condB}) we use both notations $D_R$ and $D_{\Sigma_r}$ simultaneously--- this is done to ease comparison to the exponential family version of our results stated in Section~\ref{sec:exponential}.
In particular,
\begin{align}\label{eq:dlite}
 d_{qq}= d_{pp} = d, 0 < d_{qp}
   \text{\ and if $\Sigma_q - \Sigma_p$ is negative semidefinite, then\ }
   d_{qp} \leq d,
\end{align}
the latter inequality becoming strict if $\Sigma_q - \Sigma_p$ is negative definite. To derive the inequality, note that if $\Sigma_q - \Sigma_p$ is negative definite, then $\Sigma_{\nuli}^{-1/2} \Sigma_{\alti} \Sigma_{\nuli}^{-1/2} - I$ is negative definite, and then all eigenvalues $\lambda_j$ above are smaller than 1, 
so $d_{qp} 
    = \sum_{j=1}^d \lambda_j$ is smaller than $d$.



\begin{theorem}\label{thm:simpleH1gauss}
Let $\mathcal{P}$ and $Q$ be as above
and  let $R$ be a distribution on $X=U$ with the same mean ${\mathbb E}_{R}[\g{X}] = \vm^*$ as $Q$ and with covariance matrix $\Sigma_r$. Let $U_1, U_2, \ldots$ be i.i.d. $\sim R$. Then $D_R(Q \|P_{\vm^*})$ is finite and we have:
\begin{enumerate}
    \item {\bf (UI)} Let $S_{Q,\ui}^{(n)} = \frac{q(\sn{U}{n})}{p_{\hvm_{|n}}(\sn{U}{n})}$ be defined as in (\ref{eq:uinewb}). We have:
    \begin{align}\label{eq:gaussui}
{\mathbb E}_{R} [
    \log S_{Q,\ui}^{(n)}
    ] = n D_R(Q \| P_{\vec{\mu}^*}) - 
\frac{d_{rp}}{2}.
   \end{align}
\item {\bf (COND)} Let $S_{Q, \cond}^{(n)} = \frac{q(\sn{U}{n} \mid Z)}{
p(\sn{U}{n} \mid Z)}$ be as in (\ref{eq:Scond}). We have: 
    \begin{align}\label{eq:gausscond}
{\mathbb E}_{R} [
    \log S_{Q,\cond}^{(n)}
    ] 
    = (n-1) D_R(Q \| P_{\vec{\mu}^*})
    = n \cdot D_R(Q \| P_{\vec{\mu}^*}) - D_{\Sigma_r}(\Sigma_q \| \Sigma_p).
   \end{align}
\item {\bf (seq-RIPr/RIPr Simple Case; seq-RIPr Anti-Simple Case)} Let $S_{Q,\rip}^{(n)}, S_{Q,\seqrip}^{(n)}$ be as in (\ref{eq:ripr}) and (\ref{eq:seqripnew}). If $\Sigma_{\alti}- \Sigma_{\nuli}$ is negative semidefinite (the `simple' case), then we have, for all $n$: 
\begin{align}\label{eq:raar}
S_{Q,\rip}^{(n)} = S_{Q,\seqrip}^{(n)} =  \frac{q(\sn{U}{n})}{p_{\vec{\mu}^*}(\sn{U}{n})} 
\text{\ \ so that\ \ } 
{\mathbb E}_{R} [
    \log S_{Q,\rip}^{(n)}
    ] = n D_R(Q \| P_{\vec{\mu}^*}).
   \end{align}
   If $\Sigma_{\alti}- \Sigma_{\nuli}$ is positive semidefinite (`anti-simple case'), then $S^{(n)}_{Q,\seqrip} = 1$, i.e. it is trivial. 
\item  {\bf (RIPr, Anti-Simple Case)} 
 Let $S_{Q,\rip}^{(n)}, S_{Q,\cond
}^{(n)}$ and $S_{Q,\seqrip}^{(n)}$ be defined as in (\ref{eq:ripr}), (\ref{eq:Scond}) and (\ref{eq:seqripnew}). If $\Sigma_{\alti}- \Sigma_{\nuli}$ is positive semidefinite, then  we have
\begin{equation}\label{eq:ripisgauss}
S_{Q,\rip}^{(n)} = \frac{q(U^n)}{p_{W_0}(U^n)} =
S_{Q,\cond}^{(n)},
\end{equation}
where $W_0$ is a Gaussian prior with mean $\vec{\mu}^*$ and covariance matrix $\Pi_0 :=(\Sigma_{\alti}- \Sigma_{\nuli})/n$. $P_{W_0} = P_{\leftsquigarrow q(U^n)}$ is then the RIPr of $Q$ onto $\textsc{conv}(\nulhyp(U^n))$.  
\end{enumerate}
\end{theorem}
\begin{remark}[Simple vs. Anti-Simple]
{\rm We see that, if $\Sigma_q- \Sigma_p$ is negative semidefinite, then the RIPr is simply an element of $\nulhyp$ and the growth-optimal e-variable is of the same form as it would be if $\nulhyp$ were simple. Following \cite{GrunwaldLHBJ24} and as anticipated in the introduction we call this the {\em simple\/} case and indeed Part 3 of the theorem is really a direct corollary of the main result of that paper (all other parts are novel). We will formalize the notion of `simplicity' for general exponential families in Section~\ref{sec:general}.}
\end{remark}
\begin{remark}[{Remark: comparing e-power}]
{\rm  
In the {\em strictly} simple case that  $\Sigma_{\alti} - \Sigma_{\nuli}$ is negative  definite, we have, already mentioned below (\ref{eq:dlite}), that all $\lambda_j$ in (\ref{eq:HVAchar}) are smaller than 1. Similarly, we are in the strict anti-simple case iff all these $\lambda_j$ are greater than 1. 
This eigen-characterization leads to the following corollary about the e-power for the strict anti-simple and simple cases of UI vs. the other e-variables. 
In the strict anti-simple case,
we have, by Part 4 above, $S_{Q,\cond}^{(n)}= S_{Q,\rip}^{(n)}$, so, using (\ref{eq:gausscond}) and (\ref{eq:HVAchar}) and the fact that in this strict anti-simple case, all eigenvalues $\lambda_j$ are  larger than $1$, 
we find (\ref{eq:eigenWaan}) below. Using (\ref{eq:dlite}) for the strict simple case, we find (\ref{eq:eigenAardig}):}
\end{remark}
\begin{corollary}{\bf [e-power and growth optimality]}
 \label{cor:GaussEpowera}
In the strict anti-simple case,  \begin{align}
 {\mathbb E}_{Q} \left[
    \log {S_{Q,\rip}^{(n)}}/{
    S_{Q,\ui}^{(n)}}
    \right]  =  
    & 
    \frac{d_{qp}}{2}
  - D_{\Sigma_q}(\Sigma_q \| \Sigma_p) 
  \label{eq:eigenWaan} 
  =
  \frac{d}{2} + \frac{1}{2} \sum_{j=1}^d \log \lambda_j >  \frac{d}{2},     
  \text{\ \it whereas}
\\ 
\label{eq:eigenAardig}
     {\mathbb E}_{Q} \left[
    \log S_{Q,\rip}^{(n)}/ 
    S_{Q,\ui}^{(n)}
    \right]    = &   
    \frac{d_{qp}}{2} 
    < \frac{d}{2},
    \text{\ \it in the strict simple case}. 
\end{align}
\end{corollary}
Even though $S_{Q,\ui}$ therefore always has less e-power than $S_{Q,\rip}$, the difference (in contrast to the composite $\althyp$ case in Section~\ref{sec:gausscomposite} below) does not keep growing with $n$.
The conditional e-variable $S_{Q,\cond}$ is identical to $S_{Q,\rip}$ in the anti-simple case but in the simple case it is hard to compare to $\ui$; in general neither one outperforms the other. Again, with composite $\althyp$, the situation changes.   

\subsection{$\nulhyp$, $\althyp$ both multivariate Gaussian location}
\label{sec:gausscomposite}
We now consider the case that $\nulhyp$ is a $d$-dimensional Gaussian location family 
as before, but now $\cQ$ is composite: it is itself the full $d$-dimensional Gaussian location family with nondegenerate covariance matrix $\Sigma_q\neq \Sigma_p$. 
The mean-value parameter spaces are   \(\mathtt{M}_q = \mathtt{M}_p = \mathbb{R}^d\). 

We again establish explicit formulae and corresponding expected logarithmic growth of our various e-variables. 
Since $\althyp$ is composite, we need a method to estimate or `learn' a distribution in $\althyp$ given a data sequence \( X_1, \ldots, X_n \). 
As stated in the introduction, we use two standard methods for this. The first method, usually called the plug-in method, is to use at each $i$ a regularized ML estimator based on the past and defining  $Q$ as the product of predictive distributions.
We use the variation studied by \cite{grunwald2005asymptotic} who, following \cite{Dawid84}, call this the {\em prequential ML method}, setting, for each $i$,  
\begin{equation}
    \label{eq:preq}
Q^*(U_i \mid U^{i-1}) := 
Q_{\breve{\vec{\mu}}_{|i-1}} (U_i), \text{\ where\ }
\breve{\vec{\mu}}_{|n} := \frac{x_0 n_0 + \sum_{i=1}^n x_i}{n + n_0}
\end{equation}
for some fixed constant $n_0 > 0$ (not necessarily an integer) and $x_0 \in \meanspace_p$. 
We take $n_0 > 0$ to ensure that $\breve{\vm}_{|0}$ is well-defined and also, when later applied to exponential families rather than Gaussians, to make sure that the relevant KL divergences remain finite. 

Whenever we use an e-variable of the form $S_{Q^*,\cdot}$ with $Q^*$ defined by (\ref{eq:preq}), we abbreviate this to $S_{\breve{\vm},\cdot}$. 
We use this method for e-variables of UI and sequential RIPr type. In particular, we set
\begin{align}
    \label{eq:sgauss}
    &
     S_{\breve{\vec\mu},\ui}^{(n)} = \frac{\prod_{i=1}^n q_{\breve{\vec{\mu}}_{|i-1}}(U_i)}{p_{\hvm_{|n}}(\sn{U}{n})},    
 \ \ \  S_{\breve{\vec\mu}, \seqrip}^{(n)} = \prod\limits_{i=1}^n \frac{q_{\breve{\vec{\mu}}_{|i-1}}(U_i)}{
\pseqrip 
},
\end{align}
where
$\pseqrip$
is the RIPr of $Q_{\breve{\vec{\mu}}_{|i-1}}(U_i)$ onto $\textsc{conv}(\mathcal{P}(U_i))$
as in (\ref{eq:seqripnew}).  

The second method to learn the alternative as data comes in, known as Robbins' method of mixtures \cite{ramdas2023savi}, is to set $Q^*$ to a Bayes marginal distribution, $Q^*(U^n) :=Q_{W_1}(U^n)$
where {$q_{W_1}(U^n) = \int q_{\vm}(U^n) dW_1(\vm)$}\label{notation:qw} for some prior $W_1$ on $\meanspace_q$. In Theorem~\ref{thm:compositeH1gauss} below we uniquely consider Gaussian priors, i.e. priors of the form $W_1=N(\vm_1,\Pi_1)$ where we implicitly fix the dimension to $d$, i.e. $\vm_1 \in \reals^d$  and $\Pi_1$ a $d \times d$ nondegenerate covariance matrix. 
We use this method for the UI e-variables and for the growth-optimal e-variables, i.e. of RIPr type, in the anti-simple case. 
In particular, we set
\begin{align}
    \label{eq:sgaussb}
    &
      S_{W_1,\ui}^{(n)} = \frac{q_{W_1}(U^n)}{p_{\hvm_{|n}}(\sn{U}{n})},
 \ \ \  
S_{W_1,\rip}^{(n)} = \frac{q_{W_1}(\sn{U}{n})}{
p_{ \leftsquigarrow q_{W_1}(\sn{U}{n})}(\sn{U}{n})}
\end{align}
with 
$p_{ \leftsquigarrow q_{W_1}(\sn{U}{n})}$  the RIPr of $Q_{W_1}(U^n)$ onto $\textsc{conv}(\nulhyp(U^n)
)$ as in (\ref{eq:ripr}).

\begin{theorem}\label{thm:compositeH1gauss}
Let $\mathcal{P}, \mathcal{Q}$ be as above, and let $R$ be a distribution on $X$($= U$) with mean ${\mathbb E}_{R}[\g{X}] = \vm^*$ and with covariance matrix $\Sigma_r$. Let $U_1, U_2, \ldots$ be i.i.d. $\sim R$. Then $D_R(Q_{\vm^*} \| P_{\vm^*})$ is finite  and we have: 
\begin{enumerate}
    \item {\bf (UI)} Let $S_{\breve{\vec\mu},\ui}^{(n)}$ be as in (\ref{eq:sgauss}) above. We have:  
    \begin{align}\label{eq:seq_UI_gauss}
   {\mathbb E}_{\sn{U}{n} \sim R} [
    \log S_{\breve{\vec\mu},\ui}^{(n)}
    ] = 
n \cdot D_R(Q_{\vm^*} \| P_{\vm^*})
 -    \frac{d_{rq}}{2} \log n 
 + O_{\breve{\vm},\ui}(1),
 \end{align}
 where, with $O_{\textsc{a}},O_{\textsc{b}}$ given by (\ref{eq:Oabc}) below, 
 \begin{align*}
O_{\breve{\vm},\ui}(1)=   -  \frac{d_{rp}}{2} - \frac{1}{2} \cdot O_{\textsc{a}}( \| x_0 - \vm^* \|_2^2) - \frac{d_{rq}}{2}\cdot O_{\textsc{b}}(1).
   \end{align*}
   Similarly, let
$S_{W_1,\ui}^{(n)}$ be as in (\ref{eq:sgaussb}) above. We have:  
    \begin{align}\label{eq:seq_UI_gaussb}
   {\mathbb E}_{\sn{U}{n} \sim R} [
    \log S_{W_1,\ui}^{(n)}
    ] = n \cdot 
   D_R(Q_{\vm^*} \| P_{\vm^*}) - \frac{d}{2} \log \frac{n}{2 \pi}
    + O_{W_1,\ui}(1),
    \end{align}
where a precise expression for  $O_{W_1,\ui}$ is given in (\ref{eq:Oabc}) below; up to $o(1)$, it is: 
\begin{align}\label{eq:OWui}
 O_{W_1,\ui}(1) 
    = \frac{d_{rq} - d_{rp}}{2} + \frac{1}{2} \log \det \Sigma_q + \log  w_1(\vm^*)  + o(1).
   \end{align}
\item{ \bf (COND)} Let $S_{\cond}^{(n)}$ be as in (\ref{eq:Scond}) above. We have: 
\begin{align}
    \label{eq:condB}
       {\mathbb E}_{R} [
    \log S_{\cond}^{(n)}
    ] =& (n -1) \cdot D_{\Sigma_r}(\Sigma_q \| \Sigma_p) = 
    n \cdot  D_R(Q_{\vm^*} \| P_{\vm^*}) -   D_{  \Sigma_r}(\Sigma_q \| \Sigma_p).
\end{align}
\item {\bf (seq-RIPr, Simple and Anti-Simple Case)} 
Suppose that $\Sigma_{\alti}- \Sigma_{\nuli}$ is negative semidefinite.
Let $S_{\breve{\vec\mu}, \seqrip}^{(n)}$ be as in (\ref{eq:sgauss}) above. We have: 
$\Pseqrip
= P_{\breve{\vec{\mu}}_{|i-1}}(U_i)$ and
\begin{align}\label{eq:Gaussseq_LVA}
{\mathbb E}_{R} [
    \log S_{\breve{\vec\mu},\seqrip}^{(n)}
    ] 
    = n \cdot D_R(Q_{\vm^*} \| P_{\vm^*})
    + \frac{d_{rp} - d_{rq}}{2} \log n
    + O_{\seqrip}(1),
    \end{align} 
    where  
$O_{\seqrip}(1) =  -
O_{\textsc{a}}( \| x_0 - \vm^* \|_2^2) - \frac{d_{rp} - d_{rq}}{2} \cdot O_{\textsc{b}}(1)$.
If $\Sigma_{\alti}- \Sigma_{\nuli}$ is positive semidefinite, then $S_{\breve{\vec\mu},\seqrip}^{(n)}
= 1$ is trivial.

\item  {\bf (RIPr, General Case)}
Let $S_{W_1,\rip}^{(n)}$ be as in (\ref{eq:sgaussb}) so that $S_{W_1,\rip}^{(n)}$ is the GRO e-variable. Then for all $n$ such that $\Pi_1 +(\Sigma_{\alti}- \Sigma_{\nuli})/n$ is positive semidefinite (in particular, in the Anti-Simple case, $\Sigma_{\alti}- \Sigma_{\nuli}$ is itself positive semidefinite so this will then hold for all $n$), 
\begin{equation}\label{eq:W0depends}
p_{ \leftsquigarrow q_{W_1}(\sn{U}{n})}(\sn{U}{n})=p_{W_0}(\sn{U}{n})
\text{\ where\ }
W_0 = \mathcal{N}(\vec{\mu}_1, \Pi_1 + ({\Sigma_{\alti}- \Sigma_{\nuli}})/n),\end{equation}
and we have
\begin{equation}\label{eq:ripisgaussb}
S_{W_1,\rip}^{(n)} = \frac{q_{W_1}(U^n)}{p_{W_0}(U^n)} =
S_{\cond}^{(n)}.
\end{equation}
\commentout{
\begin{align}\label{eq:GaussGlobal_HVA}
{\mathbb E}_{\sn{U}{n} \sim R} [
    \log S_{W_1,\rip}^{(n)}
    ] 
    = (n-1) \cdot D_R(Q_{\vm^*} \| P_{\vm^*}).
\end{align}
}
\end{enumerate}
\end{theorem}
\begin{remark}[Plug-In vs. Bayesian $Q^*$]{\rm 
We studied the sequential RIPr (Part 3 of the theorem) only in combination with a plug-in $Q^*$ and the global RIPr (Part 4)  only in combination with a Bayesian $Q^*$, but this was done for mathematical convenience only: one could in principle (though the analysis seems substantially more complicated in both cases) also study the RIPr for plug-in $Q^*$ or the sequential RIPr for Bayesian $Q^*$.  
}\end{remark}
\begin{remark}[Comparing e-power of the various methods]{\rm 
To determine e-power, we focus on the well-specified case again  with $R= Q_{\vm^*}$ so that $\Sigma_r = \Sigma_q$.
Theorem~\ref{thm:compositeH1gauss} then  implies:  
\begin{corollary}{\bf [e-power and growth optimality]}
\label{cor:GaussEpowerb} Under the conditions of Theorem~\ref{thm:compositeH1gauss}, ${\mathbb E}_Q[\log S_{\cond}^{(n)}] > n\epsilon $ for some $\epsilon > 0$ and all $n > 1$, and we have, 
    \begin{align}
&     {\mathbb E}_{Q_{\vm^*}} [
    \log S_{W_1,\ui}^{(n)}/ S_{\breve{\vec\mu},\ui}^{(n)}] = O(1),\ 
{\mathbb E}_{ Q_{\vm^*}} [
    \log S_{\cond}^{(n)}/ S_{\breve{\vec\mu},\ui}^{(n)}] = \frac{d}{2} \log n + O(1)
. \nonumber \\
&   {\mathbb E}_{Q_{\vm^*}} [
    \log S_{\breve{\vec\mu},\seqrip}^{(n)}/ S_{\breve{\vec\mu},\ui}^{(n)}] 
   = \frac{d_{qp}}{2} \log n + O(1) \text{\ with\ } d_{qp} \leq d, \text{\rm \ in simple case}. \label{eq:strict}
\\  \nonumber 
&   {\mathbb E}_{Q_{\vm^*}} [\log S_{\breve{\vec\mu},\seqrip}^{(n)}] = 0, \ 
{\mathbb E}_{Q_{\vm^*}} [
    \log S_{W_1,\rip}^{(n)}/ S_{\cond}^{(n)}] 
    = 0
    \text{\rm \ in anti-simple case},
    \nonumber
    \end{align}
where $d_{qp}$ is as in {(\ref{eq:eigenAardig})}, and we used the inequality in, (\ref{eq:dlite}). The inequality in (\ref{eq:strict}) becomes strict in the strict simple case.  Importantly, we only know how to explicitly calculate $S_{W_1,\rip}$ in the anti-simple case (then it is equal to $S_{\cond}$) and we never need to specify any prior $W_1$ or $\breve{\vm}$ when calculating $S_{\cond}$. 
\end{corollary}
One may now check that this is consistent with the results (\ref{eq:thefirstB})-(\ref{eq:thelast}) provided informally in the introduction, where $S_{\ui}$ may stand for both $S_{W_1,\ui}$ and $S_{\breve{\vm},\ui}$ and $S_{\rip}$ stands for $S_{W_1,\rip}$. ((\ref{eq:thefirst}) will be implied by Corollary~\ref{cor:generalEpowerb} later on).}
\end{remark}
\begin{remark}[The $O(1)$ terms]{\rm  Precise expressions for $O_{\textsc{a}}$, $O_{\textsc{b}}$, $O_{W_1,\ui}$ are as follows. \begin{equation}\label{eq:Oabc}
\begin{aligned}
&O_{\textsc{a}}( \| x_0 - \vm^* \|_2^2) = 
\sum_{i=0}^{n-1} \frac{1}{(1 + i/n_0 )^2} (\vec{\mu}^* - x_0)^\top \Sigma_\alti^{-1} (\vec{\mu}^* - x_0).\\
&O_{\textsc{b}} (1) =  - \log n + 
\sum_{i=1}^{n-1} \frac{i}{(n_0 + i)^2}.
\\ 
&O_{W_1,\ui}(1) 
= - \frac{d_{rp}}{2} - D_{\Sigma_r}\left(\Sigma_\alti|| \Sigma_\alti + n \Pi_1 \right)+    \frac{d}{2} \log \frac{n}{2 \pi} - \frac{1}{2}  O_{\textsc{c}}( \| \vm^* - \vm_1\|_2^2).
\\  
&O_{\textsc{c}}(\| \vm^* - \vm_1\|_2^2) =
(\vm^* - \vm_1)^\top (\Pi_1 + \Sigma_q/n)^{-1} (\vm^* - \vm_1).
\end{aligned}
\end{equation}

The $O_{\textsc{a}}$ and $O_{\textsc{c}}$  terms measure alignment between prior belief ($x_0$ or $\vm_1$) and true $\vm^*$, and become $0$ if the belief is correct. $O_{\textsc{b}}$ satisfies, for all $n$: $O_{\textsc{b}}(1) \leq \gamma + 1/(2n)$ (where $\gamma = 0.577 \ldots$ is Euler's constant). 
}\end{remark}\begin{remark}[Relation to optimal e-variables for group invariant testing]{\rm 
Perez et al. \cite{perez2024estatistics} showed that if $\althyp$ and $\nulhyp$ are both location families (not necessarily exponential families) then the Bayes factor obtained by equipping both models with the (improper) right Haar prior gives an e-variable; it is even the GROW (worst-case growth optimal) e-variable for the two models. The right Haar prior for a  location family is just the Lebesgue measure on $\vm$. 
We may apply this result to our Gaussian location families.
Using the fact \cite{HendriksenHG21} that the resulting  e-variable satisfies 
$S^{(n)}_{\haar}(U^1)=1$ and 
$$S^{(n)}_{\haar}(U^n)
= q_{W_1\mid X_1}(X_2, \ldots, X_n)/p_{W_0 \mid X_1}(X_2, \ldots, X_n),$$
where $W_1|X_1 = N(X_1, \Sigma_q)$ and $W_0|X_1 =N(X_1, \Sigma_p)$ are the formal Bayes posteriors (based on the right Haar prior), after observing $X_1$, we may employ the same techniques as used in proving Theorem~\ref{thm:compositeH1gauss} to prove (see Appendix~\ref{app:longproofssimpleH1gauss} for details) that 
\begin{align}\label{eq:haariscond}
{\mathbb E}_{R} [
    \log S_{\haar}^{(n)}
    ] =& (n -1) \cdot D_{  \Sigma_r}(\Sigma_q \|  \Sigma_p),
\end{align}
and thus has the same e-power as the conditional e-variable.
This e-power must then also be {\em worst-case\/} growth optimal (GROW) in the sense of \cite{GrunwaldHK19}, and by the results in that paper, we must have that $S_{\haar}^{(n)} = S_{\cond}^{(n)}$ are in fact {\em equal}.
}\end{remark}

\section{Multivariate exponential family}\label{sec:general}
We now extend our results to general exponential families.
We first formalize the simplicity condition which we mentioned in the introduction.
Then, in Section~\ref{sec:generalsimple}, we present Theorem~\ref{thm:simpleH1general}, the analogue to Theorem~\ref{thm:simpleH1gauss} (simple alternative) and then, in Section~\ref{sec:generalcomposite}, Theorem~\ref{thm:compositeH1general}, the analogue to Theorem~\ref{thm:compositeH1gauss}.

\subsection{The Simple and Anti-Simple Cases of Gr\"unwald et al. \cite{GrunwaldLHBJ24}}\label{subsec:Simple_and_Anti}
In some special situations which we shall collectively refer to as {\em The Simple Case}, $S_{Q,\rip}^{(1)}$ reduces to a {\em simple\/} likelihood ratio between $Q$ and a single element of $\nulhyp$. 
This is the upshot of the central result of \cite{GrunwaldLHBJ24}, which we now summarize.  \cite{GrunwaldLHBJ24} provides extensive discussion. 

Fix a regular exponential family null hypothesis $\nulhyp$ for outcome $U$ with sufficient statistic $X$ as in (\ref{eq:expfam}).   
Recall the canonical parameterizations of such a family: we may take any $\vm \in \meanspace_p$ and define
\begin{equation}
    \label{eq:canonical}
p^{\textsc{can}}_{\vb}(U) = \frac{1}{Z_p(\vb)} \cdot e^{\sum_{j=1}^d \beta_j t_j(U)} \cdot p_{\vm}(U),
\end{equation}
where $Z_p(\vb)$ is the normalizing constant and we define the canonical parameter space $\canspace_p = \{\vb: Z_p(\vb) < \infty \}$.  As is well known, the set of distributions  
$\{P^{\textsc{can}}_{\vb}: \vb \in \canspace_p \}$ where $P^{\textsc{can}}_{\vb}$ has density $p^{\textsc{can}}_{\vb}$, coincides with $\nulhyp$.

Next, let $Q$ be a distribution on $U$ under which $X$ has mean $\vm$. Suppose $X$ has a moment generating function under $Q$ and $Q$ is mutually absolutely continuous with $\nu$ so it has density relative to $\nu$. It will be convenient to denote this density by $q_{\vm}$, explicitly listing the mean $\vm$. 
 We may now {\em generate\/} a second exponential family $\genhyp$ with carrier density ${q_\vm}$ and the same sufficient statistic $X$ as $\nulhyp$. 
The elements of this family are defined like (\ref{eq:canonical}) but with $p_{\vm}$ replaced by $q_{\vm}$ and $Z_p$ by $Z_q$ and $\canspace_p$ by $\canspace_q$. 
We  adopt precisely the same notational conventions for $\genhyp$ as we did for $\althyp$, in particular $\meanspace^{\textsc{gen}}_q$ is its mean parameter space and $\Sigma^{\textsc{gen}}_{\alti}(\vm)$ is its $d \times d$ covariance matrix function.
The family $\genhyp$ thus created is often a natural choice to take as a composite alternative when the null is fixed to be $\nulhyp$, corresponding to testing a particular value of a natural notion of {\em effect size}, as illustrated in Section~\ref{sec:simulations}.  


Note that there exist a variety of canonical parameterizations of an exponential family, depending on which $p_{\vm}$ we take to be the carrier density in (\ref{eq:canonical}). For general regular exponential families on $U$ that share the same sufficient statistic, we say that canonical parameterizations of $\nulhyp$ and $\cQ$ {\em match\/} if $P^{\textsc{can}}_{\bf 0}$ and $Q^{\textsc{can}}_{\bf 0}$ are both well-defined and have the same mean, i.e. for some $\vm \in \meanspace_q \cap \meanspace_p$, we have $
{\mathbb E}_{P^{\textsc{can}}_{\bf 0}}[X]=  {\mathbb E}_{Q^{\textsc{can}}_{\bf 0}}[X]= \vec{\mu}.$
\begin{definition}\label{def:simple}{\rm Fix an exponential family $\nulhyp$ as above. 
\begin{enumerate}  \item  Let $\althyp$ be an exponential family for $U$. 
   We say that $\nulhyp$ and $\althyp$ are {\em matching pairs\/} if they are both regular and they have the same sufficient statistic $X$ and their elements are mutually absolutely continuous and we have that (a) $\meanspace_q \subseteq \meanspace_p$ and, (b), for every matching canonical parameterization of $\nulhyp$ and $\althyp$, we have that $\canspace_p \subseteq \canspace_q$. 
   \item Let $\althyp$ be a collection of distributions on $U$ so that each $Q \in \althyp$, together with sufficient statistic $X$, generates the same exponential family $\genhyp$. (i)
   We say that we are {\em in the simple case\/} if $\genhyp$ is a matching pair with $\nulhyp$ and, for all $\vm \in \meanspace^{\textsc{gen}}_q$,
    $\Sigma^{\textsc{gen}}_{\alti}(\vm) - \Sigma_{\nuli}(\vm)$ is negative semidefinite. We say that we are in the {\em strictly\/} simple case if we are in the simple case and, for all $\vm \in \meanspace^{\textsc{gen}}_q$, $\Sigma^{\textsc{gen}}_{\alti}(\vm)-  \Sigma_{\nuli}(\vm)$ is negative definite.
    (ii) We are in the $\vm$-anti-simple case if $\vm \in \meanspace^{\textsc{gen}}_q \cap \meanspace_p$ and    $\Sigma^{\textsc{gen}}_{\alti}(\vm) - \Sigma_{\nuli}(\vm)$ is positive semidefinite ($\nulhyp$ and $\genhyp$ are not required to be matching pairs). We are in the strict $\vm$-anti-simple case if $\vm \in \meanspace^{\textsc{gen}}_q$ and    $\Sigma^{\textsc{gen}}_{\alti}(\vm) - \Sigma_{\nuli}(\vm)$ is positive definite. 
    \end{enumerate}}
\end{definition}
We note that \cite{GrunwaldLHBJ24} only formalized the simple case; the $\vm$-{\em anti-simple\/} case is new.
In the following subsection we will be concerned with simple alternative $\althyp = \{Q \}$, and then $\genhyp$ is well-defined as soon as $X$ has a moment generating function under $Q$. Then, in Section~\ref{sec:generalcomposite}, we consider $\althyp$ that are themselves exponential families with sufficient statistic $X$. As is well-known, each member of a regular exponential family generates that same exponential family, so in that case we simply have $\althyp=\genhyp$. 
\citep{GrunwaldLHBJ24} proved the following result:

\ownparagraph{Theorem}[Corollary of Theorem 1 of \cite{GrunwaldLHBJ24}]
{\em 
Consider a testing problem with null $\nulhyp$ and alternative $\althyp$ and suppose we are in the simple case. Then for every  $\vm^*\in \meanspace_q$, 
the RIPr of $Q^{(n)}_{\vm^*}$ is given by $P^{(n)}_{\vm^*}$ so that the RIPr e-variable is given by
$$
  S_{Q_{\vm^*},\rip}^{(n)} = \frac{q_{\vm^*}(\sn{U}{n})}{
p_{\vm^*}(\sn{U}{n})}.
$$ \ \\
} 
\cite{GrunwaldLHBJ24} gives many examples of $\nulhyp$ and $\althyp$ that are `matching pairs' and that fall under the `simple case'. These include, for example, Bernoulli $k$-sample tests, the $k$-sample tests of Example~\ref{ex:twosample} below, and also a variation of linear regression with Gaussian noise. A prime example was already used implicitly in the previous section: if $\nulhyp$ and $\althyp$ are both  sets of multivariate Gaussians, we have $\meanspace_p=\meanspace_q= \canspace_p = \canspace_q= \reals^d$ so we have matching pairs. Since in this case, the covariance matrices $\Sigma_p$ and $\Sigma_q$ do not depend on $\vm$, we are in the simple case as soon as $\Sigma_q- \Sigma_p$ is negative semidefinite. 
\subsection{$\althyp$ simple, $\nulhyp$ multivariate exponential family}
\label{sec:generalsimple}
We now assume $\nulhyp$ to be a regular exponential family and consider simple alternative
$\althyp = \{Q \}$ with $Q$ a distribution on $U$
with density $q$, mean ${\mathbb E}_{Q}[X] = \vec{\mu}^*$ and covariance matrix $\Sigma_\alti$. Below we state Theorem~\ref{thm:simpleH1general} which extends Theorem~\ref{thm:simpleH1gauss} to this setting. The results for $\ui, \cond, \seqrip$ and $\rip$  will very closely follow those of that previous theorem, but, because of the added generality, regularity conditions are needed for some of them. We first discuss the most involved one, the 
\ccond\ 
condition:

We say that a $d$-dimensional random vector $X$ is {\em of lattice form\/} if there exist real numbers $b_1, \ldots, b_d$ and $h_1, \ldots, h_d$ such that, for all $j \in \{1, \ldots, d \}$, all $u \in \cU$, 
$X_j(u) \in \{b_j + s h_j  : s \in {\mathbb Z} \}$.
Obviously, most random variables with finite or countable support that are commonly encountered are of lattice form. Being either of this form or having a continuous density (with respect to Lebesgue measure) is the standard condition for the (multivariate) {\em local central limit theorem\/} 
\citep{BhattacharyaR76}
to hold, which is instrumental in the proof of (\ref{eq:gencond}) below. Concretely, we require the following:
\begin{condition} (\ccond)
    \label{cond:cond}
$X$ has a moment generating function under $Q$ and $R$ in Theorem~\ref{thm:simpleH1general} below. Moreover, either $X$ has a bounded continuous density (denoted $q^{[x]}$, $r^{[x]}$, $p_{\vm}^{[x]}$ respectively), with respect to Lebesgue measure under $Q$, $R$ and  all $\vm \in \meanspace_p$, or $X$ is of lattice form (and then has probability mass functions $q^{[x]}$, $r^{[x]}$, $p_{\vm}^{[x]}$ respectively).  The support ${\cal X} \subseteq \reals^d$ (${\cal X}$ is countable in the lattice case) under $Q$, $R$ and all $P \in \nulhyp$ coincides. 
Moreover, there is $A, a > 0$ such that
\begin{equation}\label{eq:thisisweak}
\sup_{s \in\reals: s \geq A} \frac{\sup \left\{ \log \frac{p^{[x]}_{\vm^*}(x)}{q^{[x]}(x)} : x \in {\cal X}, \|x- \vm^*\|_2 \leq s\right\}}{s^a} < \infty.
\end{equation}
\end{condition}
\noindent Since all our results are only relevant for the case that $D_R(Q \| P_{\vec{\mu}^*}) < \infty$ anyway, requirement (\ref{eq:thisisweak}) merely says that the likelihood ratios cannot be super-exponentially far apart, so in that sense it is quite weak; yet it requires $R$ to have exponentially small tails. From inspecting the proof it can be seen that if $X$ has just three moments under $R$, the result (\ref{eq:gencond}) still holds if $q$ and $p_{\vm^*}$ are only polynomially apart, i.e. if (\ref{eq:thisisweak}) holds with the logarithm removed from the equation; we have not bothered to formalize this.  The requirement that $X$ has a moment generating function under $Q$ (rather than $R$) is essential for the proof though. 

We further need a condition to prove a lower bound of the e-power of $S_{\ui}$:
\begin{condition}{\bf (UI$^{\geq}$)}
    \label{cond:uinew}
     $R$ as in Theorem~\ref{thm:simpleH1general} below is such that for some $0 < \gamma < 1/2$, for all $n$:
\begin{equation*}
 {\mathbb E}_R \left[{\bf 1}_{\| \hvm_{|n} - \vec{\mu}^* \|_2 \geq n^{- \gamma}}  \cdot   n \cdot D(P_{\hvm_{|n}} \| P_{\vec{\mu}^*} )\right]  = o(1).
\end{equation*}
\end{condition}
This condition refers to $D(P_{\hvm_{|n}} \| P_{\vm^*})$ also for cases in which $\hvm_{|n} \not \in \meanspace_p$; in Appendix~\ref{app:preparinggeneraltheoremsA} we extend the definition to that case. We suspect the condition is quite weak: it is readily verified for one-dimensional models such as e.g. Poisson, negative binomial, exponential and so on (we illustrate that it holds for the Poisson distribution, when $X$ has $\geq 3$ moments under $R$, 
in Appendix~\ref{app:checking}),
but, unlike for the `dual' Condition~\ref{cond:plugin} that we will discuss in the next subsection, we have not been able to come up  with a general easy-to-verify `$n$-free' condition that implies Condition~\ref{cond:uinew} for multivariate families. 

In this section, the covariance matrices in $\nulhyp$, $\althyp$ and $R$ are dependent on $\vm$. We thus extend our notation for the trace (\ref{eq:dlite}) to $d_{ab}(\vm) := \tr(\Sigma_a(\vm) \Sigma_b^{-1}(\vm))$. We also use (and will keep using in Theorem~\ref{thm:compositeH1general} and the proofs) the nonstandard notation 
$f(n) \leq O(1)$, to mean  that there is a positive constant $M > 0$ such that, for all $n \in \naturals$, we have $f(n) \leq M$; $f(n) \leq o(1)$ means that there is a sequence of positive numbers $M_1, M_2, \ldots$ tending to $0$ such that  for all $n \in \naturals$, we have $f(n) \leq M_n$.
\begin{theorem}\label{thm:simpleH1general}
Let $\mathcal{P}$ and  $\althyp= \{Q\}$ be null and alternative as in Section~\ref{sec:exponential}, where $Q$ has mean $\mathbb{E}_Q[X] = \vec{\mu}^*$ with $\vec{\mu}^* \in \meanspace_p$ and  covariance matrix $\Sigma_q$. Let $R$ be a distribution on $X$ with the same mean ${\mathbb E}_{R}[\g{X}] = \vm^*$ and with covariance matrix $\Sigma_r$, such that the first 3 moments of $X$ exist under $R$ and $Q$ and $D_R(Q \| P_{\vm^*})$ is well-defined and finite. Then, letting $U_1, U_2, \ldots$ be i.i.d. $\sim R$, we have:
\begin{enumerate}
    \item {\bf (UI)} Let $S_{Q,\ui}^{(n)} = \frac{q(\sn{U}{n})}{p_{\hvm_{|n}}(\sn{U}{n})}$ be as in (\ref{eq:uinewb}). 
    We have:
   \begin{equation}\label{eq:genui}
{\mathbb E}_{\sn{U}{n} \sim R} [
    \log S_{Q,\ui}^{(n)}
    ] \leq  n D_R(Q \| P_{\vec{\mu}^*}) - \frac{d_{rp}(\vm^*)}{2} + o(1).
   \end{equation}
If, moreover, $R$ is such that Condition~\ref{cond:uinew} (``{\bf UI}$^{\geq}$'') holds, 
then (\ref{eq:genui}) holds with $\leq$ replaced by $=$. 
\item {\bf (COND)} Let $S_{Q, \cond}^{(n)} = \frac{q(\sn{U}{n} \mid Z)}{
p(\sn{U}{n} \mid Z)}$ be defined as in (\ref{eq:Scond}) and $D_{\Sigma_r}(\cdot \| \cdot)$ as in (\ref{eq:triplegauss}). Suppose Condition~\ref{cond:cond} (``$\bf{COND}$'') holds. We have: 
\begin{align}\label{eq:gencond}
{\mathbb E}_{\sn{U}{n} \sim R} [
    \log S_{Q,\cond}^{(n)}
    ] \geq n D_R(Q \| P_{\vec{\mu}^*}) - D_{\Sigma_r}(\Sigma_q \| \Sigma_\nuli(\vec{\mu}^*))  + o(1).
   \end{align}
Moreover, if Condition~\ref{cond:cond} (``$\bf{COND}$'') also holds with the role of $q$ and $p_{\vm^*}$ interchanged, (\ref{eq:gencond}) holds with equality. 
\item {\bf (seq-RIPr/RIPr Simple Case, seq-RIPr Anti-Simple Case)}
Let $S_{Q,\rip}^{(n)}, S_{Q,\seqrip}^{(n)}$ be defined as in (\ref{eq:ripr}) and (\ref{eq:seqripnew}). Suppose that we are in the simple case of Definition~\ref{def:simple}. Then we have, for all $n$, that 
\begin{align}\label{eq:simple_E_power}
S_{Q,\rip}^{(n)} = S_{Q,\seqrip}^{(n)} =  \frac{q(\sn{U}{n})}{p_{\vec{\mu}^*}(\sn{U}{n})}
\ \ ; \ \ 
{\mathbb E}_{\sn{U}{n} \sim R} [
    \log S_{Q,\rip}^{(n)}
    ] = n D_R(Q \| P_{\vec{\mu}^*}).
   \end{align}
   If $\Sigma_q - \Sigma_{\nuli}(\vec{\mu}^*)$ is positive definite (`strict anti-simple case') then we have:
   \begin{align}\label{eq:noncompetitive} {\mathbb E}_{Q} [
    \log S_{Q,\seqrip}^{(n)}
    ] = n (D(Q||P_{\vec{\mu}^*})- \epsilon) \end{align}
where $0 < \epsilon = D(Q \| P_{\vec{\mu}^*}) - D(Q \| P_{\leftsquigarrow q(U)}) \leq D(Q \| P_{\vec{\mu}^*})$, where $\leq$ becomes an equality, and e-power (\ref{eq:simple_E_power}) becomes 0, if $\inf  D(Q \| P_W)= 0$, the infimum being over all priors $W$ on $\meanspace_p$. 
\item  {\bf (RIPr, Anti-Simple Case)} 
 Let $S_{Q,\rip}^{(n)}, S_{Q,\cond
}^{(n)}$ be defined as in (\ref{eq:ripr}) and (\ref{eq:Scond}) and $D_{\Sigma_q}(\cdot \| \cdot)$ as in (\ref{eq:triplegauss}). If $\Sigma_{\alti}- \Sigma_{\nuli}(\vec{\mu}^*)$ is {\em positive\/} semidefinite, then, letting $W_{0,(n)}=N(\vm^*,(\Sigma_q-\Sigma_p)/n)$ be a Gaussian distribution on $\meanspace_p$,  we have
\begin{align}\label{eq:gekb}
     {\mathbb E}_Q [
    \log S_{Q,\rip}^{(n)}
    ] \leq {\mathbb E}_Q\left[ \log \frac{q(U^n)}{p_{W_{0,(n)}}(U^n)} \right]  \leq  n D(Q||P_{\vec{\mu}^*}) - D_{\Sigma_q}(\Sigma_\alti \Sigma_\nuli^{-1}(\vec{\mu}^*)) + o(1).
    \end{align}
    Moreover, combining (\ref{eq:gekb}) with (\ref{eq:gencond}) shows that if Condition~\ref{cond:cond} (`\ccond') holds, then  (\ref{eq:gekb}) holds with equality up to $o(1)$ and also
    ${\mathbb E}_Q[\log S_{Q,\rip}^{(n)}] = {\mathbb E}_Q[\log S_{Q,\cond}^{(n)}]+ o(1)$.
\end{enumerate}
\end{theorem}
Note the very close resemblance between each part of this result and the corresponding part in Theorem~\ref{thm:simpleH1gauss}. 
We briefly remark on some specifics.

First, for the {\bf UI} part:  the case in which it holds with equality is related to the celebrated Wilks phenomenon: if we let $Q \in \nulhyp$ and set $R=Q$, the theorem is still valid and the problem becomes `well-specified'. The $D_R(Q \| P_{\vec{\mu}^*})$ term then becomes $0$ 
and we are left with $d/2$, the expectation of a $\chi^2$-random variable with $d/2$ degrees of freedom, in accordance with Wilks' result \cite{kotlowski2011maximum}.

For the anti-simple case, we only consider the well-specified setting with $R=Q$ because the statement for general $R$ is rather involved.  
Theorem~\ref{thm:simpleH1general} shows that in the anti-simple case, the  RIPr-prior that achieves GRO against $Q$ at sample size $n$ can be approximated, in the e-power sense, by a Gaussian with variance of order $O(1/n)$, but it does not tell us if the approximation is good enough to get something close to a real e-value if we use this prior in the denominator. 
Note that even though we do not know if the Gaussian prior gives us an e-variable, we can still find another e-variable that provably is close to GRO  -- the theorem shows that we can simply use $S_{Q,\cond}$. 
Performing a similar analysis as in Corollary~\ref{cor:GaussEpowera}, we find: 
\begin{corollary}{\bf [e-power and growth optimality]}
 \label{cor:generalEpowera} The relations of Corollary~\ref{cor:GaussEpowera} still hold in the setting of Theorem~\ref{thm:simpleH1general}, with $d_{qp}= d_{qp}(\vm^*)$ now dependent on $\vm^*$.
\end{corollary}
Theorem~\ref{thm:simpleH1general} provides a general asymptotic analysis, but we note that the Gaussian case of Theorem~\ref{thm:simpleH1gauss} is not the only case in which asymptotics are not  required. For example, another case which does not require asymptotics for all e-variables except $\ui$ occurs if there is $P \in \nulhyp$ whose marginal distribution $P(X)$ for $\g{X}$ coincides with $Q(X)$.
In case $\nulhyp$ and $\althyp$ are matching pairs we are then simultaneously in the (nonstrict) simple and anti-simple cases.  We then have
\begin{equation}
D_R(Q \| P_{\vm^*}) = {\mathbb E}_{X\sim R} \left[\log \frac{q(X)}{p_{\vm^*}(X)}
\right] +  {\mathbb E}_{U\sim R} \left[\log \frac{q(U|X)}{p_{\vm^*}(U|X)}\right] = {\mathbb E}_{U\sim R} \left[\log \frac{q(U|X)}{p_{\vm^*}(U|X)}\right]. \nonumber
\end{equation}
\begin{proposition}\label{prop:seqcondwins}
    Suppose that $Q(X) = P_{\vec{\mu}^*}(X)$ for some $\vec{\mu}^* \in \meanspacenul$. Then  
for all $n$, $S_{Q,\rip}^{(n)}=S_{Q,\seqrip}^{(n)}= 
S_{Q,\cond}^{(n)} = q(\sn{U}{n} \mid \sn{X}{n})/p(\sn{U}{n} \mid \sn{X}{n})$ so that 
    $$
       {\mathbb E}_{\sn{U}{n} \sim R} [
    \log S_{Q,\rip}^{(n)}] = {\mathbb E}_{\sn{U}{n} \sim R} [ S_{Q,\seqrip}^{(n)}] 
=  {\mathbb E}_{\sn{U}{n} \sim R} [
    \log S_{\cond}^{(n)}]  = n \cdot D_R(Q \| P_{\vm^*}).
    $$
\end{proposition}
\ownparagraph{Proof of Proposition~\ref{prop:seqcondwins}}
The results for $S_{Q,\rip}$ and $S_{Q,\seqrip}$ follow from Theorem~\ref{thm:simpleH1general}, the Simple Case, since $\Sigma_q =  \Sigma_\nuli(\vec{\mu}^*)$. This only leaves the proof for $S_{Q,\cond}$, which follows by inspecting the proof for Part 2 (`\ccond') of Theorem~\ref{thm:simpleH1general} and noting that the result directly follows from (\ref{eq:condy}).  
\begin{example}\label{ex:twosample}
    {\bf [Two-sample tests]} \rm  Let ${\cal P}^{\circ} = 
    \{P^{\circ}_{\mu}: \mu \in \meanspace^{\circ}\}$ be a regular 1-dimensional exponential family with sufficient statistic $Y$ and mean-value parameter space $\meanspace^{\circ}$. Set $U = (Y_a, Y_b)$, $X= Y_a+ Y_b$, and, for arbitrary $\mu_a,\mu_b \in \meanspace^{\circ}$, let $R_{\mu_a,\mu_b}$ be the distribution under which $Y_a \sim P^{\circ}_{\mu_a}$ and $Y_b\sim P^{\circ}_{\mu_b}$
    independently. Let $\nulhyp = \{R_{\mu,\mu}: \mu \in \meanspace^{\circ} \}$ stand for the null hypothesis, representing that $(Y_a,Y_b)$ are i.i.d. with the same mean $\mu$. Then (see e.g. \cite{HaoGLLA23} for details)  $\nulhyp$ is an exponential family with mean-value parameter space $\meanspacenul= 2 \meanspace^{\circ}$. Since  ${\mathbb E}_{R_{\mu,\mu}}[X] = 2 \mu$, the parameter $\vm^*$ in the mean-value parameterization corresponding to $\mu$, i.e. the $\vm^*$ for which $P_{\vm^*} = R_{\mu,\mu}$, must satisfy $\vm^* = 2\mu$. 
    Now, fix arbitrary $\mu_a, \mu_b \in \meanspace^{\circ}$ with $\mu_a \neq \mu_b$ and let $Q= \{R_{\mu_a,\mu_b}\}$ be the corresponding simple alternative. Consider the cases that (a) the underlying model is the Gaussian location family, i.e. $\meanspace^{\circ} = \reals$ and $P^{\circ}_{\mu} = N(\mu,\sigma^2)$ for some fixed $\sigma^2$, or (b), that the underlying model is Poisson, i.e.  $\meanspace^{\circ}= \reals^+$ and  $P^{\circ}_{\mu}$ is Poisson  with mean $\mu$.
\cite{HaoGLLA23} show that in both these cases (as well as their generalization from 2- to $k$-sample testing), the condition of Proposition~\ref{prop:seqcondwins} holds. Thus, in both the Gaussian location and the Poisson case, nonasymptotic results can be obtained. 
    They also looked at several other exponential families ${\cal P}^{\circ}$. In all such other cases, the premise of Proposition~\ref{prop:seqcondwins} does not apply and asymptotic analysis is needed. With the Bernoulli model,  we are in the simple case; with the exponential, beta and Gaussian scale families, we are not \cite{GrunwaldLHBJ24,HaoGLLA23}. We provide some simulations pertaining to Bernoulli and exponential distribution $2$-sample tests in Section~\ref{sec:simulations}.
\end{example}
\cite{GrunwaldLHBJ24} provides several more examples of exponential families that fall under the simple case. Intriguingly, this includes a variation of Gaussian linear regression which is different from the standard version involving t-statistics of \cite{LindonRegression22,perez2024estatistics}. We may thus apply Theorem~\ref{thm:simpleH1general} and~\ref{thm:compositeH1general} to this regression variation. In contrast to the earlier Gaussian location results and the  Gaussian- and Poisson $k$-sample tests of Example~\ref{ex:twosample}, we cannot avoid asymptotics here, due to the appearance of non-central $\chi^2$ distributions in the analysis \cite[Section 4.3.2 and 4.4]{GrunwaldLHBJ24}. 
\subsection{$\nulhyp$, $\althyp$ both multivariate exponential families}
\label{sec:generalcomposite}
In the section, we provide  similar results as above but with composite $\mathcal{Q}$. That is,
$
\mathcal{Q} = \{Q_\vec{\mu}: \vec{\mu} \in \mathtt{M}_q \}
\text{ and }
\mathcal{P} = \{P_\vec{\mu}: \vec{\mu} \in \mathtt{M}_p \},
$
where $\mathcal{Q}$ and $\mathcal{P}$ are 
matching pairs of exponential families as in Definition~\ref{def:simple}.
The results involving the prequential plug-in method now require the following condition, referring to the sampling distribution $R$ in Theorem~\ref{thm:simpleH1general} which has ${\mathbb E}_R[X] = \vm^* \in \meanspace_p$. 
\begin{condition}\label{cond:plugin}{\bf (Plug-in)}
There (a) exists an odd integer  $m \geq 3$ such that  the first $m$ moments of $X$ exist under $R$, and (b), for some $0 < \gamma < 1/2$,
\begin{align}
    \label{eq:mleplugin} & 
{\mathbb E}_R \left[{\bf 1}_{\| \hvm_{|n}- \vec{\mu}^* \|_2 \geq n^{-\gamma}}  \cdot  n \cdot  D( Q_{\vec{\mu}^*}\| Q_{\breve{\vec{\mu}}_{|n}})  \right] = o(1). 
\end{align} 
As we show in Lemma~\ref{lem:KLbounds}, (\ref{eq:mleoutrightB}), in Appendix~\ref{app:preparinggeneraltheoremsA}, a  sufficient condition for (\ref{eq:mleplugin}) is that (a) holds, and (c) 
there also exist real $s > 1$ and  $A > 0$ such that $\{\vm: \| \vm - \vm^* \|_2 < A\} \subset \meanspace_q$ and for  all $x_0 \in \meanspace_q, 0 < \alpha < 1$, 
\begin{equation}\label{eq:more}
\sup_{\ \vm \in {\meanspace}_q: \| \vm - \vm^* \|_2 \geq A} 
\frac{D(Q_{\vm^*} \| Q_{(1-\alpha) \vm+ \alpha x_0} )}{
\| \vm - \vm^* \|_2^{m-s}} < \infty.
\end{equation}
\end{condition}
As is well-known, {\em locally\/} both $D(Q_{\vm} \| Q_{\vm^*})$ and  $D(Q_{\vm^*} \| Q_{\vm})$  are equal, up to lower order terms, to the quadratic form $(\vm^* - \vm)^\top \Sigma^{-1}_q (\vm^* - \vm) \asymp \| \vm^* - \vm \|_2^2$. 
The {\bf UI}$^{\geq}$ Condition~\ref{cond:uinew}   implies that, up to $o(1)$, the KL divergence $D(Q_{\vm} \| Q_{\vm^*})$, for fixed $\vm^*$ and varying $\vm$, is determined by this local quadratic behaviour.
The {\bf plug-in} Condition~\ref{cond:plugin} implies the same for   $D(Q_{\vm^*} \| Q_{\vm})$. As can be seen, as  $\vm$ moves farther away from $\vm^*$,  this allows for a faster-than-quadratic increase of the KL divergences, but it should be at most polynomial in the $\ell_2$-distance, and the rate of increase should be `matched' by sufficiently high moments of $R$ existing. 
The plug-in condition appears to be quite weak and is easily verified with appropriate values for $m$ (details omitted) for 1-dimensional exponential families such as the Bernoulli, Poisson, geometric, binomial and negative binomial model. Below, and in the appendix, we verify it for some standard 2-dimensional models. Nevertheless, it does not {\em always\/} hold: for the 1-dimensional regular family $\althyp$ with $X=U \in \meanspace_q =\reals$ generated by the {\em Landau distribution\/} \citep{Barlev23}, we find that  $D(Q_{\vm^*} \| Q_{\vm^{\circ}})$ grows exponentially in $\vm^{\circ}$ as $\vm^{\circ} \rightarrow \infty$, thereby violating the condition. 
The appearance of the `anchor' $x_0$ derives from the existence of a similar regularization term in the definition of $\breve\vm$. Note that the condition would in some cases not hold without it; for example, in the Bernoulli model, 
$\sup_{\vm\in (0,1)} D(P_{\vm^*} \| P_{\vm}) = D(P_{\vm^*} \| P_1) = \infty$. 
\commentout{reason why in general cannot take the `raw' KL divergence $D(P_{\vm^*} \| P_{\vm})$  but have to modify it slightly can be seen from the Bernoulli example:
\begin{example}{\bf [Gaussian Location-Scale]}
{\rm 
Let $\nulhyp$ be the Bernoulli model. Then $\meanspace_p = (0,1)$ and $\bar\meanspace_p = [0,1]$. Then for any $\vm^* \in \meanspace_p$, $\sup_{\vm \in [0,1]} 
D(P_{\vm} \| P_{\vm^*})$ is bounded, implying that
$\sup_{\ \vm \in \bar{\meanspace}_p, \vm \neq \vm^* } 
\frac{
D(P_{\vm} \| P_{\vm^*})}{
\| \vm - \vm^* \|_2^{m-s}}$ is finite, in this case even for {\em every\/} $m\geq 3$ and $s> 1$. Since $D(P_{\vm^*} \| 0) = \infty$, we do not have the same statement with $P_{\vm}$ and $P_{\vm^*}$ exchanged. By mixing in a point $x_0$ in the interior (recall $\meanspace_p$ is open),  the supremum becomes finite again, as is easily shown, so that the condition holds.
}
\end{example}}

\begin{example}\label{example:Gaussian-Location-Scale}{\bf [Gaussian location-scale]}
    {\rm Let $\althyp$ be the 2-dimensional family of Gaussian distributions $N(\mu,\sigma^2)$. In terms of the mean-value parameterization, we get 
    $\vm = (\mu_1,\mu_2)$ with $\mu_1 = \mu$ and $\mu_2 = \mu^2+ \sigma^2$ and $\meanspace_q = \{ (\mu_1, \mu_2): \mu_1 \in \reals, \mu_2 > \mu_1^2\}$ so that $Q_{\vm}$ with $\vm =(\mu_1,\mu_2)$ represents $N(\mu_1, \mu_2- \mu_1^2)$.
A standard calculation gives, with  $\vm^* = (\mu^*_1, \mu^{*2}_1 + \sigma^{*2})$, 
and $\vm^{\circ} = (\mu_1^{\circ}, \mu_2^{\circ})$, that
\begin{align}\label{eq:KL_gaussian}
D(Q_{\vm^*} \| Q_{\vm^{\circ}}) = 
\frac{1}{2} \left( 
\log \frac{\mu^{\circ}_2 - \mu^{\circ 2}_1}{\sigma^{*2}} +\frac{\sigma^{*2}+ (\mu^*_1 - \mu^{\circ}_1)^2}{\mu^{\circ}_2 - \mu_1^{\circ 2}} -1
\right).
\end{align}
Now let $\vm^{\circ} = (1-\alpha) \vm + \alpha x_0$. In Appendix~\ref{proof:guassian_example}, we show (this involves some work)
that for every  $0 < \alpha < 1$, every  
$x_0 = (x_{0,1},x_{0,2}) \in \meanspace_q$,  
we have 
\begin{align}\label{eq:inf_variance}
\inf_{\vm \in \meanspace_q} (\mu^{\circ}_2 - \mu_1^{\circ 2}) = \alpha (x_{0,2} - x_{0,1}^2).
\end{align}
with $\vm^{\circ} = (\mu^{\circ}_1,\mu^{\circ}_2)$ a function of $\vm$ and 
$x_{0,2} > x_{0,1}^2$ (so (\ref{eq:inf_variance}) $> 0$) because $x_0 \in \meanspace_q$.
 This implies that the second term in (\ref{eq:KL_gaussian}) is $O(\sigma^{*2} + (\mu_1^* - \mu_1^{\circ})^2)$.
As we also show in the appendix, this then easily implies that for every $\vm^* \in \meanspace_q$, every $A > 0$, we have (\ref{eq:more}).
with $m-s=2$, verifying Condition~\ref{cond:plugin}  as soon as $R$ has $5$ or more moments.}
\end{example}\noindent
With similar arguments, one can show (\ref{eq:more}) holds with $m-s=2$ for $\althyp$ the family of Gamma distributions; details are in Appendix~\ref{app:checking}.



\begin{theorem}\label{thm:compositeH1general}
Let $\nulhyp$ and $\althyp$ be two regular exponential families that are matching pairs in the sense of Definition~\ref{def:simple} and such that $\nulhyp \cap \althyp = \emptyset$ and $\meanspace_q \subseteq \meanspace_p$, as above. Consider a distribution $R$ as above with mean ${\mathbb E}_R[\g{X}] = \vec{\mu}^* \in \mathtt{M}_\alti \subseteq \mathtt{M}_\nuli$ and with covariance matrix $\Sigma_r$, such that the first 3 moments of $X$ exist under $R$ and $D_R(Q \|P)$ is finite for all $Q \in \althyp, P \in \nulhyp$. Let $U_1, U_2, \ldots$ be i.i.d. $\sim R$. 
We have: 
\begin{enumerate}
    \item {\bf (UI)}
    Let $S_{\breve{\vm},\ui}^{(n)}
    $ be
    as in (\ref{eq:sgauss}). 
Suppose  Condition~\ref{cond:plugin} (``{\bf plug-in}'') holds. We have:
\begin{align}\label{eq:seq_UI}
  {\mathbb E}_{R} [
    \log S_{\breve{\vec\mu},\ui}^{(n)}
    ] \leq  n D_R(Q_{\vec\mu^*} \| P_{\vec{\mu}^*})
    - \frac{d_{rq}(\vm^*)}{2} \log n +O(1).
   \end{align}
Moreover, suppose that $R$ has density $r$, that $W_1$ is a prior on $\meanspace_q$ with continuous and strictly positive density in a neighborhood of $\vm^*$, and such that for some $\epsilon > 0$, we have \\ ${\mathbb E}_{R} |\log p_{\vm^*}(U)/r(U)|^{1+\epsilon} < \infty$. Then: 
   \begin{equation}\label{eq:uibayes}
{\mathbb E}_{R} [
    \log S_{W_1,\ui}^{(n)}
    ] \leq   n D_R(Q_{\vec\mu^*} \| P_{\vec{\mu}^*}) - \frac{d}{2} \log \frac{n}{2 \pi} + O_{W_1,\ui}(1),
   \end{equation}
where $O_{W_1,\ui}(1)$ was  given up to $o(1)$ in (\ref{eq:OWui}). 
If moreover Condition~\ref{cond:uinew} (``{\bf UI}$^{\geq}$'') holds for $R$ then (\ref{eq:seq_UI}) and (\ref{eq:uibayes}) hold with equality. 
\item {\bf (COND)} Let $S_{\cond}^{(n)}
$ be
as in (\ref{eq:Scond}) and $D_{\Sigma_r}(\cdot \| \cdot)$ be as in (\ref{eq:triplegauss}). Suppose Condition~\ref{cond:cond} (``$\bf{COND}$'') holds for $R$, $Q= Q_{\vm^*}$ and $P_{\vm^*}$. We have: 
\begin{align}\label{eq:Composite_gencond}
{\mathbb E}_{R} [
    \log S_{\cond}^{(n)}
    ] \geq  n D_R(Q_{\vec{\mu}^*} \| P_{\vec{\mu}^*}) - D_{\Sigma_r}(\Sigma_q(\vm^*) \| \Sigma_\nuli(\vm^*) ) + o(1).
   \end{align}
Moreover, if Condition~\ref{cond:cond} holds with the role of $q$ and $p_{\vm^*}$ interchanged, (\ref{eq:Composite_gencond}) holds with equality. 
\item {\bf (seq-RIPr, Simple and Anti-Simple Case)} Suppose 
Condition~\ref{cond:plugin} (``{\bf plug-in}'') holds. 
Let $S_{\breve{\vec\mu}, \seqrip}^{(n)}$ be as in (\ref{eq:sgauss}). If we are in the {\em simple case\/} of Definition~\ref{def:simple}, 
then we have: 
$
\Pseqrip
= P_{\breve{\vec{\mu}}_{|i-1}}(U_i)$ and
\begin{align}\label{eq:seq_LVA}
{\mathbb E}_{R} [
    \log S_{\breve{\vec\mu},\seqrip}^{(n)}
    ] 
    = n D_R(Q_{\vec\mu^*} \| P_{\vec{\mu}^*}) + \frac{d_{rp}(\vm^*) - d_{rq}(\vm^*)}{2} \log n + O_{\seqrip}(1)+ o(1),
\end{align}
   where $O_{\seqrip}(1)$ is as  below  (\ref{eq:Gaussseq_LVA}). On the other hand, if $\Sigma_q(\vm^*) - \Sigma_p(\vm^*)$ is positive definite (`strict anti-simple case') then for some $\epsilon > 0$, for all $n$,
\begin{equation}\label{eq:seqbad}
{\mathbb E}_{Q} [
    \log S_{\breve{\vec\mu},\seqrip}^{(n)}
    ] 
    \leq   n (D(Q_{\vec\mu^*} \| P_{\vec{\mu}^*})- \epsilon ).
\end{equation}
   \item {\bf(RIPr, General Case)} 
   Let $S_{W_1,\rip}^{(n)}$ be defined as in (\ref{eq:ripr}) with $W_1$ a prior with support contained in $\meanspace_p \cap \meanspace_q$ and with positive continuous density in a neighborhood of $\vm^*$, and let  $D_{\Sigma_r}(\cdot \| \cdot)$ be as in (\ref{eq:triplegauss}). We have 
\begin{align}\label{eq:gekc}
     {\mathbb E}_{Q} [
    \log S_{W_1,\rip}^{(n)}
    ] & \leq
   {\mathbb E}_Q\left[\log \frac{q_{W_1}(U^n)}{p_{W_1}(U^n)} \right] =
    n D(Q||P_{\vec{\mu}^*}) - D_{\Sigma_\alti(\vm^*) }(\Sigma_\alti(\vm^*) \| \Sigma_\nuli (\vm^*)) + o(1),
    \end{align}
and if additionally Condition~\ref{cond:cond} (`\ccond') holds, then the inequality becomes an equality.

   \end{enumerate}
\end{theorem}
Note again the close similarity to the corresponding  Theorem~\ref{thm:compositeH1gauss}. 
The theorem implies the following corollary, which, like Corollary~\ref{cor:GaussEpowerb}, commensurates with (\ref{eq:thefirst})-(\ref{eq:thelast}) in the introduction.
\begin{corollary}{\bf [e-power and growth optimality]}
Suppose Condition~\ref{cond:cond} (`\ccond') and~\ref{cond:plugin} ({\bf plug-in}) holds for $R=Q= Q_{\vm^*}$, $\nulhyp$ and $\althyp$. Then for some $\epsilon >0$, it holds  ${\mathbb E}_{Q} [
    \log S_{\cond}^{(n)}]  \geq n  \epsilon - o(1)$, and  
\label{cor:generalEpowerb}
    \begin{align}
&   {\mathbb E}_{ Q_{\vm^*}} [
    \log S_{\cond}^{(n)}/ S_{\breve{\vec\mu},\ui}^{(n)}] \overset{(a)}{\geq} \frac{d}{2} \log n + O(1).  \nonumber \\
    &  {\mathbb E}_{Q} [
    \log S_{\rip}^{(n)}/ S_{\cond}^{(n)}] = o(1). \nonumber \\
&   {\mathbb E}_{Q_{\vm^*}} [
    \log S_{\breve{\vec\mu},\seqrip}^{(n)}/ S_{\breve{\vec\mu},\ui}^{(n)}] 
   \overset{(a)}{\geq} \frac{d_{qp}(\vec{\mu}^*)}{2} \log n + O(1) \text{\ with\ } d_{qp} \leq d, \text{\rm \ in the simple case}. \nonumber\\ 
     &  {\mathbb E}_{Q} [
    \log S_{\seqrip}^{(n)}/ S_{\cond}^{(n)}]  \leq - n \epsilon + o(1)   \text{\rm \ for some $\epsilon > 0$, in the strict anti-simple case}. 
      \nonumber
    \end{align}
where (a) becomes an equality if Condition~\ref{cond:uinew} holds, and $d_{qp}(\vec{\mu}^*)$ is as in (\ref{eq:dlite}), 
with $d_{qp}(\vm^*)  \leq  d$ in the simple case (with strict inequality in the strictly simple case) following  in the same way as in (\ref{eq:dlite}).
\end{corollary}
In the next section, we illustrate the theorem with some simulations pertaining to the two-sample setting of Example~\ref{ex:twosample}. 
\section{Examples and Simulations}\label{sec:simulations}
To illustrate our asymptotic results we complement them with some finite-sample size simulations. For simplicity, we restrict ourselves to two simple 2-sample tests as in Example~\ref{ex:twosample}, which, however, are quite different in that they fall under the simple- and anti-simple case, respectively. Thus, with details as in Example~\ref{ex:twosample}, let ${\cal P}^{\circ}$ be a 1-dimensional exponential family,
set  $U = (Y_a, Y_b)$, $X= Y_a+ Y_b$, and let $R_{\mu_a,\mu_b}$ express that $Y_a \sim P^{\circ}_{\mu_a}$ and $Y_b\sim P^{\circ}_{\mu_b}$ independently such that  $\nulhyp = \{R_{\mu,\mu}: \mu \in \meanspace^{\circ} \}$ stands for the null hypothesis, representing that $(Y_a,Y_b)$ are i.i.d. with the same mean $\mu$. 

\subsection{The Simple Case: Bernoulli 2-Sample Test}
In this case ${\cal P}^{\circ}$ is the Bernoulli model, so that, recast in the mean-value parameterization, the null can be written as $\nulhyp = \{P_{\vm}: \vm \in \meanspacenul\}$ where $\meanspacenul= (0,2)$, and $P_{\vm}$ expresses that $Y_a, Y_b$ are independent Bernoulli, both with mean $\vm/2$. The results of \cite{HaoGLLA23} imply that this is an instance of the {\em simple\/} case for ${\cal Q}$ as defined below. We start with point alternative $ \{Q \}$ with $Q=R_{\mu_a,\mu_b}$. Under $Q$, $Y_a$ and $Y_b$ are independent Bernoulli, with mean $\mu_a$ and $\mu_b$ respectively. 
Note that $2 \vm^* = {\mathbb E}_Q[X] = \mu_a + \mu_b$. We  use $Q$  to  generate a full alternative $\althyp:= \althyp^{\textsc{gen}}$ with the same sufficient statistic  $X=Y_a+Y_b$ as $\nulhyp$, as in Section~\ref{subsec:Simple_and_Anti}.
Let $\delta(\mu_a,\mu_b) := \log ({\mu_a(1- \mu_b)})/((1-\mu_a) \mu_b)$ be the {\em log odds ratio\/} corresponding to $(\mu_a,\mu_b)$. 
As shown by 
\cite[Section 4.3]{GrunwaldLHBJ24}, the generated alternative $\althyp^{\textsc{gen}}$ will correspond to the set of all $(\mu'_a,\mu'_b)$ that have the same {\em log odds ratio\/} $\delta^* := \delta(\mu_a,\mu_b)$ as the original $Q$, i.e. all $(\mu'_a,\mu'_b)\in (0,1)^2$ such that
$\delta(\mu'_a,\mu'_b) = \delta^*$. 
It is easy to see that for any $(\mu_a,\mu_b)$ with $\delta(\mu_a,\mu_b)= \delta^*$, we can write $\althyp^{\textsc{gen}} = \{Q_{\mu'_a,\mu'_b}: (\mu'_a,\mu'_b) \in \mathtt{M}_{\delta^*}\}$ with 
\begin{equation}\label{eq:H1_para_space}
\mathtt{M}_{\delta^*} = \left\{\left(\frac{\mu_a e^\beta}{1-\mu_a+\mu_a e^\beta}, \frac{\mu_b e^\beta}{1-\mu_b+\mu_b e^\beta} \right)^\top: \beta \in \mathbb{R} \right\}.\end{equation}
Figure~\ref{fig:parameter-space} illustrates various \( \mathtt{M}_q^{\textsc{gen}} \) with different \( \delta^* \) values. Note that when \( \delta^*= 0 \), \( \mathtt{M}_{\delta^*} = \mathtt{M}_p \).

The log odds ratio has long been a standard notion of effect size in two-sample tests, and this approach implicitly adopts it. Of course, in practice we would often want to test the larger alternative hypothesis that $\delta \geq \delta^*$. Clearly all types of e-variables we develop in this paper can be used in that setting as well, since the validity of e-values is independent of the actually chosen alternative; we suspect that they still have close to optimal e-power in a worst-case sense (the GROW sense of \cite{GrunwaldHK19}) if used to test $\delta \geq \delta^*$ and we plan to investigate whether this is really so  in future work.  For now, we concentrate on testing $\delta=0$ against alternative $\delta^*$. 

Our simulations supplement Theorem~\ref{thm:compositeH1general}, exploring e-powers for all e-variables appearing in that result, and, for reasons to become clear, also for a ``pseudo'' e-variable. That is, we consider
\begin{equation}\label{notation-e}
    S_{W_1,\rip}^{(n)},\quad
    S_{W_1,\ui}^{(n)},\quad 
    S_{\breve{\vm},\ui}^{(n)},\quad
    S_{\breve{\vec\mu},\seqrip}^{(n)},\quad
    S_{\cond}^{(n)},\quad
    \frac{q_{W_1}(U^n)}{p_{W_1}(U^n)}.
\end{equation}
For the simulations, first we fix $\delta^*$, then we employ \( W_1 \) as a uniform prior over the corresponding \( \mathtt{M}_{\delta^*} \). 
Specifically, for the given $\delta^*$, we draw \(\beta\) uniformly from the interval \([-10, 10]\) 
and use it to determine the corresponding parameters \((\mu_a', \mu_b')^\top \in \mathtt{M}_q^{\textsc{gen}}\) as described by (\ref{eq:H1_para_space}), where we set $(\mu_a,\mu_b)$ such that $\delta^*(\mu_a,\mu_b) = \delta^*$ and $\mu_a = 1- \mu_b$. 
For computing \( S_{\breve{\vm},\ui}^{(n)} \) and \( S_{\breve{\vec\mu},\seqrip}^{(n)} \), we set \(\breve{\vm}_{|n} = {(1 + \sum_{i=1}^n x_i)}/(n+1)\).
A distribution \( Q^* \in \mathcal{Q}^{\textsc{gen}}  \), corresponding to a specific point in Figure~\ref{fig:parameter-space} for the selected $\delta^*$, is selected from \( \mathcal{Q}^{\textsc{gen}} \) to generate i.i.d. data sequences \( U^n \). Note that, turning to the mean-value parameterizaiton for \( \mathcal{Q}^{\textsc{gen}} \), we can write $Q= Q_{\vm^*}$ where $\vm^*= \mu_a + \mu_b$ for the specific $(\mu_a,\mu_b)$ with $Q= R_{\mu_a,\mu_b}$.  For each sequence, we compute the logarithm of the quantities in (\ref{notation-e}). Repeating this process under \( Q^* \), we then take the empirical mean of these logarithms to approximate the e-powers under \( Q^* \). Our figures depict the quantity $nD(Q_{\vm^*}|| P_{\vm^*}) - \mathbb{E}_{Q^*}[\log S^{(n)}]$, which may be thought of as the loss in e-power of a tester employing $S^{(n)}$ compared to an oracle tester who has access to $Q_{\vm^*}$ and can use the GRO optimal e-variable relative to $Q_{\vm^*}$. We perform simulations for \( \delta^* = 1 \) and \( \delta^* = 2 \), selecting three points for each \( \delta^* \) to generate data sequences, as well as a single point for the more extreme $\delta^* = 2 \log 19 \approx 5.9$ corresponding to $(\mu_a,\mu_b) = (0.95,0.05)$. Two instances are presented in Figure~\ref{fig:simulation_results_delta1} and~\ref{fig:extreme_param}; other instances, all exhibiting the same qualitative behavior, are in Figure~\ref{fig:simulation_results_delta2} in Appendix~\ref{app:Simulations}.


Theorem~\ref{thm:compositeH1general} indicates that the e-powers of the optimal $S_{W_1,\rip}^{(n)}$ and of \( S_{\cond}^{(n)} \) and \( \frac{q_{W_1}(U^n)}{p_{W_1}(U^n)} \) become identical (hence optimal) as $n$ increases. Our simulations indicate that converges already happens at quite small $n$. It turns out that calculating $S_{W_1,\rip}^{(n)}$ directly (via numerical optimization of the KL divergence) is highly computationally intensive for large $\delta^*$; we could not do it for $\delta^*\approx 5.9$ in Figure~\ref{fig:extreme_param}. But, interestingly, and as our main take-away from this section, the experiments indicate that this need not bother us: we know from Theorem~\ref{thm:compositeH1general} that the e-power of the valid and optimal $S_{W_1,\rip}^{(n)}$ must lie inbetween that of the (valid) e-variable $S_{\cond}^{(n)}$ and the (pseudo-) e-variable \( \frac{q_{W_1}(U^n)}{p_{W_1}(U^n)} \), i.e. it must be squeezed in between the small space between the  red and blue lines in 
Figure~\ref{fig:extreme_param}. 
Consequently, we may use $S_{\cond}^{(n)}$ -- which is much easier to calculate -- as a safe substitute of $S_{W_1,\rip}^{(n)}$ while sacrificing essentially negligible e-power. We also note that, as a basic calculation shows, $d_{qp}(\vm^*) = (\mu_a (1-\mu_a) + \mu_b (1-\mu_b))/(2 \bar\mu (1- \bar\mu)$, where $\bar\mu = (\mu_a+\mu_b)/2$ which implies via (\ref{eq:Composite_gencond}) that the red curve corresponding to $\Scond$ in Figure~\ref{fig:extreme_param} should converge to $D_{\Sigma_r}(\Sigma_q(\vm^*) \| \Sigma_\nuli(\vm^*) ) = D_{\gauss}(d_{qp}) \approx 0.425$; we see that this convergence takes place already at quite small $n$.  

Turning to $S_{\breve{\vec\mu},\seqrip}^{(n)}$, we see from  (\ref{eq:seq_LVA}) that its e-power will be asymptotically smaller than that of the optimal  \( S_{W_1,\rip}^{(n)} \) by a logarithmic amount $- c \log n + O(1)$, with
\begin{equation}\label{eq:drp-drq}
c= \frac{d_{qp}(\vm^*) - d}{2} = -\frac{(\mu_a - \mu_b)^2}{2(\mu_a + \mu_b)(2 - (\mu_a + \mu_b))} \leq 0,
\end{equation}
which follows using $d=1$ and the above expression of $d_{qp}$.
(\ref{eq:drp-drq}) is close to 0 for all  \( (\mu_a, \mu_b)^\top \in (0, 1)^2 \) with $\mu_a \neq \mu_b$ except for values close to the upper-left and lower-right corner in Figure~\ref{fig:parameter-space}. This is nicely in accord with our simulations: Figure~\ref{fig:simulation_results_delta1} and Figure~\ref{fig:simulation_results_delta2} indicate that the e-power of \( S_{\breve{\vec\mu},\seqrip}^{(n)} \) is just  slightly lower than that of the optimal \( S_{W_1,\rip}^{(n)} \) for $\delta^* \in \{1,2\}$, but if we sample from  \((\mu_a,\mu_b) = (0.95, 0.05)^\top\) with corresponding $\delta^* \approx 5.9$, then $c$ in (\ref{eq:drp-drq}) is close to $1/2$, and indeed, as Figure~\ref{fig:extreme_param} demonstrates,  the e-power loss of \(S_{\breve{\vec\mu},\seqrip}^{(n)}\) starts behaving like the sub-optimal  \( S_{\breve{\vec\mu},\ui}^{(n)} \) and \( S_{W_1,\ui}^{(n)} \), which are also suboptimal by  a logarithmic amount, with constant factor $c= d/2= 1/2$. 

\begin{figure}
    \centering
    \includegraphics[width=0.5\linewidth]{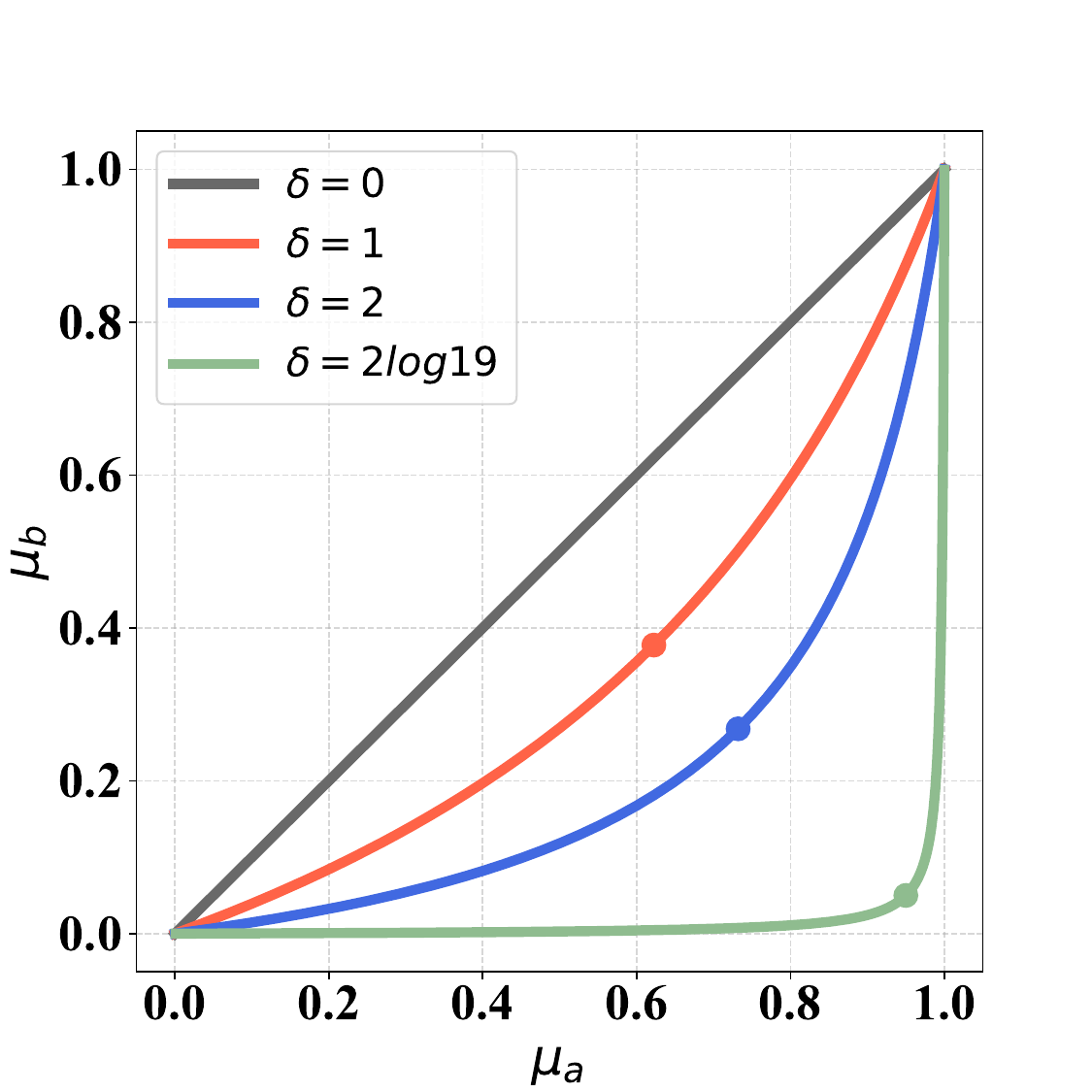}
    \caption{Parameter space visualization under the Bernoulli model, with starting points indicated in the figure for generating \(\mathtt{M}_q^{\textsc{gen}}\).}
    \label{fig:parameter-space}
    \centering
    \includegraphics[width=0.7\linewidth]{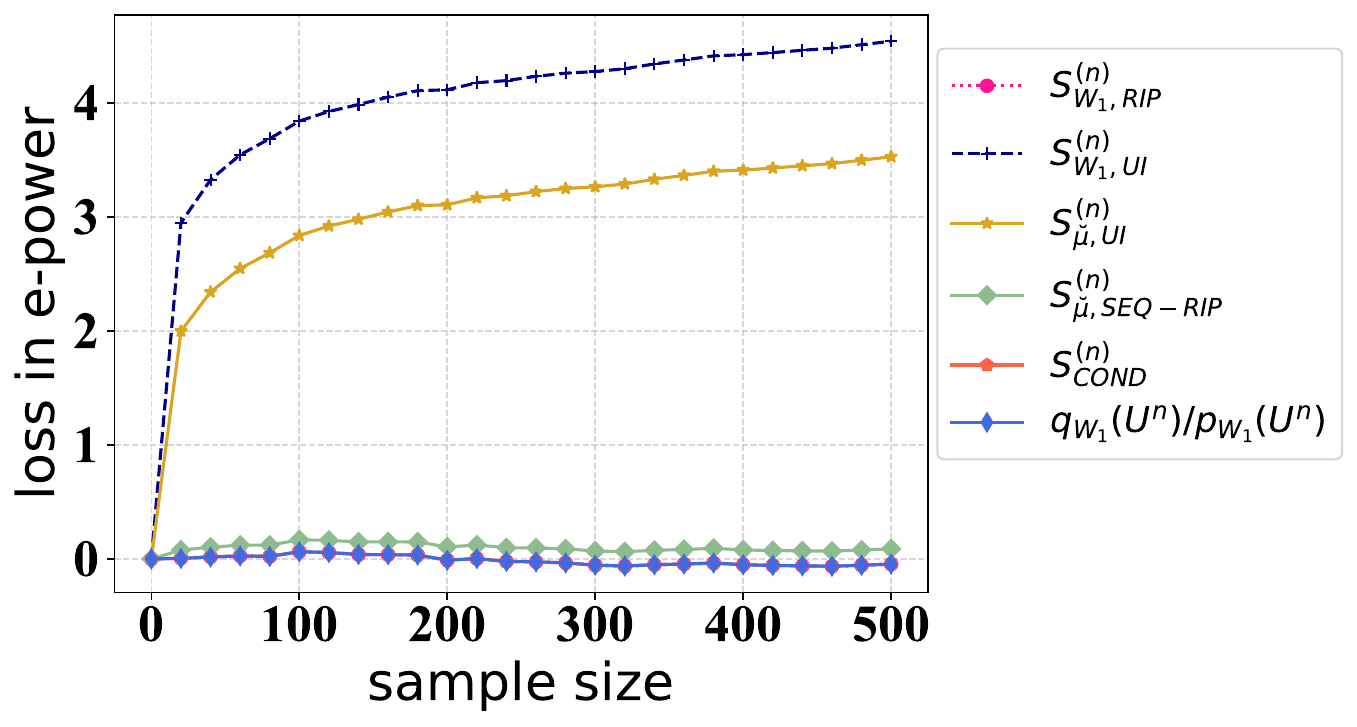}
    \caption{$\delta^*= 1$. E-powers are computed for multiple (pseudo) e-variables for the 2-sample Bernoulli test. The data sequences are generated from $Q^*= Q_{\vm^*}= R_{\mu_a,\mu_b}$  with $\mu_a= 2/7\approx 0.286, \mu_b = 2/(5e +2)\approx 0.128$ with mean $\vm^*= \mu_a + \mu_b \approx 0.41$ and $ D(Q^*|| P_{\vec{\mu}^*}) \approx  0.0385$.}
    \label{fig:simulation_results_delta1}
    \centering
    \includegraphics[width=0.7\linewidth]{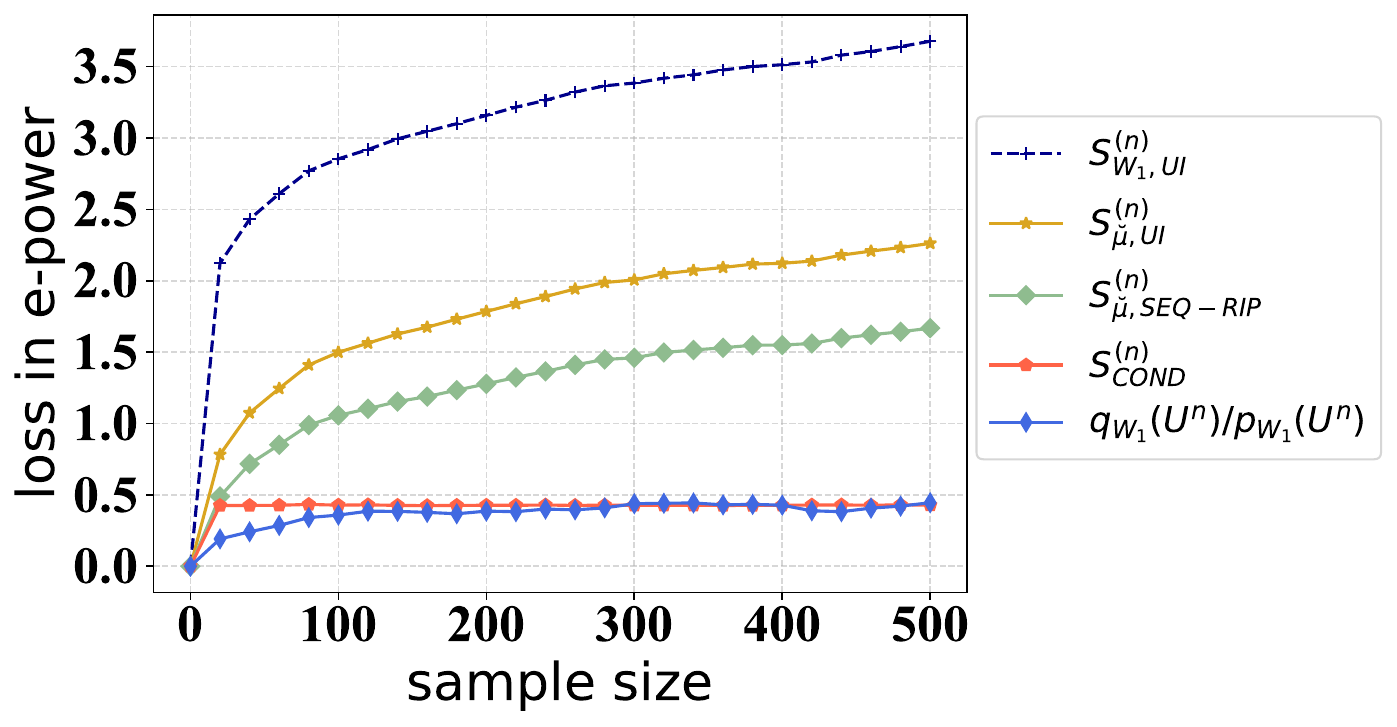}
    \caption{$\delta^* = 2\log 19\approx 5.9$. E-powers are computed for multiple (pseudo) e-variables for the 2-sample Bernoulli test. The data sequences are generated from $Q^*= Q_{\vm^*}
    R_{\mu_a,\mu_b}$ with $(\mu_a,\mu_b)= (0.95,0.05)$ so $\vm^*=1$, $P_{\vm^*} = R_{0.5,05}$ and $D(Q^*|| P_{\vec{\mu}^*}) \approx 0.9893$.}
    \label{fig:extreme_param}
\end{figure}

    

\subsection{The Anti-Simple Case: Exponential 2-Sample Test}
Consider the setting and notations described in  the beginning of Section~\ref{sec:simulations} but now let ${\cal P}^{\circ}$ be the family of exponential distributions rather than Bernoulli. 
Hao et al.~\cite{HaoGLLA23} implicitly demonstrate that, for arbitrary $\mu_a, \mu_b > 0$ with $\mu_a \neq \mu_b$ and $\vm^* = \mu_a+\mu_b$, the RIPr of \(Q_{\vm^*} = R_{\mu_a,\mu_b} \) on \(\mathcal{P} = \{P_{\mu}: \mu \in (0,\infty) \}\), the model under which $(Y_a,Y_b)$ are i.i.d. exponentially distributed with arbitrary mean $\mu$, is {\em not\/}  \(P_{\vm^*}(X)\). This means that for every $\vm^* \in \meanspace_q$, we are now in the $\vm^*$-strict-anti-simple case. 

We investigate the same  (pseudo-) e-variables (\ref{notation-e}) as before, while, as above in Figure~\ref{fig:extreme_param}, excluding \(S_{W_1,\rip}^{(n)}\) for computational reasons. This is of no great concern since, as stated earlier, 
Theorem~\ref{thm:compositeH1general} shows  that, under \(Q\), the asymptotic e-power of \(S_{W_1,\rip}^{(n)}\) lies between that of \(S_{\cond}^{(n)}\) and \(\frac{q_{W_1}(U^n)}{p_{W_1}(U^n)}\), both of which are easily computed. Our simulation results (we include just one plot but tried several more) invariably show that \(\mathbb{E}_Q\left[\log S_{\cond}\right]\) closely approximates \(\mathbb{E}_Q\left[\log \frac{q_{W_1}(U^n)}{p_{W_1}(U^n)}\right]\) even at quite small sample sizes (in Figure~\ref{fig:expo_simulation} they are indistinguishable). 
Thus, our first take-away from the experiments is that $S_{\cond}$ serves as a practical and computationally efficient alternative of $S_{W_1,\rip}$. 

To describe the simulations in more detail, we start by setting  $Q= R_{\mu_a,\mu_b}$ for some $(\mu_a,\mu_b) > 0$ and calculate the corresponding 
$\althyp:= \althyp^{\textsc{gen}}$ with the same sufficient statistic  $X=Y_a+Y_b$ as $\nulhyp$. Defining \(\delta(\mu_a,\mu_b) := \frac{1}{\mu_a} - \frac{1}{\mu_b}\), it turns out that $\althyp^{\textsc{gen}}$ is the set of all distributions $R_{\mu'_a,\mu'_b}$ with $\delta(\mu'_a,\mu'_b)= \delta(\mu_a,\mu_b)$ for the $(\mu_a,\mu_b)$ we started with. Thus $\delta$ is our new notion of effect size. 
Following a similar calculation as for the Bernoulli case, 
%
the mean parameter space of \(\mathcal{Q}^{\textsc{gen}}\) turns out to be: 
\begin{equation}\label{eq:H1_para_space_Expo}
    \mathtt{M}_q^{\textsc{gen}} = \left\{ \left(\frac{\mu_a}{1 - \mu_a \beta}, \frac{\mu_b}{1 - \mu_b \beta}\right)^\top : \beta \in \mathbb{S} \right\},
\end{equation}
which, for fixedd $\delta^*$, is the same for every $(\mu_a,\mu_b)$ with $\delta(\mu_a,\mu_b) = \delta^*$. 
Here \(\mathbb{S}\) is the set of \(\beta\) such that both \(\frac{\mu_1}{1 - \mu_1 \beta}\) and \(\frac{\mu_2}{1 - \mu_2 \beta}\) are positive. 

The simulation setup mirrors the Bernoulli case, with the exception that we now fix \(\delta^* = -1\). We continue to use \(W_1\) as a uniform prior over \(\mathtt{M}_q^{\textsc{gen}}\). Specifically, we fix a particular $(\mu_a,\mu_b)$ with $\delta(\mu_a,\mu_b)= \delta^*$ (here we arbitrarily took $\mu_a = 2, \mu_b = 2/3$), and we sample \(\beta\) uniformly from $\mathbb{S}\cap [-10,10]$ and compute the corresponding parameters \((\mu_a', \mu_b')^\top \in \mathtt{M}_q^{\textsc{gen}}\) as described by (\ref{eq:H1_para_space_Expo}). 
The simulation results are presented in Figure~\ref{fig:expo_simulation} and are consistent with Theorem~\ref{thm:compositeH1general}. It is seen that the analysis very closely resembles that of the Bernoulli model. The one difference worth noting relates to \(S_{\breve{\vec{\mu}},\seqrip}^{(n)}\). Since the 2-sample test under an exponential model falls into the anti-simple category, the e-power of \(S_{\breve{\vec{\mu}},\seqrip}^{(n)}\) should align with equation~(\ref{eq:seqbad}). Specifically, we expect that \(n D(Q_{\vm^*} \| P_{\vm^*}) - \mathbb{E}_{Q^*}[\log S_{\breve{\vec{\mu}},\seqrip}^{(n)}] \geq n \epsilon\) for some \(\epsilon > 0\), which is indeed witnessed by  Figure~\ref{fig:expo_simulation}. Since calculating  \(S_{\breve{\vec{\mu}},\seqrip}^{(n)}\) is computationally very intensive, we had to limit the computatin to a maximum sample size of \(n=200\). Nonetheless, the results clearly indicate that the e-power of \(S_{\breve{\vec{\mu}},\seqrip}^{(n)}\) will be lower than that of \(S_{W_1,\ui}^{(n)}\) and \(S_{\breve{\vm},\ui}^{(n)}\) for sufficiently large \(n\) --- even though it is interesting to see that there is a large initial regime in which $S_{\breve{\vec{\mu}},\seqrip}^{(n)}$ outperforms universal inference. 

\begin{figure}[h]
    \centering
    \includegraphics[width=0.7\linewidth]{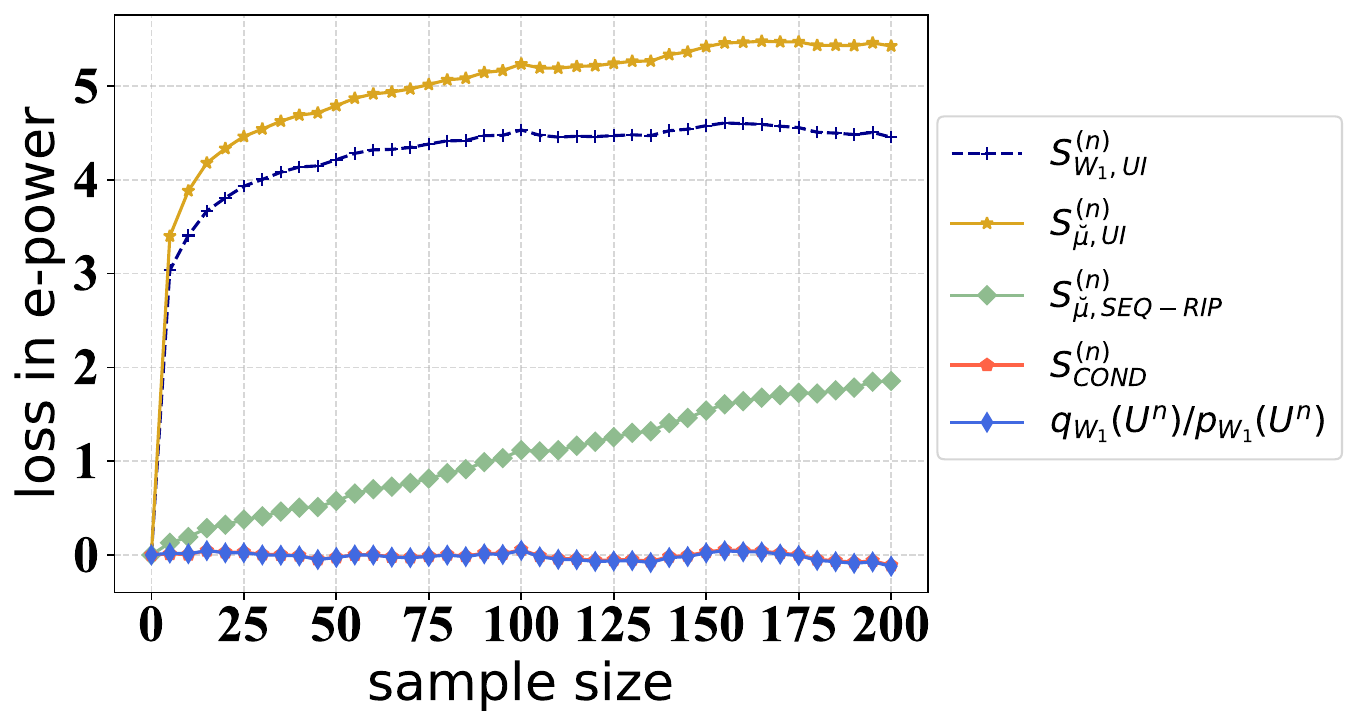}
    \caption{$\delta = -1$. E-powers are computed for multiple (pseudo) e-variables under exponential model. The data sequences are generated from $Q^*$ with mean $\mu^*$. $D(Q^*|| P_{\vec{\mu}^*}) \approx 0.2877$ when $\vec{\mu}^* = (2, 2/3)$.}
    \label{fig:expo_simulation}
\end{figure}

\section{Proofs for Section~\ref{sec:gauss}}\label{sec:gaussproofs}

\subsection{Preparation for both proofs}
We gradually build up a series of remarkable properties of KL divergences between multivariate Gaussians. These will then readily lead to the proofs of Theorem~\ref{thm:simpleH1gauss} and~\ref{thm:compositeH1gauss}.

Let $P, Q, R$ and $R_{\gaussvar}$ be distributions on $ X=U \in \reals^d$.  Suppose that, under $R$, ${\mathbb E}_{R}[X] = \vm^*$ and the nondegenerate covariance matrix is $\Sigma_r$, that $R_{\gaussvar}$ is Gaussian with the same mean and covariance, and suppose that $Q$  and $P$ are Gaussian with ${\mathbb E}_{Q}[X] = {\mathbb E}_{P}[X] = \vm^*$ and nondegenerate covariance matrices  $\Sigma_q$ and $\Sigma_p$ respectively. In Section~\ref{sec:gausscomposite} we introduced the notation $D_{  \Sigma_r}(\Sigma_q \| \Sigma_p)$. The following equation implies that $D_{  \Sigma_r}(\Sigma_q \| \Sigma_p)= D_R(Q \| P)$, as claimed in (\ref{eq:forevergauss}); even more generally, if $P$ is replaced by $P_0$ with mean $\vm_0$ and still covariance matrix $\Sigma_p$,  we obtain a generalization of the well-known formula for the KL divergence between two multivariate normals that do not necessarily have the same mean: 
\begin{align}\label{eq:normalKL}
    D_{R}(Q \| P_0) & = 
{\bf E}_R\left[\log \frac{q(X)}{p_{0}(X)}\right] =
{\bf E}_{R_{\gaussvar}} \left[\log \frac{q(X)}{p_{0}(X)}\right] 
    \\ & =
    D_{\gauss}(\Sigma_r \Sigma_p^{-1}) -D_{\gauss}(\Sigma_r \Sigma_q^{-1}) + 
    \frac{1}{2} (\vm^* - \vm_0)^\top \Sigma^{-1}_p (\vm^* - \vm_0) \nonumber \\ &= 
    D_{ \Sigma_r}(\Sigma_q \| \Sigma_p) + 
    \frac{1}{2} (\vm^* - \vm_0)^\top \Sigma^{-1}_p (\vm^* - \vm_0), \nonumber
\end{align}
with $D_{\gauss}$ as in (\ref{eq:gaussrules}) and $D_{\Sigma_r}$ as in (\ref{eq:triplegauss}). Here the first equality is definition, the second follows by writing out the definitions,
and in the third we used the standard (easily derived) formula for the KL divergence between two multivariate Gaussians.
\commentout{\begin{align}\label{eq:Rgauss}
    D_{ \Sigma_r}(\Sigma_q \| \Sigma_p) & =   
    D_{\gauss}(\Sigma_r \Sigma_p^{-1}) -D_{\gauss}(\Sigma_r \Sigma_q^{-1}) \nonumber \\ & =
    - \log \det (\Sigma_q \Sigma_p^{-1}) + \frac{1}{2} 
    \tr(\Sigma_r (\Sigma_p^{-1} - \Sigma_q^{-1})). 
\end{align}}

The following is a standard result (see for example \cite[Chapter 2.3.]{bishop2006pattern}):
\begin{proposition}\label{prop:bishop}
The marginal distribution of $\hvm_{|n} \in \reals^d$ under $P_{W}$, with $W=N(\vm^*,\Pi)$ is given by:
\begin{equation}\label{eq:marginalmle}
\hvm_{|n} \sim N(\vec{\mu}^*, \Pi + \Sigma_\nuli/n).
\end{equation}
\end{proposition}
Proposition \ref{prop:bishop} will be used in several in the proofs of results; directly for $S_{\seqrip}$ and indirectly via its use in the proof of the following lemma, which will itself be used in the proofs for $S_{\cond}$ and $S_{\rip}$ in both Theorem \ref{thm:simpleH1gauss} and \ref{thm:compositeH1gauss}.

\begin{lemma}\label{lem:referee}
 Let $\nulhyp$ and $\althyp$ be two $d$-dimensional Gaussian location families for i.i.d. data $U, U_1, U_2, \ldots$ as in Section~\ref{sec:gausssimple}, with nondegenerate covariance matrices $\Sigma_p$ and $\Sigma_q$. 
Let $R$ be a distribution as above, with ${\mathbb E}_R[X] = \vm^*$, extended to  $U_1, U_2, \ldots$
by independence. Let $W_0$ and $W_1$ be arbitrary prior distributions on $\meanspace_p = \meanspace_q= \reals^d$ (which may be degenerate, i.e. put all their mass on a single point). We have:
\begin{align}
&    \label{eq:condrules}
S_{\cond} = \frac{q_{W_1}(U^n | Z)}{p_{W_0}(U^n | Z)}
= \frac{q_{\vm^*}(U^n | Z)}{p_{\vm^*}(U^n | Z)} \text{\ and\ } \\
\label{eq:preconditionalb}
&    D_R(Q_{W_1}^{(n)} | Z \; \| \; P_{W_0}^{(n)}| Z ) = 
    D_R(Q_{\vm^*}^{(n)} | Z \; \| \; P_{\vm^*}^{(n)}| Z ) = 
    (n-1) D_R(Q_{\vm^*}\| P_{\vm^*}).
\end{align}
Further, let either (i)  $W_1$ put mass 1 on $\{\vm^*\}$ and $W_0= N(\vm^*,\Pi_0)$ with $\Pi_0$ such that 
\begin{equation}\label{eq:magica}
    \Pi_0 + \Sigma_p/n =  \Sigma_q/n,
\end{equation}
or (ii) $W_1= N(\vm^*,\Pi_1)$ with $\Pi_1$ some nondegenerate covariance matrix and $W_0= N(\vm^*,\Pi_0)$ with $\Pi_0$ such that
\begin{equation}\label{eq:magicb}
    \Pi_0 + \Sigma_p/n = \Pi_1 + \Sigma_q/n.
\end{equation}
In both case (i) and (ii) we have:
\begin{equation}
    \label{eq:condmagic}
    \frac{q_{W_1}(U^n)}{p_{W_0}(U^n)} = 
    \frac{q_{W_1}(U^n \mid Z) q_{W_1}^{\circ}(Z)}{p_{W_0}(U^n \mid Z){p_{W_0}^{\circ}(Z)}} = S_{Q,\cond}^{(n)} = S_{Q,\rip}^{(n)}.
 \end{equation}
where $q^{\circ}_{W_1}$ denotes the density of $Z$ under $Q_{W_1}$, and similarly for $p^{\circ}_{W_0}$, and we have $q^{\circ}_{W_0}=p^{\circ}_{W_0}$.
\end{lemma}
\begin{proof}
We have, by the fact that scaled sums of  Gaussians are Gaussians, with $Z= n^{-1/2} \sum_{i=1}^n X_i = n^{1/2} \cdot \hvm_{|n}$,
that
\begin{equation*}
D_R(Q_{\vm^*}(Z) \| P_{\vm^*}(Z)) = D_R(Q_{\vm^*} \| P_{\vm^*}).
\end{equation*}
(\ref{eq:condrules}) now follows  by sufficiency.
Further, abbreviating $Q(U^n)$ to $Q^{(n)}$ and similarly for $P^{(n)}$,
\begin{equation}\label{eq:preconditional}
n D_R(Q_{\vm^*}\| P_{\vm^*})= D_R(Q^{(n)}_{\vm^*} \| P_{\vm^*}^{(n)})  =
D_R(Q^{(n)}_{\vm^*} | Z \;  \| \; P_{\vm^*}^{(n)} | Z) + 
D_R(Q_{\vm^*}(Z) \| P_{\vm^*}(Z)).
\end{equation} 
Combining (\ref{eq:condrules}) and (\ref{eq:preconditional}) 
now gives (\ref{eq:preconditionalb}).

Next,  (\ref{eq:marginalmle}) implies that if (i) holds, 
then we get by (\ref{eq:marginalmle}) that $\hvm_{|n}$ and therefore $Z$ has the same distribution under $q_{W_1}$ as under $p_{W_0}$  and then the first two equations of  (\ref{eq:condmagic}) follow by (\ref{eq:condrules}); the same reasoning holds under (ii).

To show the final equality of (\ref{eq:condmagic}), note that,
for both choices of $W_1$, if a corresponding prior satisfying (\ref{eq:magica}) or, respectively, (\ref{eq:magicb}) exists, the Bayes factor $q_{W_1}(U^n)/p_{W_0}(U^n)$ is equal to $S_{\cond}$ and hence becomes an e-variable. 
From Theorem 1 (specifically Corollary 1) of \cite{GrunwaldHK19} we know that for every $W_1$ on $\meanspace_q$, there can be at most one corresponding prior $W_0$ such that this Bayes factor is an e-variable, and if this prior exists, this e-variable is the GRO e-variable relative to $Q$; the equality follows. 
\end{proof}

We need one more proposition before giving the actual proof of Theorem~\ref{thm:simpleH1gauss} and~\ref{thm:compositeH1gauss}. 
It provides variations of standard results, 
whose proofs we delegate to Appendix~\ref{app:longproofssimpleH1gauss}  --- note the intriguing similarity between (\ref{eq:MLfullsample}) and (\ref{eq:prepreq}) if $\breve{\vm}_{|n}$ is chosen equal to $\hvm_{|n}$; the reader may recognize these as variations of the asymptotic equalities used in the derivation of Akaike's classic AIC criterion. 
\begin{proposition}\label{prop:basicgauss}
Let $\breve\mu$ be as before with $n_0 \geq 0$, and let  $W_1 = N(\vm_1,\Pi_1)$ be a multivariate normal on $\reals^d$. We have, with $O_{\textsc{a}}$--$O_{\textsc{c}}$ as in (\ref{eq:Oabc}):
\begin{multline}
  {\mathbb E}_{U, U^n \sim R} \left[ \log \frac{q_{\vm^*}(U)}{q_{\breve{\vm}_{|n}}(U)} \right] =  {\mathbb E}_{U^n \sim R} \left[D(
Q_{\vm^*} \| Q_{\breve{\vm}_{|n}}) \right] 
   = \label{eq:newproof}\\
   \frac{1}{2} \left(
 \frac{1}{(1 + (n/n_0))^2} (\vec{\mu}^* - x_0)^\top \Sigma_\alti^{-1} (\vec{\mu}^* - x_0) 
+ \left(\frac{n}{(n_0 + n)^2} \right) d_{rq} \right)
= \frac{1}{n} \frac{d_{rq}}{2} + O \left( \frac{1}{n^2}\right),
\end{multline}
where the second equality holds if $n_0 > 0$. If $n_0 = 0$ (i.e. $\breve{\vm}_{|n}$ is chosen equal to $\hvm_{|n}$) and $n > 1$ we get:
\begin{align}\label{eq:prepreq}
{\mathbb E}_{U, U^n \sim R} \left[n \cdot \log \left(
q_{\vm^*}(U)/ q_{\hat{\vm}_{|n}}(U) 
\right) \right]  
= {\mathbb E}_{U^n \sim R} \left[n \cdot D(
Q_{\vm^*} \| Q_{\hat{\vm}_{|n}}) \right] = \frac{d_{rq}}{2}.
\end{align}
We further have:
    \begin{align}
\label{eq:preqfullsample}
& {\mathbb E}_{U^n \sim R} \left[\log \frac{q_{\vm^*}(U^n)}{ \prod_{i=1}^n q_{\breve{\vm}_{|i-1}}(U_i)}\right] = \sum_{i=1}^n
{\mathbb E}_{U^{i-1} \sim R} \left[D(
Q_{\vm^*} \| Q_{\breve{\vm}_{|i-1}}) \right] 
=
 \\ \nonumber &  \frac{d_{rq}}{2} \log n + \frac{1}{2} O_{\textsc{a}}(\| x_0 - \vm^* \|_2^2)+ O_{\textsc{b}}(1) \cdot  \frac{d_{rq}}{2}.
\\ \label{eq:MLfullsample}
& {\mathbb E}_{U^n \sim R} \left[\log  \frac{q_{\hvm_{|n}}(U^n)}{q_{\vm^*}(U^n)}\right] = {\mathbb E}_{U^n \sim R} \left[n \cdot D(
 Q_{\hvm_{|n}} \| Q_{\vm^*}) \right]
= \frac{d_{rq}}{2}.
 \\ \label{eq:gaussredundancy}
& {\mathbb E}_{U^n \sim R} \left[\log 
\frac{q_{\vm^*}(U^n)}{q_{W_1}(U^n)} \right]  = 
D_{ \Sigma_r}(\Sigma_q \|  (n \Pi_1 + \Sigma_q)  ) +
\frac{1}{2} O_{\textsc{c}}(\| \vm^* - \vm_1\|_2^2)  = 
\\
& \frac{d}{2} \log \frac{ n}{2 \pi} - \frac{1}{2} \log \det \Sigma_q - \log  w_1(\vm^*) - \frac{d_{rq}}{2} + o(1). \nonumber
\end{align}
\end{proposition}

\subsection{Proof of Theorem~\ref{thm:simpleH1gauss}}
\noindent 
Finiteness of $D_R(Q \| P)$ is immediate from evaluating the definition. 

\ownparagraph{\bf UI}
(\ref{eq:gaussui}) follows almost immediately from Proposition~\ref{prop:basicgauss}, (\ref{eq:MLfullsample}) applied with model $\nulhyp$ in the role of $\althyp$, i.e. all $q$'s replaced by $p$'s 
and using that
${\mathbb E}_R[\log ( q_{\vm^*}(U^n)/ p_{\vm^*}(U^n))] = n D_R(Q_{\vm^*} \| P_{\vm^*})$.

\ownparagraph{\bf COND}
(\ref{eq:gausscond}) is a direct consequence of (\ref{eq:condrules}) and (\ref{eq:preconditionalb}) in Lemma~\ref{lem:referee}. 

\ownparagraph{\bf seq-RIPr}
Assume first that $\Sigma_q - \Sigma_p$ is negative semidefinite.
According to Theorem 1 in \cite{GrunwaldLHBJ24}, this implies that the RIPr of \(Q(U)\) (i.e. $Q$ restricted to a single outcome) is equal to \(P_{\vec{\mu}^*}(U)\), i.e. the single element of $\nulhyp$ with the same mean as $Q$. This implies that $S_{Q,\seqrip}$ is as in (\ref{eq:raar}). Thus, $q(U^n)/p_{\vm^*}(U^n)$ is an e-variable. The corollary of Theorem~1 of \cite{GrunwaldHK19} (`there can be only one e-variable with $Q$ in the numerator and an element $P'$ of $\nulhyp$ in the denominator, and if it exists  then $P'$ is the RIPr so $q/p$ is growth-rate optimal') implies that $P_{\vm^*}(U^n)$ is the RIPr of $Q(U^n)$, and hence $S_{Q,\seqrip}^{(n)}= S_{Q,\rip}^{(n)}$.

Now assume $\Sigma_q - \Sigma_p$ is positive semidefinite. $S_{Q,\seqrip}^{(n)}=1$ then follows because in this case, using the result for the Bayes marginal distribution (\ref{eq:marginalmle}) applied with $n=1$, we find that there exists a prior $W_0$ such that $p_{W_0}(X) = q(X)$, namely, $W_0=N(\vm^*,(\Sigma_q - \Sigma_p))$. It follows that 
the RIPr of $Q$ onto $\nulhyp$ for a single outcome $U$ must be given by this very prior, i.e. $P_{\leftsquigarrow q(U)}(U)= P_{W_0}(U)$, and then the corresponding ratio is $1$.

\ownparagraph{\bf RIPR, Anti-Simple Case}
If we set 
$ 
\Pi_0 = \frac{\Sigma_q - \Sigma_p}{n}
$
then by positive semidefinitess of $\Sigma_q - \Sigma_p$, the prior $W_0=N(\vm^*,\Pi_0)$ is well-defined and (\ref{eq:magica}) holds and we can conclude   (\ref{eq:condmagic}) with prior $W_1$ that puts all mass on $\vm^*$ (i.e. on $Q$).  (\ref{eq:condmagic}) implies the  result (\ref{eq:ripisgauss}) and the fact that $P_{W_1}$ is the RIPr. 
\subsection{Proof of Theorem~\ref{thm:compositeH1gauss}}
Finiteness of $D_R(Q \| P)$ is immediate from evaluating the definition. 

Concerning {\bf UI},
(\ref{eq:seq_UI_gauss}) 
can be obtained  from simple algebra from Proposition~\ref{prop:basicgauss}, combined with (\ref{eq:MLfullsample}) and (\ref{eq:preqfullsample}), the former applied with model $\nulhyp$ in the role of $\althyp$, i.e. all $q$'s replaced by $p$'s (evidently it is still valid then), and using that ${\mathbb E}_R[\log ( q_{\vm^*}(U^n)/ p_{\vm^*}(U^n))] =  n D_R(Q_{\vm^*} \| P_{\vm^*})$. (\ref{eq:seq_UI_gaussb})  can be obtained similarly from  (\ref{eq:MLfullsample}) and (\ref{eq:gaussredundancy}). 
As to {\bf COND},
(\ref{eq:condB}) is a direct consequence of Lemma~\ref{lem:referee}, (\ref{eq:preconditionalb}) and (\ref{eq:condrules}).
To prove Part 3, {\bf seq-RIPr},
we `sequentialize' the reasoning of Theorem~\ref{thm:simpleH1gauss}, Part 3: 
first consider the case that $\Sigma_q - \Sigma_p$ is negative semidefinite.
According to Theorem 1 in \cite{GrunwaldLHBJ24}, this implies that the RIPr of 
$Q_{\breve{\vm}_{|i-1}}(U_i)$ onto $\textsc{conv}(\nulhyp(U_i))$
is an element of $\nulhyp$. More specifically, the theorem says it is 
$P_{\breve{\vm}_{|i-1}}(U_i)$, 
i.e. the element of $\nulhyp$ with the same mean. This proves the first part of the result. (\ref{eq:Gaussseq_LVA}) then follows by applying  (\ref{eq:preqfullsample}) in Proposition~\ref{prop:basicgauss} twice, first directly, with elements of $\althyp$, and then again with $\nulhyp$ in the role of $\althyp$, i.e. replacing all $q$'s by $p$'s, and then piecing together both equations by simple algebra. 

Now consider the positive semidefinite case. The fact that $S_{\breve{\vec\mu},\seqrip}^{(n)}
= 1$ now follows because, at each $i$, 
using the result for the Bayes marginal distribution (\ref{eq:marginalmle}) applied with $n=1$, we find that 
with the prior $W_0 = N(\breve{\vec{\mu}}_{|i-1},\Sigma_{\alti}- \Sigma_{\nuli})$, we get $q_{\breve{\vm}}(U_i)= p_{W_0}(U_i)$. It follows that 
the sequential RIPr $\Pseqrip$ must be given by this very prior, i.e. $\Pseqrip= P_{W_0}(U_i)$, and then the corresponding ratio is $1$. Since this holds for each $i$, the result follows. 

Finally, to prove the anti-simple case for {\bf RIPr},
all results follow  analogously to the RIPr Anti-Simple Case of Theorem~\ref{thm:simpleH1gauss}.
In particular, (\ref{eq:W0depends}) and (\ref{eq:ripisgaussb})
follow by the same reasoning, but now using case (ii) insteaed of (i) in Lemma~\ref{lem:referee}. 

\section{Proofs for Section~\ref{sec:general}}
\label{sec:generalproofs}
\subsection{Preparation for both proofs}
We start with Proposition~\ref{prop:basicgeneral}, a direct asymptotic analogue of Proposition~\ref{prop:basicgauss},  for general exponential families. It will play an  analogous role in the proofs. 
The one (important!) difference to Proposition~\ref{prop:basicgauss} is that now  $W_1$ is not required to be Gaussian. Note again the  symmetry, now between (\ref{eq:prepreqg}) and (\ref{eq:MLfullsampleg}).
\begin{proposition}
\label{prop:basicgeneral}
Let $\althyp$ be a regular exponential family as in Section~\ref{sec:general}. Extend the definition of $D(P_{\hat{\vm}_{|n}} \| P_{\vm^*})$ as in (\ref{eq:fullrobustness}), Appendix~\ref{app:preparinggeneraltheoremsA}, so that it is defined on the whole co-domain of $\hat{\vm}_{|n}$. 
Further, suppose that the first $m$ moments of $X$ under $R$ exist,  with $m \geq 3$ odd, and ${\mathbb E}_R[X] = \vm^* \in \meanspace_q$. Let 
$W_1$ be a fixed distribution on $\meanspace_q$ with density $w$ that is strictly positive and continuous in a neighborhood of $\vm^*$. If Condition~\ref{cond:plugin}  (`{\bf plug-in}') holds, then we have: 
\begin{align}
&  {\mathbb E}_{U^n \sim R} \left[D(
Q_{\vm^*} \| Q_{\breve{\vm}_{|n}}) \right] 
= \frac{1}{n} \frac{d_{rq}}{2} + O \left( {n^{-5/4}}\right) \text{\ so that}\nonumber \\
& \label{eq:prepreqg}
{\mathbb E}_{U, U^n \sim R} \left[n \cdot \log \left(
q_{\vm^*}(U)/ q_{\breve{\vm}_{|n}}(U) 
\right) \right]  
= {\mathbb E}_{U^n \sim R} \left[n \cdot D(
Q_{\vm^*} \| Q_{\breve{\vm}_{|n}}) \right] = \frac{d_{rq}}{2} +O(n^{-1/4})  \\
\nonumber
& \text{and\ } {\mathbb E}_{U^n \sim R} \left[\log \frac{q_{\vm^*}(U^n)}{ \prod_{i=1}^n q_{\breve{\vm}_{|i-1}}(U_i)}\right] = \sum_{i=1}^n
{\mathbb E}_{U^{i-1} \sim R} \left[D(
Q_{\vm^*} \| Q_{\breve{\vm}_{|i-1}}) \right] =
\frac{d_{rq}}{2} \log n + O(1).
\end{align}
\commentout{
Thus, if $n_0 = 0$ (i.e. $\breve{\vm}_{|n}$ is chosen equal to $\hvm_{|n}$) and $n > 1$ we get:\begin{align}\label{eq:prepreqg}
{\mathbb E}_{U, U^n \sim R} \left[n \cdot \log \left(
q_{\vm^*}(U)/ q_{\breve{\vm}_{|n}}(U) 
\right) \right]  
= {\mathbb E}_{U^n \sim R} \left[n \cdot D(
Q_{\vm^*} \| Q_{\breve{\vm}_{|n}}) \right] = \frac{d_{rq}}{2}.
\end{align}}
We further have:   
    \begin{align}\label{eq:MLfullsampleg}
{\mathbb E}_{U^n \sim R} \left[\log  \frac{q_{\vm^*}(U^n)}{q_{\hvm_{|n}}(U^n)}\right] = - {\mathbb E}_{U^n \sim R} \left[n \cdot D(
 Q_{\hvm_{|n}} \| Q_{\vm^*}) \right]
\leq - \frac{d_{rq}}{2} + o(1),
\end{align}
where the inequality becomes an equality if further Condition~\ref{cond:uinew} (`{\bf UI}$^{\geq}$') holds. And finally, if $R$ has density $r$ and for some $\epsilon > 0$, ${\mathbb E}_{R} |\log q_{\vm^*}(U)/r(U)|^{1+\epsilon} < \infty$, then 
\begin{align}
 \label{eq:gaussredundancyg}
 {\mathbb E}_{U^n \sim R} \left[\log \frac{
 q_{\vm^*}(U^n)}
 { q_{W_1}(U^n)} \right]  = 
\frac{d}{2} \log \frac{ n}{2 \pi} - \frac{1}{2} \log \det \Sigma_q(\vm^*) - \log  w_1(\vm^*) 
- \frac{d_{rq}(\vm^*)}{2}
+ o(1).
\end{align}
\end{proposition}
\begin{proof}
The  proposition arises as an immediate corollary of Lemma~\ref{lem:KLbounds}, which provides one-sided versions of the statements above. It is stated in Appendix~\ref{app:preparinggeneraltheoremsA} and proved in Appendix~\ref{app:KLboundsproof}.  To prove the corollary, take $\alpha =1/4$ in Lemma~\ref{lem:KLbounds} and simply piece together the right equations in each case. 
\end{proof}
\subsection{Proof of Theorem~\ref{thm:simpleH1general}}
Part 1, {\bf UI.\ } is proved in exactly the same way as the {\bf UI} proof for the Gaussian case, Theorem~\ref{thm:simpleH1gauss}, 
using Proposition~\ref{prop:basicgeneral} instead of Proposition~\ref{prop:basicgauss}. 
As to Part 2, {\bf COND},
With $Z:= n^{1/2} \hvm_{|n}$  defined as in the proof of Theorem~\ref{thm:simpleH1gauss} in the Gaussian case,
we can write 
\begin{align*}
S_{Q,\cond}^{(n)} = \frac{q(\sn{U}{n} \mid Z)}{p(\sn{U}{n} \mid Z)}
= \frac{q(\sn{U}{n} )}{p_{\vec{\mu}^*}(\sn{U}{n} )} 
\cdot \frac{p^{\circ}_{\vec{\mu}^*}( Z)}{q^{\circ}(Z)},
\end{align*}
where $p^{\circ}$ and $q^{\circ}$ are the densities, under $P$ and $Q$ respectively, of $Z= \sqrt{n}\hvm_{|n}$.
Taking logarithms and expectation this gives
\begin{align}\label{eq:condy}
{\mathbb E}_R[\log S_{Q,\cond}^{(n)}] = 
n D_R(Q \| P_{\vec{\mu}^*}) + 
{\mathbb E}_R \left[\log \left( {p^{\circ}_{\vec{\mu}^*}(Z)}/ {q^{\circ}( Z)}\right)
\right].
\end{align}
\newcommand{\cA}{\ensuremath{\mathcal{A}}}
\newcommand{\cB}{\ensuremath{\mathcal{B}}}
The difference to the Gaussian case is that $p^{\circ}_{\vm^*}$
and $q^{\circ}$ are now not Gaussian densities, but rather the densities or mass functions of two distributions which, by the Central Limit theorem (CLT), both converge to a normal distribution. This 
already suggests that the desired result (\ref{eq:gencond}) might hold, but the convergence implied by the standard CLT is too weak: we need convergence, including rates, of the densities (not just the distributions) to the density of a normal. This is provided by a version of the multivariate local central limit theorem
due to Bhattacharya and Rao \citep{BhattacharyaR76}. The actual proof, bounding the final term in (\ref{eq:condy}) and thereby giving (\ref{eq:gencond}), follows by Lemma~\ref{lem:condsimpleH1} provided in Appendix~\ref{app:longproofssimpleH1general}, which crucially uses the results from \cite{BhattacharyaR76} and combines them with exponential concentration bounds. 

As to Part 3, {\bf seq-RIPr},
the simple case goes precisely as in the Gaussian case of Theorem~\ref{thm:simpleH1gauss}.
It remains to prove (\ref{eq:noncompetitive}) in the strict anti-simple case. Since $\Sigma_q - \Sigma_p(\vm^*)$ is now assumed positive definite, it follows from \cite[Proposition 2]{GrunwaldLHBJ24} that $q(U^n)/p_{\vm^*}(U^n)$ is not an e-variable. Yet $P_{\vm^*}$ achieves $\min_{\vm \in \meanspace_p} D(Q \| P_{\vm})$. For such a combination, Theorem 1 in \cite{GrunwaldHK19} gives that the infimum over all distributions $P \in \textsc{conv}(\nulhyp(U))$ of $D\left(Q(U) \| P(U) \right)$ lies outside $\nulhyp(U)$, which implies that
\begin{align}\label{eq:afterarxiv}
    {\mathbb E}_Q \left[ \log {q(U)}/{
    p_{ \leftsquigarrow q(U)}(U)} \right] = {\mathbb E}_Q\left[\log {q(U)}/{p_{\vm^*}(U)} \right] - \epsilon
\end{align}
for some $\epsilon > 0$. The result then follows by the definition (\ref{eq:seqripnew}).

Finally, as to  {\bf RIPr}, the anti-simple case:
(\ref{eq:gekb}) is based on the following:
\begin{lemma}\label{prop:difficult}
In the setting of Theorem~\ref{thm:simpleH1general}, the anti-simple case (i.e. $\Sigma_q - \Sigma_p(\vm^*)$ is positive semidefinite),  we have, with $W_{0,(n)}$ as in the theorem statement:  
$D(Q(U^n) \| P_{W_{0,(n)}}(U^n))
    ] \leq  n 
    D(Q||P_{\vec{\mu}^*}) - D_{\gauss}(\Sigma_q \Sigma_p^{-1}) 
    + o(1).
$
\end{lemma}
We place the proof in Appendix~\ref{app:difficult}.
Now note that 
\begin{align*}
{\mathbb E}_{ Q} \left[
    \log S_{Q,\cond}^{(n)}
    \right] \leq 
{\mathbb E}_{ Q} \left[
    \log S_{Q,\rip}^{(n)}
    \right] = \inf_{W } D\left(Q(\sn{U}{n})\big|\big|P_{W}(\sn{U}{n}) \right) 
    \leq
    D\left(Q(\sn{U}{n})\big|\big|P_{W_{0,(n)}}(\sn{U}{n}) \right),
\end{align*}
where the first inequality follows from (\ref{eq:basicepower}), the infimum is over all priors on $\meanspace_p$ and the  equality follows from Theorem 1 in \cite{GrunwaldHK19} (see also underneath (\ref{eq:ripr})). 
The second inequality is immediate. Together with the equality it implies (\ref{eq:gekb}). The first equality together with (\ref{eq:gencond}) (use $R=Q$) implies that, if Condition~\ref{cond:cond} (`\ccond') holds, then (\ref{eq:gekb}) holds with equality.

\subsection{Proof of Theorem~\ref{thm:compositeH1general}}
Part 1, {\bf UI} is proved in exactly the same way  as the {\bf UI} proof for the Gaussian case,
using Proposition~\ref{prop:basicgeneral} instead of Proposition~\ref{prop:basicgauss}.
\commentout{
Let $S_{\breve{\vec\mu},\ui}^{(n)} = \frac{\prod\limits_{i=1}^n q_{\breve{\vec{\mu}}_{i-1}}(X_i)}{p_{\hat{\vec{\mu}}|\sn{X}{n}}(\sn{X}{n})}$, then
\begin{align}
  \mathbb{E}_{\sn{U}{n} \sim R} \left[
    \log S_{\breve{\vec\mu},\ui}^{(n)}
    \right]
    =&
    \mathbb{E}_{\sn{U}{n} \sim R} \left[
    \log \frac{\prod\limits_{i=1}^n q_{\breve{\vec{\mu}}_{i-1}}(X_i)}{q_{\vec{\mu}^*}(\sn{X}{n})}
    \right]
    + \mathbb{E}_{\sn{U}{n} \sim R} \left[
    \log \frac{q_{\vec{\mu}^*}(\sn{X}{n})}{p_{\vec{\mu}^*}(\sn{X}{n})}
    \right]\nonumber\\
    &+ \mathbb{E}_{\sn{U}{n} \sim R} \left[
    \log \frac{p_{\vec{\mu}^*}(\sn{X}{n})}{p_{\hat{\vec{\mu}}|\sn{X}{n}}(\sn{X}{n})}
    \right]\nonumber\\
    =&
    n \left(D_R(Q \| P_{\vec{\mu}^*}) \right) - \frac{1}{2}\textsc{tr}\left(\Sigma_r \Sigma_\alti^{-1}(\vec{\mu}^*) \right) \log n - \frac{1}{2}\textsc{tr}\left(\Sigma_r \Sigma_\nuli^{-1}(\vec{\mu}^*) \right) + O(1),
   \end{align}
which follows from Lemma \ref{lemma:prequential_MLE} and Lemma \ref{lem:UIsimpleH1}.
}
Part 2, {\bf COND},  follows in exactly the same way as the corresponding result in Theorem~\ref{thm:simpleH1general}. Part 3, {\bf seq-RIPr}, the simple case,
follows again in exactly the same way as the corresponding proof for the Gaussian case, {\bf seq-RIPr} in Theorem~\ref{thm:compositeH1gauss}, with all applications of Proposition~\ref{prop:basicgauss} replaced by corresponding applications of Proposition~\ref{prop:basicgeneral}.
As to Part 3, {\bf seq-RIPr}, the anti-simple case, (\ref{eq:seqbad}) follows by extending the proof of (\ref{eq:noncompetitive}) with straightforward continuity arguments; details are in Appendix~\ref{app:noncompetitive}.
Finally, for Part 4, {\bf RIPr}, to 
prove (\ref{eq:gekc}), we equip $\meanspace_p$ with the {\em same\/} prior $W_1$ as $\meanspace_q$ (with density set to $0$ on the set $\meanspace_p \setminus \meanspace_q$). We then  use the familiar Laplace approximation of the Bayesian marginal likelihood given by (\ref{eq:gaussredundancyg}), both for $\althyp$ and $\nulhyp$. 
Thus, we apply these bounds twice, both times with $R=Q$ and then first with $q_{W_1}$ and $q_{\vm^*}$ and then with $p_{W_1}$ and $p_{\vm^*}$ respectively; the moment regularity condition automatically holds, because it reduces to the expectation under $Q$ of an expression that is a quadratic form of $X$, and $Q$ is a member of an exponential family with sufficient statistic $X$, implying that such expectations are finite.
This gives that 
\begin{align*}
    {\mathbb E}_Q\left[ \log\frac{q_{W_1}(U^n)}{p_{W_1}(U^n)}\right] & = 
    n D(Q_{\vm^*} \| P_{\vm^*}) +
   {\mathbb E}_Q\left[ \log\frac{q_{W_1}(U^n)}{q_{\vm^*}(U^n)} 
   + \log \frac{p_{\vm^*}(U^n)}{p_{W_1}(U^n)}\right] \nonumber \\ 
& = n D(Q_{\vm^*} \| P_{\vm^*})+ \frac{1}{2}  \log \frac{\det \Sigma_q(\vm^*)}{\det \Sigma_p(\vm^*)} + \frac{d}{2} - \frac{d_{qp}(\vm^*)}{2} + o(1),
\end{align*}
where in the second equality we used (\ref{eq:gaussredundancyg}). The result follows by recognizing that the latter three terms together are equal to $- D_{\gauss}(\Sigma_q \Sigma_p^{-1})$.

\section{Implications, conclusions and future work}
\label{sec:conclusion}
\newcommand{\seqcond}{\ensuremath{{\textsc{seq-cond}}}}
It is of interest to determine when a sequence of e-variables $(S^{(n)})_{n \geq 1}$ constitutes an e-process, since in that case we can use it not just in optional continuation but also in optional stopping scenarios. 
While  $S_{\seqrip}^{(n)}$ and $S_{\ui}^{(n)}$ provide e-processes  by construction, {\em e-processness\/} is not so clear for $S_{\cond}$ and $S_{\rip}$. 
In Appendix~\ref{app:eprocess} we provide the definition of e-process and a proposition that allows us to identify several  $\nulhyp$ and $\althyp$ for which they are {\em not\/} e-processes; on the other hand, we also identify cases in which they do provide an e-process.
As discussed there, one may informally conjecture that with the conditional and RIPr e-variables we `almost' obtain an e-process, leading perhaps to `approximate' handling of optional stopping. Investigating further and formalizing `approximate optional stopping' (perhaps based on asymptotic anytime-validity \cite{waudby2021time}) is a main avenue for further research.

\section{Acknowledgements}
Peter Grünwald is also affiliated with the Mathematical Institute of Leiden University. We had substantial, generous and much appreciated help from Michael Lindon (Netflix) with an earlier version of Theorem 1, Part 3.  Shaul K. Bar-Lev and G\'erard Letac pointed us towards the Landau family (see discussion underneath Condition~\ref{cond:plugin}). We would also like to thank Tyron Lardy for various discussions. This work is partly supported by the China Scholarship Council State Scholarship Fund Nr.202006280045.

\section*{Supplementary Material}
This supplementary material contains
Appendices A to D, which include the proofs of technical lemmas and propositions and contain additional details for Section~3 and Section~6.


\DeclareRobustCommand{\VANDER}[3]{#3}
\bibliographystyle{imsart-nameyear.bst}
\bibliography{references,SAVIreferences}
\ \\ \ \\

\begin{supplement}
\appendix

\section{Proofs underlying Theorem~\ref{thm:simpleH1gauss} and~\ref{thm:compositeH1gauss}:  Gaussian case}
\label{app:longproofssimpleH1gauss}

\ownparagraph{Proof of Proposition~\ref{prop:basicgauss}}
    We first prove (\ref{eq:MLfullsample}). 
    The first equality follows by simple algebra (it is also a consequence of the robustness property of exponential families, see Lemma~\ref{lem:KLbounds} later on). We first establish that the first expression is equal to the third. Filling in the densities we find:
    \begin{align*}
   {\mathbb E}_R \left[\log q_{\hat{\vm}_{|n}}(U^n)/q_{\vm^*}(U^n)\right]  = & 
   \frac{1}{2} \sum_{i=1}^n {\mathbb E}_R \left[
   (X_i - {\vm}^*)^\top \Sigma_q^{-1} (X_i - \vm^*) \right] \nonumber \\ & 
   - \frac{1}{2} \sum_{i=1}^n {\mathbb E}_R \left[
   (X_i - \hat{\vm}_{|n})^\top \Sigma_q^{-1} (X_i - \hat\vm_{|n}) \right] .
    \end{align*}
    We have $\mathbb{E}_R[X_i - \vm^*] = 0$ and 
    $\text{\sc cov}_R(X_i - \vm^*, X_i - \vm^*) = \Sigma_r$. By the property of expectation of a quadratic form, we get
    that for the $i$-th term in the first sum above that it is equal to $\textsc{tr}(\Sigma_r \Sigma_q^{-1})$.
It is easy to show $\mathbb{E}_R[X_i - \hvm_{|n}] = 0$ and $\text{\sc cov}_R(X_i - \hvm_{|n}, X_i - \hvm_{|n}) = \left(1 - \frac{1}{n} \right)\Sigma_r$. Again by the property of expectation of a quadratic form, we thus get for the $i$-th term in the second sum:
$$
\mathbb{E}_{\sn{X}{n}\sim R}\left[(X_i - \hvm_{|n})^\top \Sigma_q^{-1} (X_i - \hvm_{|n}) \right]
=
\textsc{tr}\left(\left(1 - \frac{1}{n} \right)\Sigma_r \Sigma_q^{-1} \right), 
$$
so that the right-hand side of (\ref{eq:MLfullsample}) becomes $\frac{1}{2} \tr(\Sigma_r \Sigma_q^{-1})$.
\\ \ \\ \noindent
We proceed to prove (\ref{eq:newproof}). The first equality is straightforward.
Filling in the densities, we find that it is equal to 

  \begin{align}
\frac{1}{2} 
   \mathbb{E}_{X, \sn{X}{n}\sim R}\left[(X - \breve{\vec{\mu}}_{|n})^\top \Sigma_\alti^{-1} (X - \breve{\vec{\mu}}_{|n}) \right]
   -
   \frac{1}{2}{\mathbb E}_R \left[
   (X - {\vm}^*)^\top \Sigma_q^{-1} (X - \vm^*) \right]. \nonumber 
     \end{align}
We already showed the second term is equal to $(1/2) \textsc{tr}(\Sigma_r \Sigma_q^{-1})$.
As to the first term, it is easy to show $\mathbb{E}_R[X - \breve{\vec{\mu}}_{|n}] = 
\frac{n_0(\vec{\mu}^* - x_0)}{n_0 + n}$ and $\text{\sc cov}_R(X - \breve{\vec{\mu}}_{|n}, X - \breve{\vec{\mu}}_{|n}) = \left(1 + \frac{n}{(n_0 + n)^2} \right)\Sigma_r$.
By the property of expectation of a quadratic form, we have
\begin{align*}
\mathbb{E}_{X\sim R}\left[(X - \breve{\vec{\mu}}_{|n})^\top \Sigma_\alti^{-1} (X - 
\breve{\vec{\mu}}_{|m}) \right]
=&
\frac{n_0^2}{(n_0 + n)^2} (\vec{\mu}^* - x_0)^\top \Sigma_\alti^{-1} (\vec{\mu}^* - x_0)\\
&+
\textsc{tr}\left(\left(1 + \frac{n}{(n_0 + n)^2} \right) \Sigma_r \Sigma_\alti^{-1} \right),
\end{align*}
and from this the second equality in (\ref{eq:newproof}) follows. The third equality is immediate, and so is 
(\ref{eq:prepreq}). 
(\ref{eq:preqfullsample}) follows by repeatedly applying, for each $i$, the first equality in (\ref{eq:newproof}) and then the second equality in (\ref{eq:newproof}), both with $i$ in the role of $n$. 

We now prove (\ref{eq:gaussredundancy}). If $W_1 =  N(\vm_1,\Pi_1)$ is a normal prior, we find  
\begin{equation}\label{eq:r_KL}
D_R(Q_{\vm^*}(Z) \| Q_{W_1}(Z)) = 
D_{\Sigma_r}(\Sigma_q \| (n \Pi_1 + \Sigma_q)   ) 
+ \frac{1}{2} (\vm^* - \vm_1)^\top (\Pi_1 + \Sigma_q/n)^{-1} (\vm^* - \vm_1)
\end{equation}
by combining  (\ref{eq:marginalmle}) and (\ref{eq:normalKL}) above (setting $P$ in (\ref{eq:normalKL}) to $Q_W(Z)$) .
This gives the first equality in (\ref{eq:gaussredundancy}).
In combination with (\ref{eq:preconditionalb}) and (\ref{eq:prepreq}), which we already proved, 
and using that 
\begin{align*}
- \log w_1(\vm^*) = \frac{1}{2} ( \vm^* - \vm_1)^\top \Pi_1^{-1} ( \vm^* - \vm_1) 
+ \frac{d}{2} \log 2 \pi + \frac{1}{2} \log \det \Pi_1, 
\end{align*}
we get the second equality, by  plugging it into (\ref{eq:r_KL}) directly; the $o(1)$ derives from $\Pi_1 + \frac{\Sigma_q}{n} \rightarrow \Pi_1$ as $n$ increases.

\ownparagraph{Proof of (\ref{eq:haariscond})}
We have:
\begin{align}
  &  {\mathbb E}_{\sn{U}{n} \sim R} [
    \log S_{\haar}^{(n)}
    ] =  (n-1) D_{\Sigma_r}(\Sigma_q \| \Sigma_p) + \nonumber \\
&    
    {\mathbb E_{X_1 \sim R}} \left[ 
    {\mathbb E}_{X_2, \ldots, X_n} \left[ \log \frac{q_{W_1\mid X_1}(X_2, \ldots, X_n)}{q_{\vm^*}(
    X_2, \ldots, X_n
    )} - \log \frac{p_{W_0 \mid X_1}(X_2, \ldots, X_n)}{
p_{\vm^*}(
    X_2, \ldots, X_n
    )    
    } \right] \right]  
    \nonumber \\ & = {\mathbb E}_R\left[ 
    - D_{ \Sigma_r}(\Sigma_q \|  ((n-1) \Sigma_q + \Sigma_q)  ) 
- \frac{1}{2} (\vm^*  - X_1)^\top (\Sigma_q + \Sigma_q/(n-1))^{-1} (\vm^* - X_1)\right] +
\nonumber\\ \label{eq:pushthrough}
&{\mathbb E}_R\left[ 
    D_{\Sigma_r}(\Sigma_p|| ((n-1) \Sigma_p + \Sigma_p)  ) 
+ \frac{1}{2} (\vm^*  - X_1)^\top (\Sigma_p + \Sigma_p/(n-1))^{-1} (\vm^* - X_1)\right]+ \nonumber\\
& (n-1) D_{  \Sigma_r}(\Sigma_q \| \Sigma_p)   \\ 
&=  
    - D_{  \Sigma_r}(\Sigma_q \| n \Sigma_q )  
- \frac{1}{2} \tr(\Sigma_r (\Sigma_q + \Sigma_q/(n-1))^{-1})  \nonumber \\ &  + D_{\Sigma_r}(\Sigma_p|| n \Sigma_p  ) 
+ \frac{1}{2} \tr(\Sigma_r (\Sigma_p + \Sigma_p/(n-1))^{-1} )
+ (n-1) D_{  \Sigma_r}(\Sigma_q \| \Sigma_p) \nonumber \\
& 
=   - \frac{1}{2} \tr(\Sigma_r ((n\Sigma_q )^{-1}- \Sigma_q^{-1})
- \frac{1}{2} \tr(\Sigma_r (\Sigma_q \cdot n/(n-1))^{-1})  \nonumber \\ &  +\frac{1}{2} \tr(\Sigma_r ((n\Sigma_p )^{-1}- \Sigma_p^{-1})
+ \frac{1}{2} \tr(\Sigma_r (\Sigma_p \cdot n/(n-1))^{-1} )
+ (n-1) D_{  \Sigma_r}(\Sigma_q \| \Sigma_p) \nonumber
\\ &= (n-1) D_{  \Sigma_r}(\Sigma_q \| \Sigma_p),\nonumber
\end{align} 
where for (\ref{eq:pushthrough}) we used (\ref{eq:gaussredundancy}) in Proposition~\ref{prop:basicgauss}, and in the final equation we used definition (\ref{eq:triplegauss}).

\section{Proofs underlying Theorem~\ref{thm:simpleH1general} and~\ref{thm:compositeH1general}: the general case}
\label{app:longproofssimpleH1general}
\subsection{Definitions for and statement of Lemma~\ref{lem:KLbounds}}
\label{app:preparinggeneraltheoremsA}
Condition~\ref{cond:uinew} and Lemma~\ref{lem:KLbounds} refer to the KL divergence $D(P_{\hat{\vec{\mu}}_{|n} } \| P_{\vec{\mu}^*} )$, which  is undefined if $\hvm_{|n} \not \in \meanspace_p$ (which may happen even for the simple Bernoulli model if $\hvm_{|n} \in \{0,1\}$). It is highly convenient to extend its definition to such cases, and this can be done in a straightforward manner. 
We first set  $\bar{\meanspace}_p$ to be the union of the mean-value parameter space $\meanspace_p$ of $\nulhyp$ and the set of values that $\hvm_{|n} = n^{-1} \sum X_i$ can take.  For all  $\hvm_{|n} \in \bar\meanspace_p \setminus \meanspace_p$, we first set 
\begin{equation}\label{eq:fulrobust}
\frac{p_{\hvm_{|n}}(U^n)}{p_{\vec{\mu}^*}(U^n)} := 
\sup_{\vm \in \meanspace_p} \frac{p_{\vm}(U^n)}{p_{\vec{\mu}^*}(U^n)},
\end{equation}
and we note that then, in terms of the canonical parameterization, 
by (\ref{eq:canonical}), with $\vb^*$ the canonical parameter corresponding to $\vm^*$ and $X_i = (t_1(U_i), \ldots, t_d(U_i))^\top$, 
$$
\log \frac{p_{\hvm_{|n}}(U^n)}{p_{\vec{\mu}^*}(U^n)}
= \sup_{\vb \in \canspace_p} ({\vb} - \vb^*)^\top \sum_{i=1}^n X_i
- \log (Z_p(\vb) / Z_p(\vb^*) ),
$$
which can be written as a function of $\hat{\vm}_{|n}$. 
Therefore,  for  $\hvm_{|n} \in \bar\meanspace_p \setminus \meanspace_p$, we can unambiguously set 
\begin{equation}\label{eq:fullrobustness}
D(P_{\hat{\vm}_{|n}} \| P_{\vm^*}) := \frac{1}{n} \log \frac{p_{\hat{\vm}_{|n}}(U^n)}{p_{\vm^*}(U^n)},
\end{equation}
since the expression on the right only depends on $U^n$ through $\hat{\vm}_{|n}$. 
The rationale underlying definition (\ref{eq:fullrobustness}) is that (\ref{eq:fullrobustness}) holds any way as long as $\hat{\vm}_{|n} \in \meanspace_p$. This is a well-known result, with straightforward proof, sometimes referred to as
\commentout{we  will  then have for all $U^n \in \cU^n$, 
\begin{equation}\label{eq:fullrobustness}
n D(P_{\hvm_{|n}} \| P_{\vec{\mu}^*} ) = \log \frac{p_{\hvm_{|n}}(U^n)}{p_{\vec{\mu}^*}(U^n)}.
\end{equation}
This result, known as}
the {\em KL robustness property for exponential families\/} \citep[Chapter 19]{grunwald2007minimum}.
We have thus extended it to hold for all $\hat{\vm}_{|n} \in \bar\meanspace_p$ by definition (\ref{eq:fulrobust}), and thereby already proved the first statement in Lemma~\ref{lem:KLbounds} below.

We only give the remaining results in the lemma under existence of an odd number of moments $m$ under $R$, since,  for even $m'$, if $X$ has $m'$ moments then it also has $m=m'-1$ moments and due to certain cancellations, the result one obtains with $m'$ in the proof would not be stronger than what one obtains for $m$. 
\begin{lemma}\label{lem:KLbounds}
First, extend the definition of $D(P_{\hat{\vm}_{|n}} \| P_{\vm^*})$ as above. Then the  robustness property for exponential families holds for all $U^n \in \cU^n$:
\begin{equation}\label{eq:fullrobustnessb}
n D(P_{\hvm_{|n}} \| P_{\vec{\mu}^*} ) = \log \frac{p_{\hvm_{|n}}(U^n)}{p_{\vec{\mu}^*}(U^n)}.
\end{equation}
Suppose that the first $m$ moments of $X$ under $R$ exist,  with $m \geq 3$ odd, and ${\mathbb E}_R[X] = \vm^*$. Fix $0 < \alpha < 1/2$. Let $\Sigma_p$ be a general $d \times d$ positive definite matrix. We have: 
\begin{align}
    \label{eq:mlesquare}
& 
{\mathbb E}_R \left[{\bf 1}_{\| \hvm_{|n}- \vec{\mu}^* \|_2 \geq n^{\alpha - 1/2}}  \cdot  n \cdot \frac{1}{2} (\hat\vm_{|n} - \vm^*)^\top \Sigma_p^{-1} (\hat\vm_{|n} - \vm^*)\right]
= 
 O\left(n^{-(m-2) \alpha -  1/2} \right), \\
 \label{eq:KLboundsSquare}
 & 
 {\mathbb E}_R \left[ n \cdot \frac{1}{2} (\hvm_{|n} - \vm^*)^\top \Sigma_p^{-1}(\vm^*) (\hvm_{|n} -\vm^*)\right]
= \frac{\tr(\Sigma_r\Sigma_p^{-1}(\vm^*))}{2}.
 \end{align}
Next, let $\nulhyp$ be a regular exponential family with $\vm^* \in \meanspace_p$. 
Fix $c > 0$ and let $\meanspace' \subset \meanspace_p$ be an open neighborhood of $\vm^*$.
Let $(W_{0,(n)})_{n \in \naturals}$ be any sequence of distributions on $\meanspace_p$ with continuous densities $w_{0,(n)}$ such that  for every $n$, $w_{0,(n)}(\vm) > c$ for all $\vm \in \meanspace'$. We have: 
 \begin{align}
\label{eq:bayesbound} &  {\mathbb E}_R \left[{\bf 1}_{\| \hvm_{|n}- \vec{\mu}^* \|_2 \geq n^{\alpha - 1/2}}  \cdot  \log 
\frac{p_{\vm^*}(U^n)}{p_{W_{0,(n)}}(U^n)}  \right] \leq 
 O( n^{-(m-2) \alpha -  1/2}  ),
 \end{align}
and, if $W_{0,(n)}$ is as above and for some $\epsilon > 0$, ${\mathbb E}_{R} |\log p_{\vm^*}(U)/r(U)|^{1+\epsilon} < \infty$, then 
\begin{align}
\label{eq:bayesboundb} &  {\mathbb E}_R \left[{\bf 1}_{\| \hvm_{|n}- \vec{\mu}^* \|_2 \geq n^{\alpha - 1/2}}  \cdot  \log 
\frac{p_{W_{0,(n)}}(U^n)}{p_{\vm^*}(U^n)}  \right] \leq O(n^{- \alpha m \cdot (1 \wedge \epsilon)}).
 \end{align}
Now let  $\breve\vm_{|n}$ be defined as in (\ref{eq:preq}), with $n_0 > 0$. If (c) in  
Condition~\ref{cond:plugin} (`{\bf plug-in}') holds (with $P_{\cdot}$ in the role of $Q_{\cdot}$), then  
\begin{align}
    \label{eq:mleoutrightB} & 
{\mathbb E}_R \left[{\bf 1}_{\| \hvm_{|n}- \vec{\mu}^* \|_2 \geq n^{\alpha - 1/2}}  \cdot  n \cdot  D( P_{\vec{\mu}^*}\| P_{\breve{\vec{\mu}}_{|n}})  \right] = 
 O( n^{-m \alpha +  1/2}  ).
\end{align} 
With
 $$
f(n) :=  {\mathbb E}_R \left[{\bf 1}_{\| \hvm_{|n}- \vec{\mu}^* \|_2 <  n^{\alpha - 1/2}}  \cdot  n \cdot \frac{1}{2} (\hat\vm_{|n} - \vm^*)^\top \Sigma_p^{-1}(\vm^*) (\hat\vm_{|n} -\vm^*)\right],
 $$
 we have, with $W$ now a fixed distribution on $\meanspace_p$ with density $w$ that is positive and continuous in a neighborhood of $\vm^*$,
 \begin{align}
 \label{eq:KLboundsLeft}
&
{\mathbb E}_R \left[{\bf 1}_{\| \hvm_{|n}- \vec{\mu}^* \|_2 <  n^{\alpha - 1/2}}  \cdot  n \cdot D(P_{\hvm_{|n}} \| P_{\vec{\mu}^*} )\right] = f(n) + O\left(n^{\alpha - 1/2}\right),
\\ \label{eq:KLboundsRight}
 & 
{\mathbb E}_R \left[{\bf 1}_{\| \hvm_{|n}- \vec{\mu}^* \|_2 < n^{\alpha - 1/2}}  \cdot  n \cdot  D( P_{\vec{\mu}^*}\| P_{\breve{\vec{\mu}}_{|n}})  \right] 
= f(n) + O\left(n^{\alpha - 1/2}\right), \\ \label{eq:bic}
 & 
{\mathbb E}_R \left[{\bf 1}_{\| \hvm_{|n}- \vec{\mu}^* \|_2 < n^{\alpha - 1/2}}  \cdot  \log \frac{p_{\hat{\vm}_{|n}}(U^n)}{p_W(U^n)}   \right] \\ \nonumber
=& \frac{d}{2} \log \frac{n}{2\pi} -  \log   w(\vm^*) - \frac{1}{2}  \log  \det \Sigma_p(\vm^*) + o(1).
 \end{align}
\end{lemma}

\subsection{Preparatory results for Lemma~\ref{lem:KLbounds} and Lemma~\ref{prop:difficult}}
\label{app:preparinggeneraltheoremsB}
Throughout the proofs for Lemma~\ref{lem:KLbounds} and Lemma~\ref{prop:difficult} and Theorem~\ref{thm:simpleH1general}, Part 3 (Appendix~\ref{app:condproof}), we will use the following additional notations and abbreviations: 
\begin{align}
    \label{eq:notations}
\vm^* = {\bf 0}, \Sigma_p:= \Sigma_p({\bf 0}), \Sigma_q:= \Sigma_q({\bf 0 }),
K := (\Sigma_q - \Sigma_p)^{-1}, J(\vm) := \Sigma_p(\vm) ^{-1}, \Jp:= J({\bf 0})= \Sigma_p^{-1}.
\end{align}
Note that the first equation, $\vm^* = {\bf 0 }$ --- which will allow substantially shortening equations ---  is really a definition rather than a notation; but we can enforce this without loss of generality --- since $\vm^*$ is fixed throughout all proofs, we can get the same result as for ${\bf 0}$ with arbitrary $\vm^*$ simply by adding $\vm^*$ as a constant to each outcome $X$ and transforming the resulting equations. Or, alternatively, one may modify all equations by putting in `$- \vm^*$' at the appropriate places and note that nothing in the derivation changes.

We will make repeated use of the following basic results whose proof  can be found in, for example, \cite{grunwald2005asymptotic}.
\begin{lemma}\label{lemma:Using Markov Inequality}
Fix $n_0 \geq 0$ and $x_0 \in \reals$. Let $X, X_1, \ldots$ be (scalar-valued) i.i.d. random variables,
define $\breve{\mu}_{|n} := ({n_0\cdot x_0 + \sum\limits_{i=1}^n X_i})/({n+n_0})$.
Suppose the first $m \in \mathbb{N}$ moments of $X$ exist and let $\delta > 0$. Then $\Pr(|\breve{\mu}_n -  \mathbb{E}[X] | \geq \delta) = O(n^{-\lceil \frac{m}{2} \rceil} \delta^{-m})$. As a consequence, via the union bound, for fixed $d$, our $d$-dimensional regularized MLE $\breve{\vm}_{|n}$ also satisfies $\Pr(\|\breve{\vm}_{|n} - \vm^* \|_2 \geq \delta) = O(n^{-\lceil \frac{m}{2} \rceil} \delta^{-m})$.
\end{lemma}

\begin{lemma}\label{lemma:lceil_rceil}
Consider the setting of Lemma~\ref{lemma:Using Markov Inequality}. If the first $m$ moments of (scalar) $X$ exist, then $\mathbb{E}\left[\left( \breve{\mu}_{|n} - \mathbb{E}[X] \right)^m \right] = O(n^{-\lceil \frac{m}{2} \rceil})$.
\end{lemma}

We also need the following two propositions: 
\begin{proposition}
    \label{prop:thirdorderTaylor}
    Let $0 < \alpha < 1/2$ and fix $c > 0$. Then  there is a constant $C> 0$ such that for all large $n$, 
    for all $\vm= (\vm_1, \ldots, \vm_d)$, $\vm'= (\vm'_1, \ldots, \vm'_d)$ with $(\| \vm \|_2 \vee \| \vm' \|_2  )<   n^{\alpha -1/2}$, 
    we have
    $$\left| 
    D(P_{\vm'}\| P_{\vm}) - \frac{1}{2} (\vm- \vm')^\top J (\vm- \vm')  \right| \leq
    C \sum_{i,j,k \in [d]} |\mu_i - \mu'_i |\cdot | \mu_j- \mu'_j |\cdot  |\mu_k - \mu'_k|
    \leq C n^{3 \alpha - 3/2}.$$
\end{proposition}
\begin{proof}
Note that for all  $n$ large enough, with $\cA_n := \{ \vm: \| \vm \|_2 <   n^{\alpha - 1/2} \}$ and $\cA'_n := \{ \vm: \| \vm \|_2 <   n^{\alpha - 1/2} \}$, we have that  $\cA_n \cup \cA'_n\in \meanspace_p$.
A third order Taylor-expansion of $D(P_{\bf 0 } \| P_{\vm})$ in terms of $\vm$ then gives that: 
$$
    D(P_{\vm'} \| P_{\vm}) = \frac{1}{2} (\vm - \vm')^\top J (\vm - \vm') + 
    O\left( \sum_{i,j,k \in [d]} |C_{i,j,k}| |\mu_i- \mu'_i |\cdot | \mu_j- \mu'_j |\cdot  |\mu_k- \mu'_k |
    \right),$$
where $C_{i,j,k}$ are the corresponding third-order derivatives. But since, for large $n$, $\cA_n$ is a closed set in the interior of $\meanspace_p$, these derivatives are bounded; the result follows.  
\end{proof}

\begin{proposition}\label{prop:frombrevetohat}
We have:
$$
\sup_{U^{n}: \| \hat{\vm}_{|n} \|_2 < n^{\alpha - 1/2} }\left| D(P_{\bf 0} \| P_{\hat{\vm}_{|n}}) - D(P_{\bf 0} \| P_{\breve{\vm}_{|n}}) \right| = O(n^{\alpha - 3/2}).
$$
\end{proposition}
\begin{proof}
Define $f(\vm) := D(P_{{\bf 0}} \| P_{\vm})$ and $g(\vm) := \nabla f(\vm)$. A first-order Taylor approximation gives that 
\begin{align}\label{eq:eierbal}
D(P_{\bf 0} \| P_{\breve{\vm}_{|n}})=  D(P_{\bf 0} \| P_{\hat{\vm}_{|n}})
+ (\breve{\vm}_{|n}- \hat{\vm}_{|n})^{\top} g(\vm')
=  D(P_{\bf 0} \| P_{\hat{\vm}_{|n}})
+ \frac{n_0}{n+n_0} (x_0 - \hat{\vm}_{|n})^{\top} \cdot g(\vm'),
\end{align}
where $\vm'$ is some point with $\| \vm' \|_2 < n^{\alpha - 1/2}$,
and the second equation is straightforward rewriting.

We can write $g(\vm) = (g_1(\vm), \ldots, g_d(\vm))^{\top}$
with $g_j(\vm) = (\partial/\partial \vm_j) f(\vm)$.
Taylor approximating $g_j(\vm)$ itself gives for each $j=1..d$: 
$$
g_j(\vm') = g_j({\bf 0}) + \left(\frac{\partial}{\partial \vm_1} g_j (\vm^{\circ}_{[j]}), \ \ldots \  , \frac{\partial}{\partial \vm_d} g_j (\vm^{\circ}_{[j]}) \right)^\top \vm'
$$
\commentout{
$$
g(\vm') = g({\bf 0}) + \left(\vm'_1  \cdot \frac{\partial}{\partial \vm_1} g_1 (\vm^{\circ}_{[1]}), \ \ldots \  , \vm'_d \cdot \frac{\partial}{\partial \vm_d} g_d (\vm^{\circ}_{[d]}) \right)^{\top}.
$$}
for $d$ points $\vm^{\circ}_{[j]} = (\vm^{\circ}_{[j],1}, \ldots, \vm^{\circ}_{[j],d})^\top$, where, for $j=\{1, \ldots, d\}$, we have  $\| \vm^{\circ}_{[j]} \|_2 < n^{1/2-\alpha}$.
But since $\frac{\partial}{\partial \vm_j} g_j (\vm) = \frac{\partial^2}{ \partial \vm_j^2} f(\vm)$ and for all large $n$, the set $\{ \vm: \| \vm\|_2 < n^{1/2 - \alpha}\} \in \meanspace_p$, these $d$ second-derivatives are uniformly bounded. 
Also, $g({\bf 0})=0$ since the KL divergence is minimized at ${\bf 0}$ and, for exponential families, is a smooth function in the mean-value parameters. It follows that the final term in (\ref{eq:eierbal}) is $O(n^{-1} n^{\alpha-1/2})$, and the result follows.
\end{proof}

\subsection{Proof of Lemma~\ref{lem:KLbounds}}
\label{app:KLboundsproof}
Throughout the proof, we adopt the notations of (\ref{eq:notations}). (\ref{eq:fullrobustnessb}) was already proven. (\ref{eq:bic}) is a well-known result, see for example \cite[Chapter 8]{grunwald2007minimum}. (\ref{eq:KLboundsSquare}) is also known; basically, the expectation on the right can be computed exactly like we did in the proof of Proposition~\ref{prop:basicgauss}, yielding the desired result. 

    \ownparagraph{Proof of (\ref{eq:mlesquare})}
In a variation of the proof of Markov's inequality, we can write the left-hand side, for fixed $\gamma > 0$, as 
\begin{align*}
  &   \mathbb{E}_{R}\left[{\bf 1}_{\| \hvm_{|n} \|_2 \geq n^{\alpha -1/2}} \cdot \frac{(\hvm_{|n}^\top \Sigma_p^{-1} \hvm_{|n} )^{\gamma+1}}{(\hvm_{|n}^\top \Sigma_p^{-1} \hvm_{|n} )^{\gamma}}
  \right] \leq C \cdot  \mathbb{E}_{R}\left[
  \frac{(\hvm_{|n}^\top \Sigma_p^{-1} \hvm_{|n} )^{\gamma+1}}{n^{\gamma (2 \alpha - 1)}}
  \right] \\
& =  C  \cdot  \mathbb{E}_{R}\left[
  \frac{(\hvm_{|n}^\top \Sigma_p^{-1} \hvm_{|n} )^{m/2}}{n^{
  m \alpha - m/2 - 2\alpha +1}}
  \right] = O \left( n^{- m \alpha +m/2 + 2 \alpha -1 } \cdot n^{- (m+1)/2} \right) \\
  & =  O \left( n^{- m \alpha  + 2\alpha  - 3/2 } \right) = 
  O \left( n^{- (m-2)\alpha -3/2} \right)
  \end{align*}
for some constant $C > 0$, where we used that $\Sigma_p^{-1}$ is positive definite and fixed, independent of $\hvm_{|n}$. In the second line we set $\gamma = m/2 -1$, using that the $m$-th moment of the expectation exists, and then again the fact that $\Sigma_p^{-1}$ is positive definite and Lemma~\ref{lemma:lceil_rceil}, using that $m$ is odd. (\ref{eq:mlesquare}) follows by multiplying left- and right-hand side by $n$.

\commentout{According to (\ref{eq:mlesquare}), which we already proved,   we get \\
$\frac{n}{2} \mathbb{E}_R\left[{\bf 1}_{\| \hat{\vec{\mu}} \|_2 > n^{\alpha -1/2}} \cdot \sum\limits_{i=1}^d \sum\limits_{j=1}^d \hat{\vec{\mu}}_i \hat{\vec{\mu}}_j J_{ij}(0) \right] = O(n^{1- m/2})$. Then
\begin{align*}
\frac{n}{2} \mathbb{E}_R\left[{\bf 1}_{\| \hat{\vec{\mu}} \|_2 \leq n^{\alpha -1/2}} \cdot \sum\limits_{i=1}^d \sum\limits_{j=1}^d \hat{\vec{\mu}}_i \hat{\vec{\mu}}_j J_{ij}(0) \right]
=
\frac{n}{2} \mathbb{E}_R\left[\sum\limits_{i=1}^d \sum\limits_{j=1}^d \hat{\vec{\mu}}_i \hat{\vec{\mu}}_j J_{ij}(0) \right]
+
O(n^{1- m/2}).
\end{align*}}

\ownparagraph{Proof of (\ref{eq:bayesbound})}
Assume without loss of generality that $n$ is large enough so that ${\cal A }_n := \{ \vm: \|\vec{\mu} \|_2 \leq n^{-1/2}\}$ is contained in $\meanspace_p$ and for some constant $c> 0$, $w_{(0),n}(\vm) \geq c$ for all $\vm \in {\cal A}_n$. Then 
\begin{align}\label{eq:maxC}
&    - \log 
\frac{\int p_{\vec{\mu}}(U^n) w_{0,(n)}(\vec{\mu}) d \vec{\mu}}{p_{\mathbf 0}(U^n)} 
\leq  - \log 
\frac{\int {\bf 1}_{
{\cal A}_n}
\cdot \left(\min_{\vm' \in \cA_n} p_{\vec{\mu}'}(U^n) \right) w_{0,(n)}(\vec{\mu}) d \vec{\mu}}{p_{\mathbf 0}(U^n)} \\ \nonumber
= & \left( \max_{\vec{\mu} \in {\cal A}_n} 
\ - \log \frac{p_{\vec{\mu}}(U^n)}{p_{\mathbf 0}(U^n)} \right) - \log W_{0,(n)}(\cA_n)
= n \cdot \max_{\vec{\mu} \in {\cal A}_n } f(\hat{\vm}_{|n},\vm) + O\left(\frac{d}{2} \log n \right),
\end{align}
where 
\begin{align*}
     f(\vm^{\circ},\vm) = 
D(P_{\vm^{\circ}} \| P_{\vm}) -D(P_{\vm^{\circ}} \| P_{\bf 0}), 
\end{align*}
and we used the robustness property of exponential families (\ref{eq:fullrobustnessb}) in the final step.

We now move to the canonical parameterization to further analyze the function $f$.
Let $\vec{\beta}_{\vec{\mu}}$ be the function mapping mean-value parameter vector $\vec{\mu}$ to the corresponding canonical parameter vector. 
For general $\vec{\mu}^{\circ} \in \meanspace_p$, we can write 
\begin{align}
& \nabla_{\vec\mu}  f(\vm^{\circ}; \vm)  = \nabla_{\vec\mu} \left( 
D(P_{\vm^{\circ}} \| P^{\textsc{can}}_{\vb_{\vm}}) -D(P_{\vm^{\circ}} \| P_{\bf 0})
\right) = \nonumber \\
& \nabla_{\vec\mu} \left( - (\vec{\beta}_{\vec{\mu}} - \vec{\beta}_{\mathbf 0})^\top 
{\vec{\mu}}^{\circ} + \log Z(\vec{\beta}_{\vec{\mu}})
\right) =  J(\vec{\mu}) \cdot (\vec{\mu}- {\vec{\mu}}^{\circ}), \nonumber
\end{align}
which can be found by the fact that $J(\vm)$ is the matrix of partial derivatives of $\vec{\beta}_{\vec{\mu}}$ as a function of $\vec{\mu}$, and that $\vec{\mu}$ is the gradient of $\log Z(\vec{\beta})$, and using the chain rule of differentiation.
A first order Taylor approximation of $f(\hat{\vm}_{|n}, \vm)$  around $\vec{\mu}=0$, and  then bounding its value on the set $\{ \vm: ||\vm||_2 \leq n^{-1/2} \}$, now gives: 
\begin{align*}
      n \cdot f(\hat{\vm}_{|n}, \vm)
    =  & 
 - n \vec{\mu}^\top  J(\vec{\mu}') \cdot (\vec{\mu}'- \hvm_{|n}) 
\\ 
= &  - n \vec{\mu}^\top  J(\vec{\mu}') c\vec{\mu} + n \vec{\mu}^\top  J(\vec{\mu}') \cdot 
\hvm_{|n}
\leq
n \vec{\mu}^\top  J(\vec{\mu}') \cdot 
\hvm_{|n}
\\
= &  O( n \vec{\mu}^\top  \cdot 
\hvm_{|n}) = 
O(n \cdot  \| \vec{\mu}\|_2 \cdot \| \hvm_{|n} \|_2)
= O( \sqrt{n} \cdot \| \hvm_{|n}\|_2),
\end{align*}
for $\vec{\mu}'= c \vec{\mu}$, for some $c\in [0,1]$.
Here we used that the maximum eigenvalue of $J(\vm')$ is bounded away from $0$ since $\vm$ is inside a compact set in the interior of the parameter space, and we used Cauchy-Schwartz. 

This gives that (\ref{eq:maxC}) is $O (\sqrt{n} \cdot \| \hvm_{|n} \|_2 + (d/2) \log n)$.
On the set with $\| \hvm_{|n} \|_2 \geq n^{\alpha -1/2}$, we can weaken this 
to give that (\ref{eq:maxC}) is $O ({n} \cdot \| \hvm_{|n} \|^2_2)$.
We can now bound this further using (\ref{eq:mlesquare}) from Lemma~\ref{lem:KLbounds}, and the result follows.

\ownparagraph{Proof of (\ref{eq:bayesboundb})}
Let $\cA_n$ denote the event that $\| \hat{\vm}_{|n}- \vm^*\|_2 \geq n^{\alpha-1/2}$. We have, with $W= W_{0,(n)}$: 
\begin{align}\label{eq:louiseb}
   {\mathbb E}_R\left[ {\bf 1}_{\cA_n} \cdot \log \frac{p_W(U^n)}{p_{\vm^*}(U^n)}\right] = {\mathbb E}_R\left[ {\bf 1}_{\cA_n} \cdot \log \frac{p_W(U^n)}{r(U^n)}\right]+ {\mathbb E}_R\left[ {\bf 1}_{\cA_n} \cdot \log \frac{r(U^n)}{p_{\vm^*}(U^n)}\right].
\end{align}
By Lemma~\ref{lemma:Using Markov Inequality} we have $R(\cA_n) = O(n^{(1/2 - \alpha)m}\cdot n^{-(m+1)/2} = O(n^{-\alpha m - 1/2})$, which we use to further bound the first term  by
\begin{align*}
& R(\cA_n) \cdot  {\mathbb E}_R\left[ {\bf 1}_{\cA_n} \cdot \log \frac{p_W(U^n\mid \cA_n)}{r(U^n \mid \cA_n)} \mid \cA_n \right] + R(\cA_n) \log \frac{P_W(\cA_n)}{R(\cA_n)}
\leq \\ & R(\cA_n) \log (1/R(\cA_n)) = O(n^{- \alpha m}),
\end{align*}
where the penultimate inequality follows by Jensen and the fact that $P_W(\cA_n) \leq 1$.
The second term in (\ref{eq:louiseb}) can be bounded by H\"older's inequality, for arbitrary $\epsilon > 0$, as
\begin{align*}
   (R(\cA_n))^{\epsilon/(1+\epsilon)} \cdot  
  \left(  {\mathbb E}_R\left( \log \frac{r(U^n)}{p_{\vm^*}(U^n)} \right)^{1+\epsilon}  \right)^{1/(1+\epsilon)} = O(n^{(-\alpha m - 1/2) \epsilon}),
\end{align*}
and it can be seen that both terms in (\ref{eq:louiseb}) go to $0$ with increasing $n$. The result follows. 

\ownparagraph{Proof of (\ref{eq:mleoutrightB})}
Fix $A> 0$ as in Condition~\ref{cond:plugin} (`{\bf plug-in}') with $P$ in the role of $Q$. Define $ {\cal E}_j =  \{\vec{\mu} \in \bar{\meanspace}_p: \| \vec{\mu} \|_2 \in [A + j-1, A + j]\}$.
With $\alpha= n_0/(n+n_0)$, we have:
\begin{align*}
& {\mathbb E}_R\left[{\bf 1}_{\| \hat{\vec{\mu}}_{|n} - \vec{\mu}^*\|_2 \geq A}
\cdot  D( P_{\vec\mu^*} \| P_{\breve{\vec{\mu}}_{|n}})
\right] 
        \leq
        \sum_{j \in \naturals} R( \hat{\vec{\mu}}_{|n} \in {\cal E}_j ) \max\limits_{\hat{\vec{\mu}}_{|n} \in {\cal E}_j} D( P_{\vec\mu^*} \| P_{(1-\alpha) \hat{\vec{\mu}}_{|n} + \alpha x_0})
        \\
        \overset{\text{(a)}}{=}  & \sum_{j \in \naturals} R( \hat{\vec{\mu}}_{|n} \in {\cal E}_j ) O\left( (A+j)^{m-s} \right)
             {\leq}   \sum_{j \in \naturals} 
        R\left(\| \hat{\vec{\mu}}_{|n}  \|_2 \geq A+ j-1 ) \right)
         \cdot O\left( (A+j)^{m-s} \right) \\
         \overset{\text{(b)}}{=} &  \sum_{j \in \naturals}  O \left(n^{- \lceil m/2 \rceil} \cdot \left( A+ j-1 \right)^{-m} \right) \cdot \left(  A + j   \right )^{m-s} 
       =   O(n^{ - \lceil m/2 \rceil}) \cdot   \sum_{j \in \naturals} j^{-s   } \overset{\text{(c)}}{=}  O(n^{ -  m/2 - 1/2}),
    \end{align*}
      where (a) follows from Part (c) of Condition~\ref{cond:plugin}, (b) follows from Lemma~\ref{lemma:Using Markov Inequality}, and (c) follows from the assumption $s> 1$ in Part (c) of Condition~\ref{cond:plugin}, and we use that $m$ is odd. 
    It follows that
    \begin{equation}\label{eq:inbetween}
{\mathbb E}_R\left[{\bf 1}_{\| \hat{\vec{\mu}}_{|n} - \vec{\mu}^*\|_2 \geq A}
\cdot  D(
P_{\vec\mu^*} \| P_{\breve{\vec{\mu}}_{|n}}
)
\right]  = O(n^{ -  m/2 - 1/2}).
\end{equation}
We now show how (\ref{eq:inbetween}) implies the result. A  second order Taylor approximation gives:
    \begin{align*}
        & 
        {\mathbb E}_R \left[{\bf 1}_{\| \hvm_{|n} \|_2 \geq n^{\alpha - 1/2}}  \cdot 
 n  D(
 P_{\vec\mu^*} \| P_{\breve{\vec{\mu}}_{|n}}
 )     \right] \\
\leq
&  {\mathbb E}_R \left[{\bf 1}_{n^{\alpha - 1/2} \leq \| \hvm_{|n} \|_2  \leq A}  \left( 
        \frac{1}{2} n \sup_{\vec{\mu}: \| \vec{\mu}\|_2 \leq A} \left( \left(\frac{n\hvm_{|n} + n_0 x_0}{n+n_0}\right)^\top  
        J(\vec{\mu}) \left(\frac{n {\mu} + n_0 x_0}{n+n_0}\right)\right)       \right)\right] + \\
        &
{\mathbb E}_R \left[{\bf 1}_{\| {\vec{\mu}} \|_2 > A} 
\cdot  n D(
 P_{\vec\mu^*} \| P_{\breve{\vec{\mu}}_{|n}}
 )     
\right].
\end{align*}
The second term is, from the above, $O(n^{1/2-m/2 })$. 
The first term can be bounded further, using that the supremum over $\vm$ is on a bounded set in the interior of the parameter space, so that the eigenvalues of $J(\mu)$ are all bounded on this set, and using  again Lemma~\ref{lemma:Using Markov Inequality}, 
by 
\begin{multline*}
O\left((n \cdot  A^2 R(n^{\alpha - 1/2} \leq \| \hvm_{|n} \|_2  \leq A
)\right)= O\left( n \cdot  R( \| \hvm_{|n} \|_2  \geq 
n^{\alpha - 1/2} )\right) \\
=
O(n\cdot n^{- \lceil m/2 \rceil} n^{-m (\alpha-1/2)}) = O(n^{-m \alpha + 1/2}).
\end{multline*}
The result follows.

\ownparagraph{Proof of (\ref{eq:KLboundsLeft}) and (\ref{eq:KLboundsRight})}
We first prove (\ref{eq:KLboundsLeft}). 
Combining  Proposition~\ref{prop:frombrevetohat} with the first inequality in Proposition~\ref{prop:thirdorderTaylor} (in which  we bound $|\hat{\vm}_i|$ by $n^{\alpha -1/2}$), we find that, absorbing $d$-factors into the $O(\cdot)$-notation:
\begin{align*}
&    {\mathbb E}\left[{\bf 1}_{ \| \hat{\vm}_{|n} \|_2 < n^{\alpha - 1/2} } 
    \cdot D(P_{\bf 0} \| P_{\breve{\vm}_{|n}})
    \right] =
 {\mathbb E}\left[{\bf 1}_{ \| \hat{\vm}_{|n} \|_2 < n^{\alpha - 1/2} } 
    \cdot D(P_{\bf 0} \| P_{\hat{\vm}_{|n}}) + O\left(n^{\alpha - 3/2}\right)
    \right]  =  \\
   & {\mathbb E}\left[{\bf 1}_{ \| \hat{\vm}_{|n} \|_2 < n^{\alpha - 1/2} } 
    \cdot \left( \frac{1}{2} \vm^{\top} J \vm+ O\left( n^{\alpha -1/2} \sum_{j,k \in [d]} |\hat{\vm}_j | |\hat{\vm}_k| \right) \right) + O\left(n^{\alpha - 3/2}\right)
   \right] = \\& 
   {\mathbb E}\left[{\bf 1}_{ \| \hat{\vm}_{|n} \|_2 < n^{\alpha - 1/2} } 
    \cdot \frac{1}{2} \vm^{\top} J \vm+ \left( O\left( n^{\alpha -1/2} \max_{k \in [d]} |\hat{\vm}_k |^2 \right)  \right)+ O\left(n^{\alpha - 3/2}\right)
   \right] =  \\& {\mathbb E}\left[{\bf 1}_{ \| \hat{\vm}_{|n} \|_2 < n^{\alpha - 1/2} } 
    \cdot \frac{1}{2} \vm^{\top} J \vm+  O\left( n^{\alpha -3/2} \right) \right], 
\end{align*}
where in the final equality we used Lemma~\ref{lemma:lceil_rceil}. This proves (\ref{eq:KLboundsLeft}). The proof of (\ref{eq:KLboundsRight}) is analogous (now only using Proposition~\ref{prop:thirdorderTaylor}, with the components of $D(\cdot \| \cdot)$ interchanged, and without the need to `match' $\breve{\vm}_{|n}$ with $\hat{\vm}_{|n}$); we omit further details.

\subsection{Proof for COND part}
\label{app:condproof}
Throughout the proof, we adopt the notations of (\ref{eq:notations}). 

We write $p^{\circ}_{\vm^*}$ and  $q^{\circ}$ as the densities of $Z = n^{1/2} \hat{\vm}_{|n}$ under $P_{\vm^*}$ and $Q=Q_{\vm^*}$, respectively. 
\begin{lemma}\label{lem:condsimpleH1}
    Under the regularity Condition~\ref{cond:cond} (`\ccond'), we have:
    \begin{equation}\label{neq:Lower_bound_CLT}
        {\mathbb E}_R \left[\log \frac{p^{\circ}_{\vec{\mu}^*}( \sqrt{n} \hvm_{|n})}{q^{\circ}( \sqrt{n} \hvm_{|n})}
\right] \geq - D_{\Sigma_r}(\Sigma_q(\vm^*) \| \Sigma_p(\vm^*)  )
+ o(1). 
    \end{equation}
\end{lemma}
\begin{proof}
For the continuous case ($X$ has Lebesgue density), the proof is based on the following immediate corollary of Theorem 19.2 of Bhattacharya and Rao \cite{BhattacharyaR76} (page 192): consider an  i.i.d. sequence of random vectors $X, X_1, X_2, \ldots$, with $X \sim \Psi$ with ${\mathbb E}_{\Psi}[X] = {\bf 0}$ where $\Psi$ has  bounded continuous (Lebesgue) density  and a moment generating function (in particular, all moments of each component of $X$ exist). Then we have for $\psi$ the density of $n^{-1/2} (\sum_{i=1}^n X_i)$ that, for all integers $t \geq 1$, 
uniformly for all $y \in \reals^d$,  
\begin{align}\label{eq:bhatta}
1 + \sum_{j=1}^{t} n^{-j/2} f_j(y) -  o\left(\frac{n^{-t/2}}{\phi_\Sigma(y)} \right)  \leq     \frac{\psi(y)}{\phi_\Sigma(y)} \leq 1 + \sum_{j=1}^t n^{-j/2} f_j(y) + o\left(\frac{n^{-t/2}}{\phi_\Sigma(y)} \right),  
\end{align}
\commentout{
\begin{align}\label{eq:bhatta}
1 +  n^{-1/2} f(y) -  o\left(\frac{n^{-1/2}}{\phi_\Sigma(y)} \right)  \leq     \frac{\psi(y)}{\phi_\Sigma(y)} \leq 1 + n^{-1/2} f(y) + o\left(\frac{n^{-1/2}}{\phi_\Sigma(y)} \right),  
\end{align}}
where $\phi_{\Sigma}$ is the density of a normal distribution with mean $0$ and covariance matrix $\Sigma$ and the
$f_j: \reals^d \rightarrow \reals$ are specific  $3j$-degree 
degree polynomials in the components of $y$ (to see how this follows using Theorem 19.2. of \cite[page 192]{BhattacharyaR76} note that what we call $t$  is, in their notation, $s-2$,  
note that their $P_0(- \phi(0,V): \{ \chi_v\}) = \phi_{0,V}$ is the density of a normal distribution with mean ${\bf 0}$ and covariance matrix $V$, and their 
$P_1(- \phi(0,V): \{ \chi_v\})(y) = f(x) \phi_{0,V}$ \cite[page 55--56]{BhattacharyaR76}); their precondition on the characteristic function is equivalent to our condition of bounded Lebesgue density by their Theorem 19.1 and their precondition on moments being finite holds automatically because we assume $Q$ has a moment generating function.  

Below we only give the proof for the continuous case, based on (\ref{eq:bhatta}); the proof for the discrete (lattice) case goes in exactly the same way, but now we use Theorem 22.1 of \cite{BhattacharyaR76}, which is their analogue of  Theorem 19.2 for the discrete, lattice case; we omit details and continue with the continuous case.  

We again adopt the notations (\ref{eq:notations}), i.e. we set $\vm^* = {\bf 0}$ 
and we write $\Sigma_p := \Sigma_p(\vm^*) = \Sigma_p({\bf 0})$.
For fixed $A > 0$, we define 
\begin{align*}
    \cA_n & = \left\{ x^n: \| \hvm_{|n} \|_2 \leq \sqrt{\frac{A \log n }{n}} \right\} \\
    \cB_n & = \{ x^n: \max\{ \|X_i\|_2: i \in [n] \}  \leq n \} \\
    \bar{\cB}_{n,j} & = \{ x^n:  n+ j -1< \max\{ \|X_i \|: i \in [n] \}  \leq n +j  \}     , 
\end{align*}
and $\bar{\cA}_n,\bar{\cB}_n$ their respective complements. 
The expectation in (\ref{neq:Lower_bound_CLT})  can be rewritten as
\begin{align}\label{eq:threeterms}
{\mathbb E}_R \left[{\bf 1}_{\cA_n} \log \frac{p^{\circ}_{\vec{\mu}^*}( \sqrt{n} \hvm_{|n})}{q^{\circ}( \sqrt{n} \hvm_{|n})}
\right]+ {\mathbb E}_R \left[{\bf 1}_{\bar{\cA}_n \cap \cB_n}\log \frac{p^{\circ}_{\vec{\mu}^*}( \sqrt{n} \hvm_{|n})}{q^{\circ}( \sqrt{n} \hvm_{|n})} 
\right] + {\mathbb E}_R \left[{\bf 1}_{\bar{\cA}_n \cap \bar{\cB}_n}\log \frac{p^{\circ}_{\vec{\mu}^*}( \sqrt{n} \hvm_{|n})}{q^{\circ}( \sqrt{n} \hvm_{|n})} 
\right].
\end{align}
We now use the local central limit theorem 
with expansion as in (\ref{eq:bhatta}) above to analyze (minus) the first term.
Since $\Sigma_\nuli$ is positive definite, we find
$$
1/\phi_{\Sigma_\nuli}(y) = O\left(
\exp\left(\frac{1}{2} \|y\|^2_2 \lambda_{\max, p }\right)
\right),
$$
where $\lambda_{\max,p}$ is the maximum eigenvalue of $\Sigma_p^{-1}$.  
Setting $y= n^{-1/2} \sum_{i=1}^n x_i$ 
and plugging in $\| y\| \leq \sqrt{A \log n}$
which will hold on the set ${\cal A}_n$ (note scaling takes care of $1/\sqrt{n}$)  we find that 
$
\frac{n^{-t/2}}{\phi_{\Sigma_\nuli}(y)} = o(1)   
$
if $t \geq A \lambda_{\max,p}$. The same derivation holds with $\Sigma_q$ instead of $\Sigma_p$, so we get 
$$
\frac{n^{-t/2}}{\phi_{\Sigma_\nuli}(y)}  = o(1), 
\frac{n^{-t/2}}{\phi_{\Sigma_\alti}(y)}  = o(1),
$$
if $t \geq A \cdot  \max \{ \lambda_{\max,p},\lambda_{\max,q} \}$.
We then find that the remainder terms $n^{-t/2} f_j(y)$ are all $O( (\log n)^{3j}/n)^{j/2})= o(1)$ on all $y$ corresponding to $x^n$ in $\cA_n$. 
We can thus set  $\psi$ to $p^{\circ}$ and use the left inequality in (\ref{eq:bhatta}) with this $t$ and then set $\psi$ to $q^{\circ}$ and using the right inequality with the same $t$ 
to get
\begin{align}\label{eq:firsttermahoy}
{\mathbb E}_R \left[{\bf 1}_{\cA_n} \log \frac{q^{\circ}( \sqrt{n} \hvm_{|n})}{p^{\circ}_{\vec{\mu}^*}( \sqrt{n} \hvm_{|n})}
\right]
\leq&  {\mathbb E}_R \left[{\bf 1}_{\cA_n} \log \frac{\phi_{\Sigma_\alti}(\sqrt{n} \hvm_{|n})(1 + o(1))}{\phi_{\Sigma_\nuli} (\sqrt{n} \hvm_{|n}) (1- o(1))} \right]\nonumber\\
=& 
 D_{\Sigma_r}(\Sigma_q(\vm^*) \| \Sigma_p(\vm^*)  )
 + o(1), 
\end{align}
where the final equality follows because $Q(\cA_n) \rightarrow 1$, which follows because we assume that $Q$ has a moment generating function.  This deals with the first term in (\ref{eq:threeterms}).

Now consider the second term in (\ref{eq:threeterms}). 
Since we assume $R$ has a moment generating function, we can with some work (details omitted) employ the Cram\`er-Chernoff method to get that
\commentout{
Fix $B > 0$ and fix $C > 0$ such that the hypercube $[-C,C]^d$ is contained in the mean-value parameter space. Such a $C$ exists because the mean-value parameter space $\meanspace_p$ is open. 
Consider the $j$ components $\hat{\vec{\mu}}_j$, $j \in [d]$, separately.  By the Cramer-Chernoff method, we have for each $j$, $n$ (where in the notation we omit the dependence of $\hat{\vec{\mu}}, \vec{\mu}^{\circ}$ on $n$), 
\begin{align*}
& R( \| \hat{\vec{\mu}}_j \|^2 \geq c_n, \hat{\mu}_j > 0) 
= R( \frac{1}{2} I(\vmc_j) \| \hat{\vec{\mu}}_j \|^2 \geq c_n
\frac{1}{2} I(\vmc_j), \hat{\mu}_j > 0)  = \\
& R( D( \hat{\vec{\mu}}_j \| 0) \geq c_n
\frac{1}{2} I(\vmc_j), \hat{\mu}_j > 0) \leq  \exp\left(
- n  c_n
\frac{1}{2} I(\vmc_j)
\right) \leq \exp\left(
- n  c_n
\frac{1}{2} I_{n,j}
\right)  \\
& = O\left( n^{- B}
\right),
\end{align*}
where we set $c_n = \min \{C,2 (\log n^B)/(n I_{j})\} $ $I(\vmc_j)$ is the Fisher information at $\vmc_j$ of the 1-dimensional family TOD and $I_{j} = \min \{I(\vm_j): \| \vm_j\| \leq C \}$. 

Repeating the reasoning $2d$ times, once for each $j \in [d]$ and for the case $\hvm_j >0$ and $\hvm_j < 0$, and  using the union bound we find that, for all $A > 0$, all $B > 0$,
for all sufficiently large $n$,
$$
R \left( \| \hat{\vec{\mu}}_{|n} \|^2 \geq 
\left( \frac{2 B}{n \underline{I}} \cdot \log n \right)
\right) = O(n^{-B}).  
$$
where $\underline{I} = \min \{I_1, \ldots, I_d \}$.
}
for all $B> 0$, if we take $A$ sufficiently large, we get: 
\begin{equation}\label{eq:friday}
R \left( \| \hat{\vec{\mu}}_{|n} \|_2 \geq 
 \sqrt{\frac{A \log n}{n}} \right) = O(n^{-B}).  
\end{equation}
Now fix $B= a+2$, with $a$ the exponent in Condition~\ref{cond:cond},  and choose $A$ large enough so that (\ref{eq:friday}) holds and then $t$ large enough so that (\ref{eq:firsttermahoy}) also holds.  
For the second term we then get, 
\begin{align}\label{eq:second_term}
  &  {\mathbb E}_R \left[{\bf 1}_{\bar{\cA}_n \cap \cB_n}\log \frac{p^{\circ}_{\vec{\mu}^*}( \sqrt{n} \hvm_{|n})}{q^{\circ}( \sqrt{n} \hvm_{|n})} 
\right] = O(n^{-a-2}) \cdot \sup_{x^n\in \bar{\cA}_n \cap \cB_n } \log \frac{p_{\vm^*}(x^n)}{q(x^n)} = 
\\
=& O(n^{-a-2}) \cdot \sup_{x^n\in \bar{\cA}_n \cap \cB_n } \sum_{i=1}^n \log \frac{p_{\vm^*}(x_i)}{q(x_i)} 
\leq O(n^{-a-2}) \cdot n \cdot  \sup_{x^n\in \cB_n }
\max_{i\in [n]} \log \frac{p_{\vm^*}(x_i)}{q(x_i)} 
\nonumber \\
=& \nonumber
O(n^{-1}),
\end{align}
where we in the penultimate equality we used $x^n \in {\cal B}_n$ and Condition~\ref{cond:cond}. The first equality follows because with ${\cal G}_{\epsilon}(V^n) = \{ x^n: | \sqrt{n} ( \hat{\vm}_{|v^n} - \hat{\vm}_{|x^n}) | \leq \epsilon \}$, we have: 
\begin{align}\label{eq:epsilontozero}
 &   \sup_{v^n\in \bar{\cA}_n \cap \cB_n } 
    \frac{p^{\circ}_{\vec{\mu}^*}( \sqrt{n} \hat{\vec{\mu}}_{|v^n})}{q^{\circ}( \sqrt{n} \hat{\vec{\mu}}_{|v^n})}  = 
\sup_{v^n\in \bar{\cA}_n \cap \cB_n } 
    \lim_{\epsilon \downarrow 0} 
    \frac{ \epsilon^{-1}  \int_{{\cal G}_{\epsilon}(v^n)} p_{\vec{\mu}^*}(x^n ) d x^n}{
    \epsilon^{-1} \int_{{\cal G}_{\epsilon}(v^n)} q(x^n )
    d x^n} \\
& =   \sup_{v^n\in \bar{\cA}_n \cap \cB_n } 
    \lim_{\epsilon \downarrow 0}  
     \frac{  \int_{{\cal G}_{\epsilon}(v^n)} p_{\vec{\mu}^*}(x^n ) d x^n}{
   \int_{{\cal G}_{\epsilon}(v^n)} q(x^n )
    d x^n}  = \sup_{v^n\in \bar{\cA}_n \cap \cB_n } 
    \frac{p_{\vm^*}(v^n)}{q(v^n)},\nonumber
\end{align}
where the final inequality follows by the assumed continuity of $p_{\vm^*}$ and $q$. 

(\ref{eq:second_term}) gives the second term. For the third term we get, with again $a$ the exponent in Condition~\ref{cond:cond} and using, in the first inequality below, analogously to (\ref{eq:epsilontozero}):
\begin{align*}
   &  {\mathbb E}_R \left[{\bf 1}_{\bar{\cA}_n \cap \bar{\cB}_n}\log \frac{p^{\circ}_{\vec{\mu}^*}( \sqrt{n} \hvm_{|n})}{q^{\circ}( \sqrt{n} \hvm_{|n})} 
\right] \leq 
\sum_{j \in \naturals} R(\bar{\cA}_n, \bar{\cB}_{n,j}) \cdot n \cdot (n +j)^a  \nonumber \\
\leq 
 & \sum_{j \in \naturals} R(\exists i \in [n]: |X_i| \geq n+j-1) \cdot n \cdot (n +j)^a  
\leq
 n^2 \sum R(|X| \geq n + j-1) (n +j)^a 
\\ \leq & 2 n^{2} \min_k {\mathbb E}_R[|X|^k]\sum_{j \in \naturals} (n + j-1)^{-k}  (n+j)^a    \\
\leq &
2 \min_k {\mathbb E}_R[|X|^k]\sum_{j \in \naturals} (n + j-1)^{-k}  (n+j)^{2+a} = O  
\left( n^{-1} \right)
\end{align*}
%
where the penultimate equality follows by the fact that we can employ Markov's inequality with $X^k$, for all $k$, and the final equality is obtained by setting $k=a+4$. This proves (\ref{neq:Lower_bound_CLT}) and hence (\ref{eq:gencond}). The statement right below (\ref{eq:gencond}) follows by the fact that the role of $p_{\vm^*}$ and $q$ becomes completely symmetric if Condition~\ref{cond:cond} (`\ccond') holds with $q$ and $p_{\vm^*}$ interchanged.  
\end{proof}

\subsection{Proof of Lemma~\ref{prop:difficult}, underlying Anti-Simple case part of Theorem~\ref{thm:simpleH1general}}
\label{app:difficult}

\begin{proof}
We use the notations as summarized in (\ref{eq:notations}). 
Fix $0 < \alpha < 1/6$. Without loss of generality, let $\vec{\mu}^*=0$. The integral below is over $\meanspace_p$. 
Since $n$ remains fixed throughout the proof, we abbreviate $W_{0,(n)}$ to $W$ and $w_{0,(n)}$ to $w$. Let $S_{Q,\rip'}^{(n)} = \frac{q(\sn{U}{n})}{p_{W}(\sn{U}{n})}$.
It is easy to get that
\begin{align*}
{\mathbb E}_{\sn{U}{n} \sim Q} [\log S_{Q,\rip'}^{(n)}]
= nD(Q||P_{\vec{\mu}^*})
+
{\mathbb E}_{\sn{U}{n} \sim Q} \left[\log \frac{p_{\vec{\mu}^*}(\sn{U}{n})}{\int w(\vec{\mu}) p_{\vec{\mu}}(\sn{U}{n}) d\vec{\mu} } \right],
\end{align*}
so we just need to focus on the second term.
By (\ref{eq:bayesbound}) in Lemma~\ref{lem:KLbounds},
 we can write it as:
\begin{align}\label{eq:fasttrain}
& {\mathbb E}_{ Q} \left[\log \frac{p_{\vec{\mu}^*}(\sn{U}{n})}{\int w(\vec{\mu}) p_{\vec{\mu}}(\sn{U}{n}) d\vec{\mu} } \right]
= 
{\mathbb E}_{Q} \left[{\bf 1}_{\| \hvm_{|n}\|_2 < n^{\alpha -1/2} } \cdot \log \frac{p_{\vec{\mu}^*}(\sn{U}{n})}{\int w_n(\vec{\mu}) p_{\vec{\mu}}(\sn{U}{n}) d\vec{\mu} } \right] + o(1)
= \\ \nonumber 
& {\mathbb E}_{Q} \left[{\bf 1}_{\| \hvm_{|n}\|_2 < n^{\alpha -1/2} } \cdot \log \frac{p_{\vec{\mu}^*}(\sn{U}{n})}{p_{\hvm_{|n}}(\sn{U}{n})} \right] 
+
{\mathbb E}_{Q} \left[{\bf 1}_{\| \hvm_{|n}\|_2 < n^{\alpha -1/2} } \cdot \log \frac{p_{\hvm_{|n}}(\sn{U}{n})}{\int w(\vec{\mu}) p_{\vec{\mu}}(\sn{U}{n}) d\vec{\mu} } \right]
+ o(1) = \\ & \nonumber
-\frac{1}{2} \textsc{tr}\left(\Sigma_\alti \Sigma_\nuli^{-1} \right)
 + {\mathbb E}_{Q} \left[{\bf 1}_{\| \hvm_{|n}\|_2 < n^{\alpha -1/2} } \cdot \log \frac{p_{\hvm_{|n}}(\sn{U}{n})}{\int w(\vec{\mu}) p_{\vec{\mu}}(\sn{U}{n}) d\vec{\mu} } \right]
+
o(1),
\end{align}
where the final equality is obtained by Lemma \ref{lem:KLbounds}, (\ref{eq:KLboundsLeft}). 
So we just need to upper bound the expectation in (\ref{eq:fasttrain}). 
For this, we fix any  $\beta$ with $\alpha < \beta < 1/6$ and note that, uniformly for all $U^n \in \cU^n$ with $\|\hvm\|_2 \leq n^{\beta - 1/2}$, we have:
\begin{align}
&   -  \log \int w(\vec{\mu}) \frac{p_{\vec{\mu}}(\sn{U}{n}) }{p_{\hvm_{|n}}(\sn{U}{n})} d\vec{\mu}  =  -  \log \int w(\vec{\mu}) \exp\left( -n D(P_{\hat{\vm}_{|n}} \| P_{\vm}) \right)  d\vec{\mu}  \leq \nonumber \\
&  -  \log \int_{\vm: \| \vm \|_2 \leq n^{\beta - 1/2}} w(\vec{\mu}) \exp\left( -n D(P_{\hat{\vm}_{|n}} \| P_{\vm}) \right)  d\vec{\mu}  \leq \nonumber \\
& -  \log \int_{\vm: \| \vm \|_2 \leq n^{\beta - 1/2}}  w(\vec{\mu}) \exp\left( -\frac{n}{2} (\vm - \hvm_{|n})^{\top} J (\vm - \hvm_{|n}) 
+  O\left(n \cdot n^{3 \beta -3/2}\right) \right)  d\vec{\mu} \leq 
\nonumber \\ \nonumber 
& -  \log \int_{\vm: \| \vm \|_2 \leq n^{\beta- 1/2}} w(\vec{\mu}) \exp\left( -\frac{n}{2} (\vm - \hvm_{|n})^{\top} J (\vm - \hvm_{|n})  \right)  d\vec{\mu} + o(1) \leq \\ 
&\frac{d}{2} \log \frac{2\pi}{n}
- \frac{1}{2} \log \det K \nonumber \\ & - \log 
\int_{\vm: \| \vm \|_2 \leq n^{\beta- 1/2}}  \exp\left(- \frac{n}{2} \left( \vm^{\top} K \vm  + (\vm - \hvm_{|n})^{\top} J (\vm - \hvm_{|n}) \right) \right)  d\vec{\mu} + o(1),\label{eq:integral}
\end{align}
where we first used the robustness property of exponential families (Lemma~\ref{lem:KLbounds}, (\ref{eq:fullrobustnessb})) and then for the second inequality, we used the second inequality in Proposition~\ref{prop:thirdorderTaylor} (with the $\alpha$ in that proposition set to $\alpha \vee \beta = \beta$). The third inequality uses  that $\beta$ was set $< 1/6$,
and in the final line we used
the definition of $w$ as the density of a Gaussian with mean ${\bf 0}$ and covariance $K^{-1}/n$. By a little computation, we can rewrite the expression in the exponent in the integral in (\ref{eq:integral}) as
\begin{align*}
\vec{\mu}^\top K \vec{\mu} + (\vec{\mu} - \hvm_{|n})^\top \Jp (\vec{\mu} - \hvm_{|n})
=
(\vec{\mu} - \vec{m})^\top (K + \Jp) (\vec{\mu} -  \vec{m})
+ C,
\end{align*}
where $\vec{m} = (K + \Jp)^{-1} \Jp\hvm_{|n}$ and $C = \hvm_{|n}^\top \Jp \hvm_{|n} - \vec{m}^\top (K+\Jp) \vec{m}$. 
which can be simplified  to
$C= \hvm_{|n}^\top \Sigma_q^{-1}  \hvm_{|n}$ by first noting
$C = \hvm_{|n}^\top (\Jp - \Jp (K + \Jp)^{-1} \Jp) \hvm_{|n}$ and then \begin{align}
& \Jp - \Jp(K+ \Jp)^{-1} \Jp =
\Jp \left(I - (K+ \Jp)^{-1} \Jp  \right)= 
\Jp \left(I - (I +\Jp^{-1} K)^{-1} \right)= \nonumber \\
&
\Jp \left(I - (I- \Jp^{-1}(I+ K \Jp^{-1} )^{-1} K) \right) = 
\Jp \left( \Jp^{-1} (I+ K \Jp^{-1} )^{-1} K\right) = \nonumber
(K^{-1} + \Jp^{-1})^{-1} =
\Sigma_q^{-1},
\end{align}
\commentout{
\begin{align}& K - K(K+ \Jp)^{-1} K =
K \left(I - (K+ \Jp)^{-1} K  \right)= 
K \left(I - (I +K^{-1} \Jp)^{-1}   \right)= \nonumber \\
&
K \left(I - (I- K^{-1}(I+ \Jp K^{-1} )^{-1} \Jp) \right) = 
K \left( K^{-1} (I+ \Jp K^{-1} )^{-1} \Jp\right) = \nonumber
(\Jp^{-1} + K^{-1})^{-1} =
\Sigma_q^{-1},
\end{align}}
where the third equality follows by the reduced Woodbury matrix identity (see e.g. wikipedia). 
It follows that (\ref{eq:integral}) can be further rewritten as 
\begin{align}\label{eq:integralb}
&  
\frac{d}{2} \log \frac{2\pi}{n}
- \frac{1}{2} \log \det K 
+ \frac{n}{2} \cdot \hvm_{|n}^\top \Sigma_q^{-1}  \hvm_{|n} \\ & 
- \log
\int_{\vm: \| \vm \|_2 \leq n^{\beta- 1/2}}  \exp\left(- \frac{n}{2} \left( (\vm - \vec{m})^{\top} (K+J) (\vm - \vec{m}) \right) \right)  d\vec{\mu} + o(1), \nonumber
\end{align}
where $\vec{m} = (K + \Jp)^{-1} \Jp\hvm_{|n}$ does not depend on $\vm$ and has norm $\| \vec{m} \|_2 = O(\|\hat{\vm}_{|n}\|_2)= O(n^{\alpha -1/2})$ (recall that we are only evaluating the integral for values of $\hat{\vm}_{|n}$ with $\|\hat{\vm}_{|n}\|_2 < n^{\alpha - 1/2}$). Now note (a) the distance between $\vec{m}$ and the boundary of the set over which we integrate, $\{ \vm: \| \vm \|_2 \leq n^{\beta- 1/2} \}$ is therefore of order $n^{\beta-1/2} - n^{\alpha - 1/2}= n^{-1/2} g(n)$ for a function $g(n)$ with $\lim_{n \rightarrow \infty} g(n) = \infty$ (recall we chose $\beta$ larger than $\alpha$) ; and (b) since $K$ and $J$ are inverses of positive definite matrices, they are themselves positive definite and so is $K+J$. Therefore the integral in (\ref{eq:integralb}) converges, with increasing $n$, to a Gaussian integral with covariance $(K+J)/n$, so we can rewrite (\ref{eq:integralb}) as:
\begin{align*}
&  
\frac{d}{2} \log \frac{2\pi}{n}
- \frac{1}{2} \log \det K 
+ \frac{n}{2} \cdot \hvm_{|n}^\top \Sigma_q^{-1}  \hvm_{|n} \nonumber 
- \frac{d}{2} \log \frac{2\pi}{n}
+  \frac{1}{2} \log \det (K+J) + o(1)
\\ 
=& \frac{1}{2} \log \frac{\det (K+J)}{\det K} + \frac{n}{2} \cdot \hvm_{|n}^\top \Sigma_q^{-1}  \hvm_{|n} +o(1),
\end{align*}
and then plugging this back into (\ref{eq:fasttrain}) we see that 
(\ref{eq:fasttrain})  can be bounded as 
\begin{align*}
 -\frac{1}{2} \textsc{tr}\left(\Sigma_\alti \Sigma_\nuli^{-1} \right)
 + 
 \frac{1}{2} \log \frac{\det (K+J)}{\det K}  + 
\frac{n}{2} \cdot   {\mathbb E}_{Q} \left[{\bf 1}_{\| \hvm_{|n}\|_2 < n^{\alpha -1/2} } \cdot  \hvm_{|n}^\top \Sigma_q^{-1}  \hvm_{|n} 
 \right] +
o(1),   
\end{align*}
so that , by (\ref{eq:mlesquare}) and (\ref{eq:KLboundsSquare}) (recall we set $\vm^* = {\bf 0}$), (\ref{eq:fasttrain}) can be further bounded as
\begin{align}\label{eq:final}
    &\frac{1}{2} \log 
    \frac{\det (K+J)}{\det K} + \frac{d}{2} - \frac{\tr\left(\Sigma_\alti \Sigma_\nuli^{-1} \right)}{2} + o(1)\\
    =& \frac{1}{2} \log \det (\Sigma_p^{-1} \Sigma_q )  + \frac{d}{2} - \frac{\tr\left(\Sigma_\alti \Sigma_\nuli^{-1} \right)}{2} + o(1)\nonumber
    = - D_{\gauss}(\Sigma_q \Sigma_p^{-1}) + o(1).
\end{align}
The penultimate equality holds by direct computation, using the definitions of $K$ and $J$.  But (\ref{eq:final}) implies the result.
\end{proof}
\subsection{Proof of (\ref{eq:seqbad}) in Theorem~\ref{thm:compositeH1general}}
\label{app:noncompetitive}
Let $\cA_{\gamma} := \{\vm' \in \meanspace_p \cap \meanspace_q: \| \vm' - \vm^* \|_2 \leq \gamma \}$.

Define  $S_{Q_{\vm'}}^{(1)}:= \frac{q_{\vm'}(U)}{
    p_{ \leftsquigarrow q_{\vm'}(U)}(U)
    }$ and note that $S_{Q_{\vm'}}^{(1)}$ is an e-variable for all $\vm' \in \cA_{\gamma}$, and in particular $S_{Q_{\vm^*}}^{(1)}$ is the GRO (optimal) e-variable relative to $Q_{\vm^*}$. Therefore, we must have, for some $\epsilon > 0$, that 
$$
{\mathbb E}_{Q_{\vm^*}} \left[ \log \frac{q_{\vm'}(U)}{
    p_{ \leftsquigarrow q_{\vm'}(U)}(U)} \right]= 
{\mathbb E}_{Q_{\vm^*}} \left[\log S^{(1)}_{Q_{\vm'}} \right]
\leq {\mathbb E}_{ Q_{\vm^*}} \left[\log S^{(1)}_{Q_{\vm^*}} \right]
=  {\mathbb E}_{Q_{\vm^*}} \left[ \log \frac{q_{\vm^*}(U)}{
    p_{ \vm^*}(U)} \right] - \epsilon,
$$    
where the final equality follows because, by the same reasoning as in the proof of Theorem~\ref{thm:simpleH1general}, Part 3, we have (\ref{eq:afterarxiv}), with $Q_{\vm^*}$ in the role of $Q$.
But, taking expectations over $U_1, U_2, \ldots \text{i.i.d.} \sim Q_{\vm^*}$, from the above we immediately get, using Fubini's theorem, 
\begin{align*}
&     {\mathbb E}_{U^n \sim Q_{\vm^*}}
     \left[ \log \prod_{i=1}^n 
     \frac{q_{\breve{\vm}_{|i-1}}(\bar{U}_i)}
{p_{ \leftsquigarrow q_{\breve{\vm}_{|i-1}}(\bar{U}_i)}}
     \right] = 
     \sum_{i=1}^n {\mathbb E}_{U^{i-1} \sim Q_{\vm^*}}
     {\mathbb E}_{U \sim Q_{\vm^*}}\left[
\log 
     \frac{q_{\breve{\vm}_{|i-1}}(U)}
{p_{ \leftsquigarrow q_{\breve{\vm}_{|i-1}}(U)}
(U)}
     \right]  \\ \leq& n (D(Q_{\vm^*} \| P_{\vm^*}) - \epsilon),
\end{align*}
and the result is proved. 

\section{Additional details for Section~\ref{sec:general}: checking {\bf UI} and {\bf plug-in} regularity conditions for example families}
\subsection{Condition {\bf UI} for the Poisson model}\label{app:checking}

\begin{example}\label{example:Poisson}{\bf [Poisson]}
{\rm Let $\nulhyp$ be the Poisson family, given in its standard parameterization as $p_{\mu}(x) = P_\mu(X=k) = \frac{\mu^k \cdot e^{-\mu}}{k!}$ with parameter $\mu > 0$, and suppose $X$ has $m$ moments under $R$. Standard and straightforward computations show that this is an exponential family with mean-value parameter $\mu = {\mathbb E}_{P_{\mu}}[X]$ and
%
$D(P_{\hat{\mu}_{|n}} \| P_{{\mu}^*} ) = \hat{\mu}_{|n} \log \frac{\hat{\mu}_{|n}}{\mu^*} + \mu^* - \hat{\mu}_{|n}$. 
Fix $0 < \gamma < 1/2$ and define $ {\cal E}_j =  \{{\mu} \in (0, \infty): | \mu - \mu^*| \in [n^{-\gamma} + j-1, n^{-\gamma} + j]\}$. We have
\begin{align*}
& {\mathbb E}_R\left[{\bf 1}_{\| \hat{{\mu}} - {\mu}^*\|_2 \geq n^{-\gamma}}
\cdot  D(P_{\hat{{\mu}}} \| P_{{\mu}^*} )
\right] 
        \leq \sum_{j \in \naturals^+} R( \hat{\mu}_{|n} \in {\cal E}_j )  \max\limits_{\hat{\mu}_{|n} \in {\cal E}_j} D(P_{\hat{\mu}_{|n}} \| P_{{\mu}^*})
        \\
        = & \sum_{j \in \naturals^+} 
        R( \hat{\mu}_{|n} \in {\cal E}_j ) \max\limits_{\hat{\mu}_{|n} \in {\cal E}_j} \left(\hat{\mu}_{|n} \log \frac{\hat{\mu}_{|n}}{\mu^*} + \mu^* - \hat{\mu}_{|n} \right)\\
             {\leq} &
             R\left(| \hat{\mu}_{|n} - \mu^*| \geq n^{-\gamma}\right) \cdot \left((-\log {\mu^*}) \cdot O(1) + \mu^* - O(1) \right)\\
             &+ \sum_{j \in \naturals^+} 
        R\left(| \hat{\mu}_{|n} - \mu^*| \geq n^{-\gamma} + j \right)
         \cdot O\left( (n^{-\gamma}+j + 1)\log(n^{-\gamma}+j + 1) \right) \\
         \overset{\text{(a)}}{=} &  O\left(n^{- \lceil m/2 \rceil} \cdot n^{\gamma m - \gamma} \right) + \sum_{j \in \naturals^+}  O\left(n^{- \lceil m/2 \rceil} \cdot \left(n^{-\gamma} + j \right)^{-m} \right) \cdot O\left( j \log j\right )\\
       =& O\left(n^{- \lceil m/2 \rceil} \cdot n^{\gamma m - \gamma} \right) + O(n^{- \lceil m/2 \rceil}) \cdot \sum_{j \in \naturals^+} j^{1-m} \log j
       := f(n, m, \gamma),
    \end{align*}
where (a) follows from Lemma~\ref{lemma:Using Markov Inequality} in Appendix~\ref{app:preparinggeneraltheoremsB}.
Plugging in any $m\geq 3$ and $\gamma=1/3$ we find that $f(n,m,\gamma) = O(n^{-4/3})= o(n^{-1})$, which proves the result. 
}
\end{example}
\begin{figure}
    \centering
    \subfigure[Gaussian mean parameter space]{\includegraphics[width=0.45\linewidth]{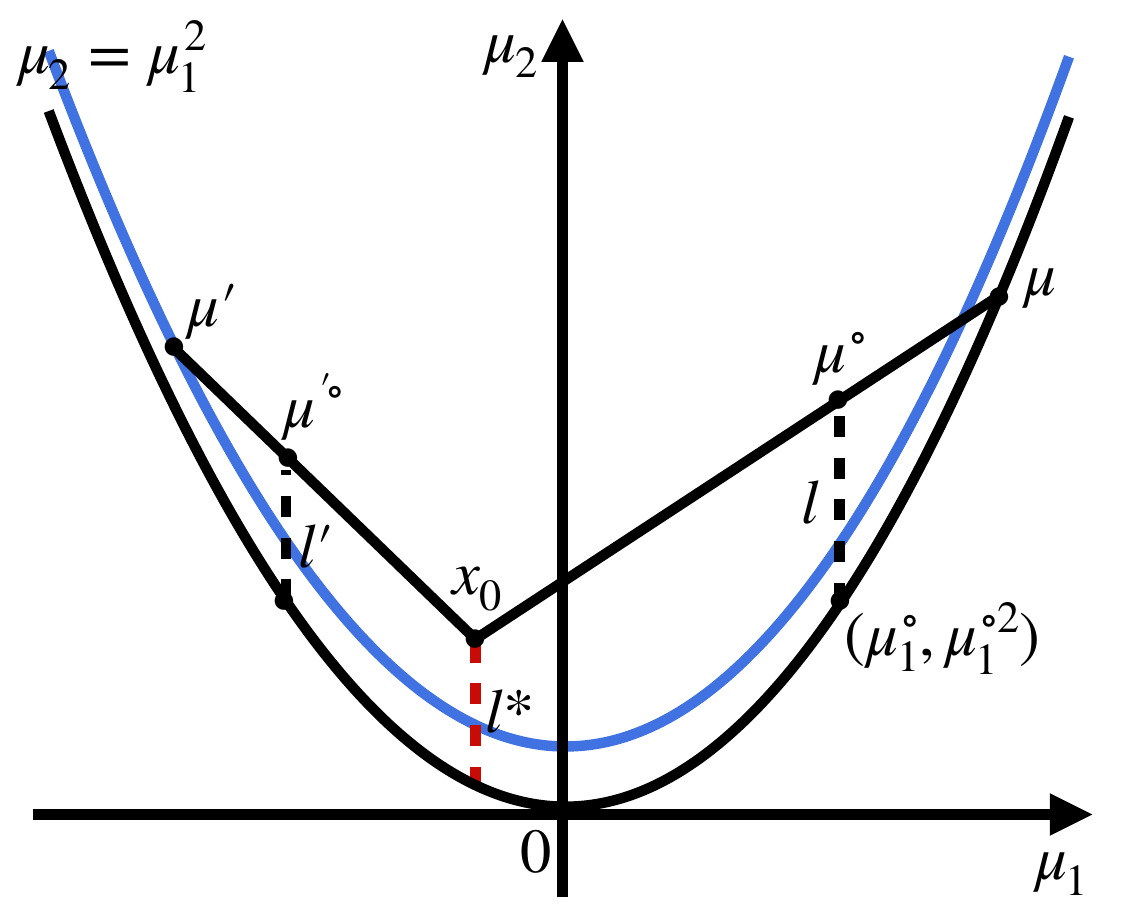}}
    \subfigure[Gamma mean parameter space]{\includegraphics[width=0.45\linewidth]{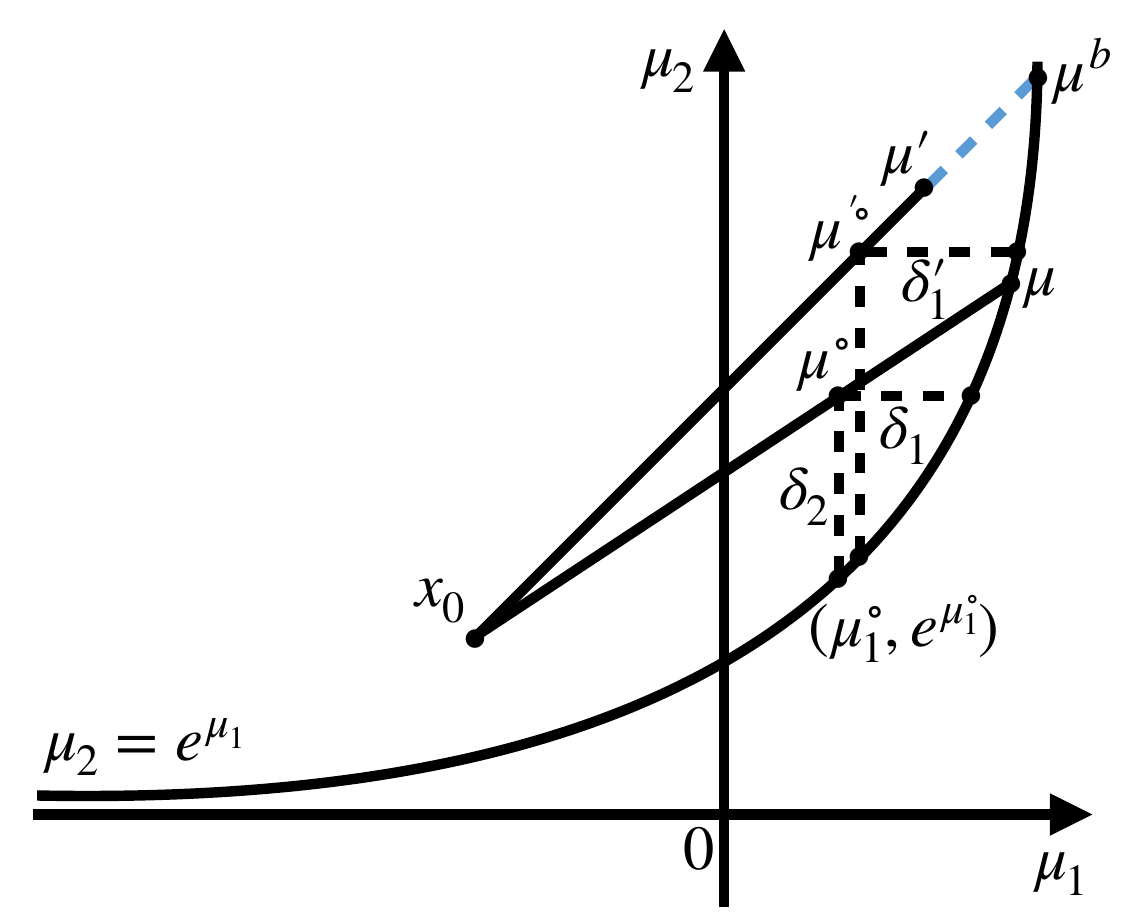}}
    \caption{In the left figure, the area above the black curve $\mu_2 = \mu_1^2$ is the Gaussian mean parameter space $\meanspace_q$. The blue curve is obtained by vertically shifting the black one upwards such that $\vec{\mu}'$ is located on it. (\ref{eq:inf_variance}) expresses that  $\alpha l^*$ is not larger than any other $l$ for every $\vec{\mu}$, $l'$ for every $\vec{\mu}'$ with $x_0 \in \meanspace_q$ and $\alpha \in (0,1)$.}
    \label{fig:mean-space}
\end{figure}
\subsection{Proof for Example~\ref{example:Gaussian-Location-Scale}: plug-in condition for the full Gaussian}\label{proof:guassian_example}
\ownparagraph{Proof of (\ref{eq:inf_variance})}
We first consider the case of a $\vec{\mu}$  located on the boundary $\partial \meanspace_q$ of ${\meanspace}_q$ (actually this cannot happen because $\meanspace_q$ is open, but it is still useful for proving the result as it constitutes a limiting worst-case). 
In this case, for fixed $0 < \alpha < 1$ and $x_0$, we have:
\begin{align}\label{eq:compute_var}
\mu^\circ_2 - \mu^{\circ 2}_1 = (1-\alpha)\mu_2 + \alpha x_{0,2} - \mu^{\circ 2}_1 = (1-\alpha)\mu^2_1 + \alpha x_{0,2} - ((1-\alpha)\mu_1 + \alpha x_{0,1})^2.
\end{align}
Taking the first derivative of $\mu^\circ_2 - \mu^{\circ 2}_1$: $\frac{d}{d\mu_1} (\mu^\circ_2 - \mu^{\circ 2}_1) = 2(1-\alpha)\alpha (\mu_1 - x_{0,1})$, we find
\begin{equation}\label{eq:argminmu}
\argmin\limits_{\vm \in \partial \meanspace_q}\  (\mu^{\circ}_2 - \mu_1^{\circ 2}) = (x_{0,1}, x_{0,1}^2).
\end{equation}
(note again $\vm^{\circ}$ is a function of $\vm$). Then (\ref{eq:inf_variance}) holds with the infimum over $\partial\meanspace_q$ instead of $\meanspace_q$, as seen by plugging (\ref{eq:argminmu}) into (\ref{eq:compute_var}).
The situation is depicted in Figure~\ref{fig:mean-space}(a). $l$ denotes $\mu^{\circ}_2 - (\mu^{\circ}_1)^2$, $l^*$ denotes $x_{0,2}- x_{0,1}^2$, and the statement (\ref{eq:inf_variance}) says that $\alpha l^* \leq l$. 

Now consider the case of a $\vec{\mu}' \in \meanspace_q$ not located on the boundary, and the corresponding $\vm'^{\circ}$. In terms of the figure, we now need to prove $\alpha l^* \leq l'$, where $l'$ denotes
$\mu'^{\circ}_2 - (\mu'^{\circ}_1)^2$. Let $\delta$ be the vertical translation distance between the black curve and the blue curve in the figure, i.e. the blue curve is $y=x^2 + \delta$. Extending the boundary-$\vm$ analysis to this case, we find $l' - \delta \geq \alpha (l^* - \delta)$. This implies $l' > \alpha l^*$ (the boundary case above is really the worst-case). Note that the same reasoning still applies when $x_0$ is below the blue curve --- then  $l^*$ is a negative real number and also $l$ might be negative in some cases.  (\ref{eq:inf_variance}) follows.

\ownparagraph{Proof that (\ref{eq:more}) holds}
For $\vm \in \mathtt{S}_q \subset \{ \vm \in  \meanspace_q: \| \vm - \vm^* \|_2 \geq A \}$ where $\mathtt{S}_q$ is compact, we have $\max\limits_{\vm \in \mathtt{S}_q} \frac{D(Q_{\vm^*} \| Q_{\vm^{\circ}})}{\| \vm - \vm^* \|^2_2} < \infty$ because $\frac{D(Q_{\vm^*} \| Q_{\vm^{\circ}})}{\| \vm - \vm^* \|^2_2}$ is continuous w.r.t. $\vec{\mu}$. Thus, we just need to check points that tend to the boundary of $\mathtt{M}_q$. 
For this, let $(\vm_{[m]})_{m \in \naturals}$ be a sequence in $
\{ \vm \in  \meanspace_q: \| \vm - \vm^* \|_2 \geq A \}$
such that the limit in (\ref{eq:finallimit}) below exists. 
We know by (\ref{eq:inf_variance}) that 
\begin{equation}\label{eq:finallimit}
 \lim_{m \rightarrow \infty} \frac{D(Q_{\vm^*} \| Q_{\vm^{\circ}_{[m]}})}{\| \vm_{[m]} - \vm^* \|^2_2} 
= O\left( 
\frac{\log ((\mu_2^{\circ})_{[m]} - (\mu_1^{\circ 2})_{[m]}) 
}{\| \vm_{[m]} - \vm^* \|^2_2} 
\right) +
O\left( 
\frac{
- 2 \mu_1^* (\mu_1^{\circ})_{[m]} + (\mu_1^{\circ 2})_{[m]}
}{\| \vm_{[m]} - \vm^* \|^2_2} 
\right).
\end{equation}
We need to show that for all such sequences with $\vm_{[m]}$ tending to the boundary of $\meanspace_q$, the above limit is finite. 
We first note that by (\ref{eq:inf_variance}), we cannot have
that $\mu_2^\circ - \mu^{\circ 2}_1 \rightarrow 0$, because   $\mu_2^\circ - \mu^{\circ 2}_1 \geq \alpha (x_{0,2} - x_{0,1}^2)$. Thus, we only need to consider the sequences with either  $\mu_2^\circ - \mu^{\circ 2}_1 \rightarrow \infty$ or  $|\mu_1^{\circ}| \rightarrow \infty$. Using that  $\| \vm - \vm^* \|^2_2 = (\mu_1 - \mu_1^*)^2 + (\mu_1^2 + \sigma_1^2 - \mu_1^{*2} - \sigma^{*2}_1)^2$, 
it is easily checked that (\ref{eq:finallimit}) is finite on all such sequences. 
\subsection{{\bf Plug-in} regularity conditions for Gamma model}
\begin{example}\label{example:Gamma}{\bf [Gamma]}
{\rm The argument follows analogous stages as the one for the Gaussian case, Example~\ref{example:Gaussian-Location-Scale}. Let $\althyp$ be the  $\textsc{Gamma}(\alpha,\beta)$ family with densities $q(x ;\alpha, \beta) = \frac{x^{\alpha-1} e^{-x/\beta}}{\Gamma(\alpha) \beta^\alpha}$, for $x>0$ and $\alpha, \beta>0$. In terms of the mean-value parameterization, we can parameterize this family by 
    $\vm = (\mu_1,\mu_2)$ with $\mu_1 = \psi(\alpha) + \log \beta$ and $\mu_2 = \alpha\beta$, where $\psi(\alpha)$ is the digamma function. 
    We  first claim
    \begin{equation}\label{eq:firstclaim}
    \meanspace_q = \{ (\mu_1, \mu_2): \mu_1 \in \reals, \mu_2 > e^{\mu_1}\}. 
    \end{equation}
To see this,     
rewrite $\vm = (\psi(\alpha) - \log\alpha + \log (\alpha\beta), \alpha\beta)$. For any fixed $c = \alpha\beta$, we have  that $\psi(\alpha) - \log\alpha$ is increasing and $\psi(\alpha) - \log\alpha \in (-\infty, 0), \alpha \in (0, \infty)$. Then the mean parameter space $\meanspace_q$ is located on the upper left part of $f(\alpha\beta) = e^{\alpha\beta}$, i.e. $\mu_2 = e^{\mu_1}$, which proves (\ref{eq:firstclaim}); see Figure~\ref{fig:mean-space}(b).
Next, it is well-known that 
\begin{align}\label{eq:gammakl}
D(Q_{\vm^*} \| Q_{\vm^{\circ}}) = 
\alpha^\circ \log \frac{\beta^\circ}{\beta^*} - \log\frac{\Gamma(\alpha^*)}{\Gamma(\alpha^\circ)} + (\alpha^* - \alpha^\circ)\psi(\alpha^*) - \left(\frac{1}{\beta^*} - \frac{1}{\beta^\circ} \right)\alpha^*\beta^*.
\end{align}
with $(\alpha^{\circ},\beta^{\circ})$ the parameters in the standard parameterization corresponding to $(\mu_1^{\circ}, \mu_2^{\circ})$.

Now let $\vm^{\circ} = (1-k) \vm + k x_0$. 
Further below we show
that for every  $0 < k < 1$, every  
$x_0 = (x_{0,1},x_{0,2}) \in \meanspace_q$, we have 
\begin{align}\label{eq:secondclaim}
\mu^\circ_1 - \log\mu^\circ_2 = \psi(\alpha^\circ) - \log\alpha^\circ \text{\ and\ }
0 < \alpha^\circ \leq \alpha_{x_0},    
\end{align}
where $\alpha_{x_0}$ is a constant depending on $x_0$. And 
we also show 
\begin{align}\label{eq:thirdclaim}
\mu^\circ_2 - e^{\mu^\circ_1} = (\alpha^\circ - e^{\psi(\alpha^\circ)})\beta^\circ \text{\ and\ }
\beta^\circ \geq \frac{k(x_{0,2} - e^{x_{0,1}})}{\alpha_{x_0} - e^{\psi(\alpha_{x_0})}}.
\end{align}
This provides constraints on the values  that $D(Q_{\vm^*} \| Q_{\vm^{\circ}})$ in (\ref{eq:gammakl}) can take. As we also show
below, it implies that for every $\vm^* \in \meanspace_q$, every $A > 0$, 
\begin{align}\label{neq:gamma_finite_bound}
\sup_{\vm \in \meanspace_q: \| \vm - \vm^* \|_2 > A} \frac{D(Q_{\vm^*} \| Q_{\vm^{\circ}})}{\| \vm - \vm^* \|^2_2} < \infty,
\end{align}
verifying Condition~\ref{cond:plugin}  as soon as $R$ has $5$ or more moments.}
\end{example}

\ownparagraph{Detailed proofs for Example~\ref{example:Gamma}}
To prove (\ref{eq:secondclaim}) and (\ref{eq:thirdclaim}), we first need to prove the following formulas (recall $\vm^{\circ}$ is a function of $\vm$):
\begin{align}
\sup_{\vm \in \meanspace_q} (\mu^\circ_1 - \log\mu^\circ_2) =& k (x_{0,1} - \log x_{0,2}),\label{eq:gamma_sup_alpha}\\
\inf_{\vm \in \meanspace_q} (\mu^\circ_2 - e^{\mu^\circ_1}) =& k (x_{0,2} - e^{x_{0,1}}).\label{eq:gamma_inf_beta}
\end{align}

\ownparagraph{Proof of (\ref{eq:gamma_sup_alpha})}
Reasoning analogously to the Gaussian case, Example~\ref{example:Gaussian-Location-Scale}, we first consider the case that $\vec{\mu}$ is located on the boundary $\partial \meanspace_q$ of ${\meanspace}_q$.

In this limiting case, for every fixed $0 < k < 1$, we get 
\begin{align*}
\delta_1 := \mu^\circ_1 - \log\mu^\circ_2
= (1-k) \log\mu_2 + k\cdot x_{0,1} - \log\left((1-k)\mu_2 + k\cdot x_{0,2} \right).
\end{align*}
Take the first derivative of $\delta_1$ w.r.t. $\mu_2$, we get $$
\argmax\limits_{\vm\in \partial \meanspace_q} (\mu^\circ_1 - \log\mu^\circ_2) = (\log x_{0,2}, x_{0,2}).$$ 
This implies (\ref{eq:gamma_sup_alpha}) holds (with $\meanspace_q$ replaced by $\partial \meanspace_q$), by plugging $(\log x_{0,2}, x_{0,2})$ into the above formula.
In the case of a $\vm'$ not located on the boundary (an instance is shown with corresponding $\vm'^{\circ}$ in 
Figure~\ref{fig:mean-space}(b)), we may consider the  line
connecting $x_0$ and $\vm'$ in the figure; it intersects the boundary of $\meanspace_q$ at some point $\vm^b$. Letting $\vm'^\circ = (1-k')\vm^b + k' x_0$, we have $k' > k$ since $\vm'^\circ = (1-k)\vm' + k\cdot x_0$. Using the above worst-case result, we get, in terms of Figure~\ref{fig:mean-space} (b),
$$
\delta_1' \leq k' (x_{0,1} - \log x_{0,2})
<
k (x_{0,1} - \log x_{0,2}),
$$
where the last inequality holds since $x_{0,1} - \log x_{0,2} < 0$. (\ref{eq:gamma_sup_alpha}) now follows.

\ownparagraph{Proof of (\ref{eq:gamma_inf_beta})}
We still consider the worst-case, $\vm$ on the boundary, and the case with $\vm'$ in the interior of $\meanspace_q$, as above.
In the boundary case, for every fixed $0 < k < 1$, we have:
\begin{align*}
\delta_2 := \mu^\circ_2 - e^{\mu^\circ_1}
= (1-k) e^{\mu_1} + k\cdot x_{0,2} - \exp\left((1-k)\mu_1 + k\cdot x_{0,1} \right).
\end{align*}
Take the first derivative of $\delta_2$ w.r.t. $\mu_1$, we get $\argmin_{\vm \in \partial \meanspace_q}  (\mu^\circ_2 - e^{\mu^\circ_1}) = (x_{0,1}, e^{x_{0,1}})$. Then this implies (\ref{eq:gamma_inf_beta}) holds (with $\meanspace_q$ replaced by $\partial \meanspace_q$) by plugging $(x_{0,1}, e^{x_{0,1}})$ into the above formula.
The case with $\vm' \in \meanspace_q$, not on the boundary, can now be proved in the same way as (\ref{eq:gamma_sup_alpha}); we omit further details.

\ownparagraph{Proof of (\ref{eq:secondclaim}) and (\ref{eq:thirdclaim})}
$\psi(\alpha^\circ) - \log\alpha^\circ$ is increasing and $\psi(\alpha^\circ) - \log\alpha^\circ = \mu^\circ_1 - \log\mu^\circ_2 \leq  k (x_{0,1} - \log x_{0,2})$ from (\ref{eq:gamma_sup_alpha}), which implies that $\alpha^{\circ}$ is bounded by some $\alpha_{x_0}$ depending on $x_0$. Further, $(\alpha^\circ - e^{\psi(\alpha^\circ)})\beta^\circ = \mu^\circ_2 - e^{\mu^\circ_1} \geq k (x_{0,2} - e^{x_{0,1}})$ from (\ref{eq:gamma_inf_beta}), $\alpha^\circ - e^{\psi(\alpha^\circ)}$ is increasing and $0 < \alpha^\circ \leq \alpha_{x_0}$, which implies $\beta^\circ \geq \frac{k(x_{0,2} - e^{x_{0,1}})}{\alpha_{x_0} - e^{\psi(\alpha_{x_0})}}$.

\ownparagraph{Proof of (\ref{neq:gamma_finite_bound})}
We have 
$||\vec{\mu} - \vec{\mu}^*||_2^2 = (\psi(\alpha) + \log\beta - \psi(\alpha^*) - \log\beta^*)^2 + (\alpha\beta - \alpha^* \beta^*)^2$. 
Reasoning analogously to the Gaussian case (\ref{eq:finallimit}), 
we need to show that for any sequence $(\vm_{[m]})_{m \in \naturals}$
with all $\vm_{[m]} \in \{ \vm \in \meanspace_q: \| \vm - \vm^* \|_2 \geq A \}$, 
that tends to the boundary of $\meanspace_q$, we have
\begin{align*}
\lim_{m \rightarrow \infty} \frac{D(Q_{\vm^*} \| Q_{(\vm^{\circ})_{[m]}})}{\| \vm_{[m]} - \vm^* \|^2_2} < \infty.
\end{align*}
We already know $0 < \alpha^\circ \leq \alpha_{x_0}$ and $\beta^\circ \geq \frac{k(x_{0,2} - e^{x_{0,1}})}{\alpha_{x_0} - e^{\psi(\alpha_{x_0})}}$, so the above limit satisfies, for some constant $C$,
\begin{align}
& \lim_{m \rightarrow \infty} \frac{D(Q_{\vm^*} \| Q_{(\vm^{\circ})_{[m]}})}{\| \vm_{[m]} - \vm^* \|^2_2} = 
O\left( \frac{\alpha^{\circ}_{[m]} \left( \log \beta^{\circ}_{[m]} + C \right) + \log \Gamma(\alpha^{\circ}_{[m]}) }{(\psi(\alpha_{[m]}) - \log \alpha_{[m]} + \log \beta_{[m]})^2 + (\alpha_{[m]} \beta_{[m]})^2} \right)\label{eq:lastone}  \\ 
& = \nonumber
O\left(
\frac{\alpha_{[m]}^{\circ} \log \beta_{[m]}^{\circ} }{(\psi(\alpha_{[m]})  + \log \beta_{[m]})^2}
\right) +
O\left(
\frac{\log \Gamma(\alpha_{[m]}^{\circ}) }{(\psi(\alpha_{[m]})  + \log \beta_{[m]})^2}
\right)
,
\end{align}
where we also plugged in (\ref{eq:gammakl}) and the definition of $\vm$ in terms of $\alpha,\beta$.
Any sequence tending to the boundary has  $\mu_2^\circ - e^{\mu^{\circ}_1} \rightarrow 0$ or  $\mu_2^\circ \rightarrow \infty$ or $\mu_1^\circ \rightarrow \infty$, i.e. $\alpha^\circ \rightarrow 0$ or  $\alpha^\circ \rightarrow \infty$ or  $\beta^\circ \rightarrow 0$ or $\beta^\circ \rightarrow \infty$. 
Again using that $\alpha^{\circ}$ is bounded above and $\beta^{\circ}$ is bounded below, 
we just need to check the cases (a) $\beta^{\circ} \rightarrow \infty, \alpha^{\circ}  \rightarrow \alpha^*$ with $\alpha^* \not \in \{0,\infty\}$, (b) $\beta^{\circ} \rightarrow \beta^*$, with $\beta^* \not \in \{0,\infty\}$, $\alpha^{\circ}  \rightarrow 0$, and (c) 
$\beta^{\circ} \rightarrow \infty, \alpha^{\circ}  \rightarrow 0$. 

In Case (a), $\beta^\circ \rightarrow \infty$,  and then $\mu^\circ_1 = \psi(\alpha^\circ) + \log\beta^\circ \rightarrow \infty$, so $\psi(\alpha) + \log \beta = \mu_1 \rightarrow \infty$ since $\mu^\circ_1 = (1-k)\mu_1 + k\cdot x_{0,1}$ and then the first term in (\ref{eq:lastone}) is $o(1)$ and the second $O(1)$. 
In Case (b), the first term in (\ref{eq:lastone}) is $o(1)$ and we can evaluate the second term by L'Hôpital's rule:
$$\lim\limits_{\alpha^\circ \rightarrow 0} \frac{\log\Gamma(\alpha^\circ)}{\psi(\alpha^\circ)^2} = \lim\limits_{\alpha^\circ \rightarrow 0} \frac{\psi(\alpha^\circ)}{2 \psi(\alpha^\circ) \cdot \psi'(\alpha^\circ)} = \lim\limits_{\alpha^\circ \rightarrow 0} \frac{1}{2 \psi'(\alpha^\circ)} < \infty,$$  since  $\psi'(\alpha^\circ) = \sum\limits_{k=0}^{\infty}\frac{1}{(\alpha^\circ + k)^2}$.
It remains to show Case (c); but this follows by combining Case (a) and Case (b) above.

\section{Additional simulation results for the setting of  Section~\ref{sec:simulations}}\label{app:Simulations}
The figures in this appendix are described in the main text, in Section~\ref{sec:simulations}.
\begin{figure}[h]
    \flushleft
    \subfigure[$\mu^* = (0.62, 0.38)$, $D(Q^*|| P_{\vec{\mu}^*}) \approx 0.0606$]{\includegraphics[width=0.4\linewidth]{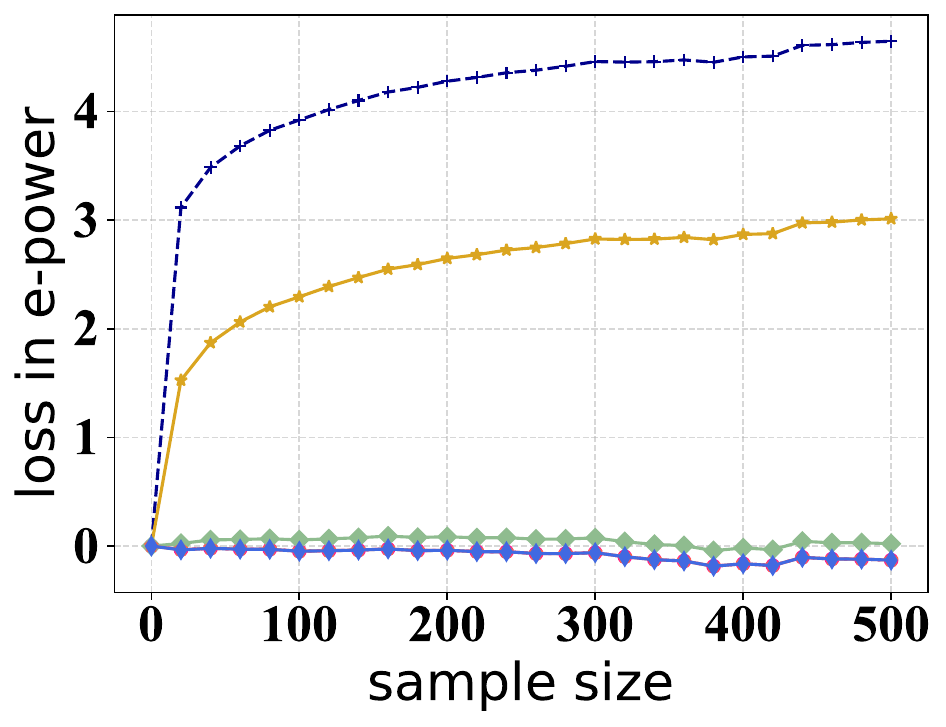}}
    \subfigure[$\mu^* = (0.83, 0.65)$, $D(Q^*|| P_{\vec{\mu}^*}) \approx 0.0458$]{ \includegraphics[width=0.575\linewidth]{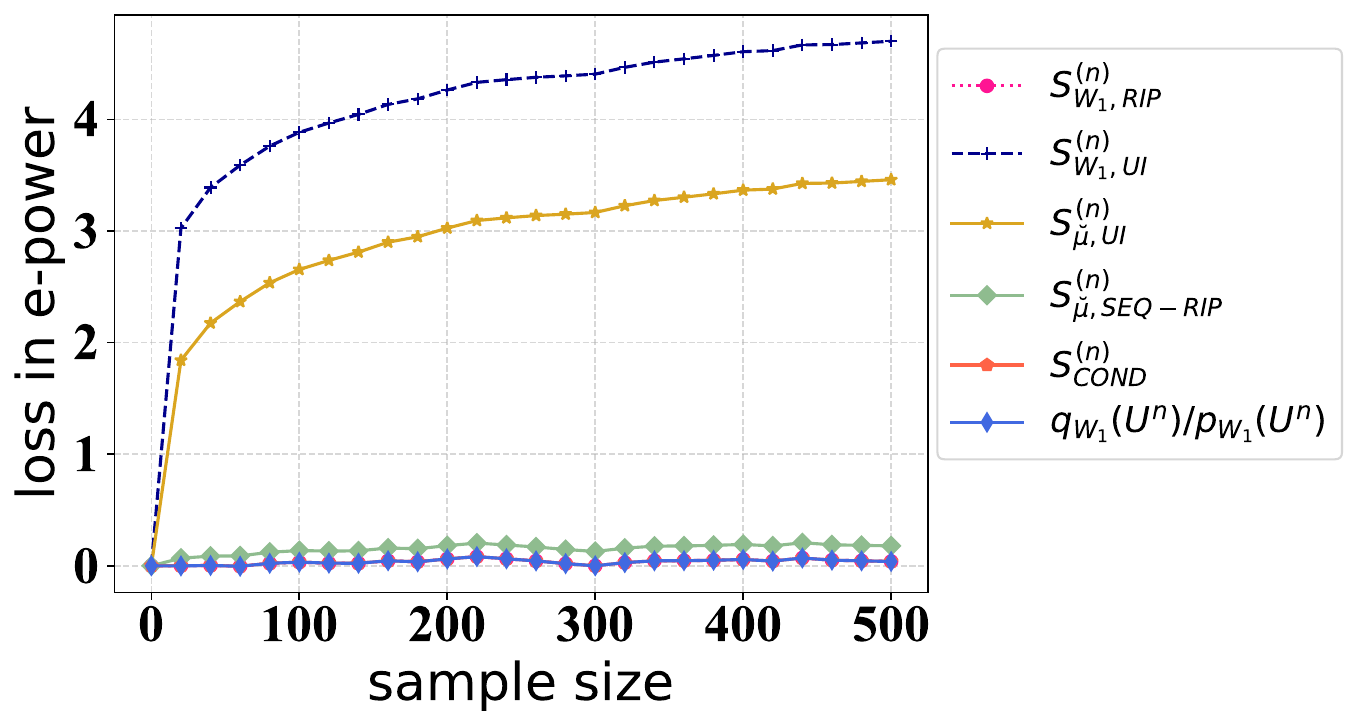}}
    
    \caption{$\delta = 1$. E-powers are computed for multiple (pseudo) e-variables. The data sequences are generated from $Q^*$ with mean $\mu^*$.}
    \label{fig:Appendix_simulation_results_delta1}
\end{figure}

\begin{figure}[h]
    \centering
    \subfigure[$\mu^* = (0.79, 0.34)$, $D(Q^*|| P_{\vec{\mu}^*}) \approx 0.2161$]{\includegraphics[width=0.65\linewidth]{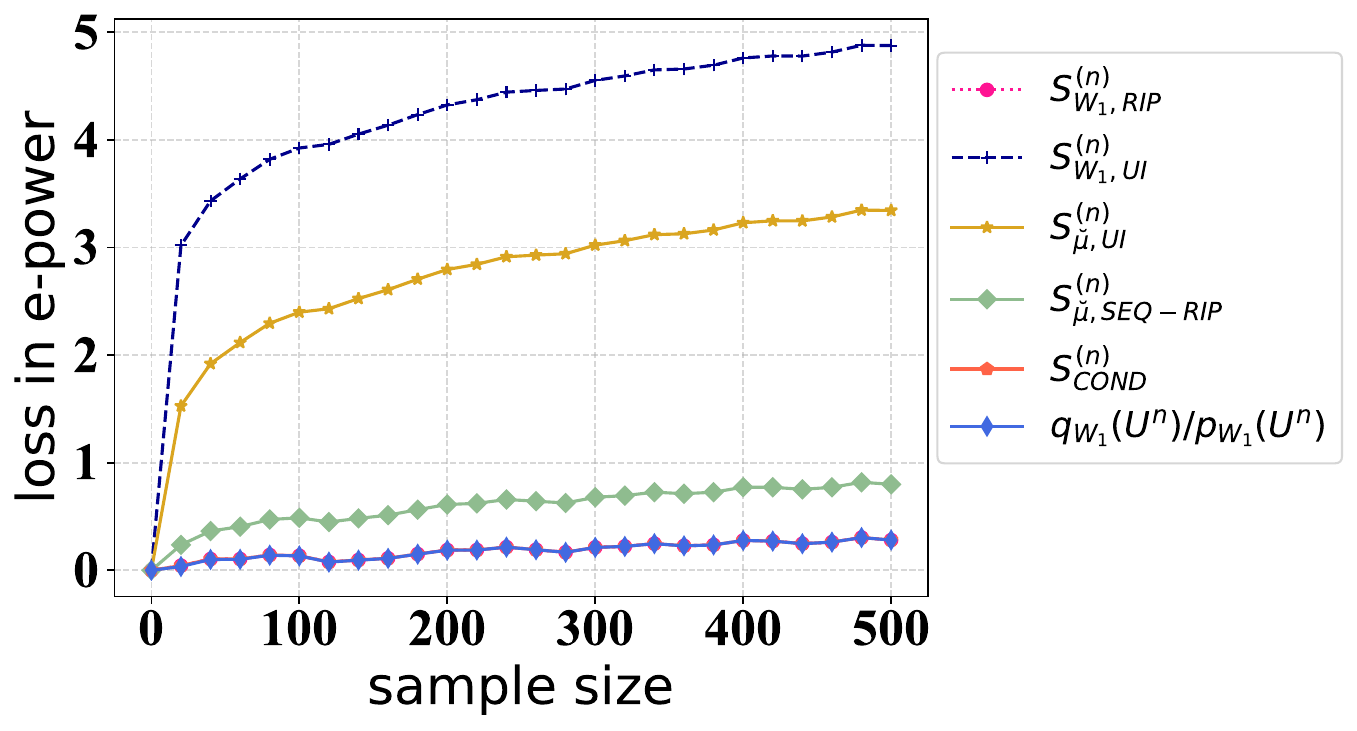}}
    \subfigure[$\mu^* = (0.46, 0.10)$, $D(Q^*|| P_{\vec{\mu}^*}) \approx 0.1673$]{\includegraphics[width=0.49\linewidth]{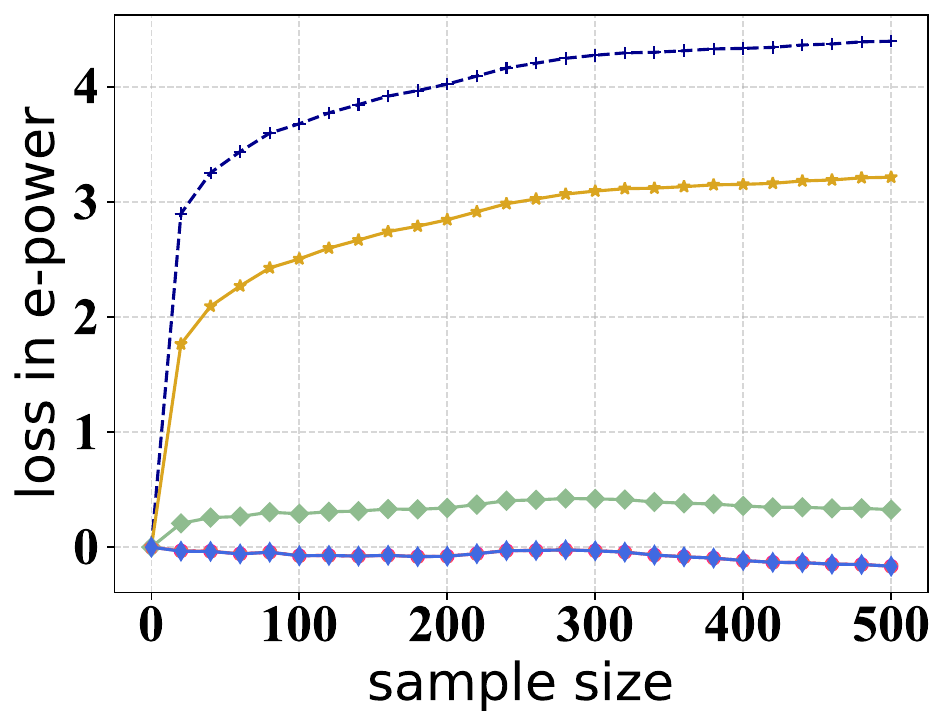}}
    \subfigure[$\mu^* = (0.97, 0.81)$, $D(Q^*|| P_{\vec{\mu}^*}) \approx 0.0712$]{\includegraphics[width=0.49\linewidth]{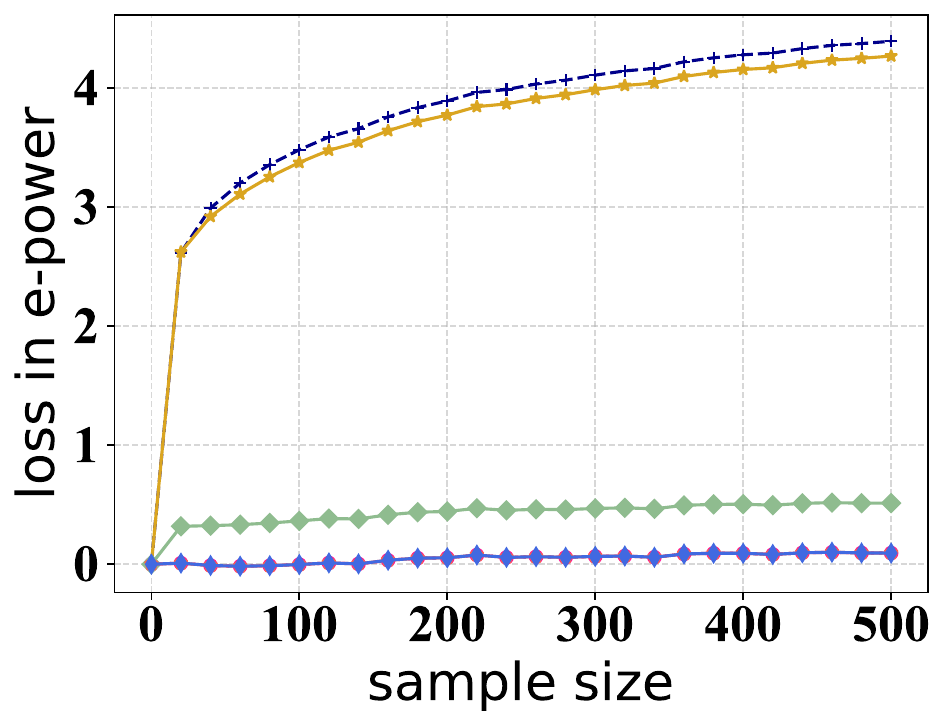}}
    
    \caption{$\delta = 2$. E-powers are computed for multiple (pseudo) e-variables. The data sequences are generated from $Q^*$ with mean $\mu^*$.}
    \label{fig:simulation_results_delta2}
\end{figure}

\section{Additional material on e-process-ness for Section~\ref{sec:conclusion}}\label{app:eprocess}
To define an e-process, suppose that ${\cal U} = \reals^d$ for some $d > 0$ and let $\sigma(U^n)$ be the $\sigma$-algebra generated by $U_1, \ldots, U_n$. We call $(\sigma(U^n))_{n \in \naturals}$ the {\em filtration induced by the data\/} or `data filtration' for short. 
An {\em e-process\/}  {\em relative to the data filtration\/} is a random process $(S^{(n)})_{n \in \naturals}$ defined relative to the data filtration such that for each stopping time $\tau$ (defined again relative to the data filtration), $S^{(\tau)}$ is an e-variable. We refer to \cite{ramdas2023savi} for more background on e-processes. 

The following result (proof further below) enables us to determine some scenarios in which either $S_{\cond}$ or $S_{\rip}$  do not provide e-processes.
\begin{proposition}
    \label{prop:notaprocess}
Suppose $\nulhyp = \{P_{\theta}: \theta \in \Theta_0 \}$ is a set of distributions (not necessarily an  exponential family) for $U$, extended to $n$ outcomes by independence, where $P_{\theta}$ has density $p_{\theta}$, and $Q$ is a distribution for random process $(U^n)_{n \geq 0}$ 
with densities $(q(U^n))_{n \geq 0}$. Let $W_1, W_2, \ldots$ be a sequence of distributions on $\Theta$ such that for every $n$, we have that $S^{(n)} := q(\sn{U}{n})/p_{W_{n}}(\sn{U}{n})$ is an e-variable for samples of size $n$.
Suppose that  for some $n$, 
\begin{equation}\label{eq:qunequal}
Q(P_{W_n}(U^n)\neq P_{W_{n+1}}(U^n)) > 0.\end{equation}
Then the process $(S^{(n)})_{n \geq 0}$ is not an e-process relative to the data filtration.
\end{proposition}
Note that
(\ref{eq:qunequal}) implies that, but is stronger than,
$W_n \neq W_{n+1}$. 
The proposition implies that if the process of RIPr e-variables $(S^{(n)}_{Q,\rip})_{n =1,2,\ldots}$ takes on the form $(P_{W_n}(U^n))_{n = 1,2,\ldots}$ for some sequence $W_1, W_2, \ldots$ as above, then the RIPr does not yield an e-process relative to the data filtration. In particular, this will be the case in the Gaussian anti-simple cases of Theorem~\ref{thm:simpleH1gauss} and Theorem~\ref{thm:compositeH1gauss}. In these cases, $S_{\rip}$ is equal to $S_{\cond}$, implying that $S_{\cond}$ also does not provide an e-process. Although we have no formal proof, the proof for the anti-simple case for general exponential family nulls with simple alternative (anti-simple case in Theorem~\ref{thm:simpleH1general}) suggests that the same holds for $S_{\cond}$ and $S_{\rip}$ in the {\em strict\/} anti-simple setting, since then the RIPr prior $W_0$ can be approximated by a Gaussian prior with variance of order $\asymp 1/n$, i.e. again changing over time.

However, an interesting phenomenon happens in the case of composite alternative $\althyp$, both in the Gaussian and in the general case. 
First, in the anti-simple setting in the general case (Theorem~\ref{thm:compositeH1general}, Part 4, anti-simple case) the RIPr prior $W_0$ is approximated by taking a prior {\em identical} to the prior $W_1$ that was put on $\althyp$, and this prior does {\em not\/} change with the sample size. Similarly, in the Gaussian case with composite alternative, the exact RIPr prior changes with $n$, but only very minimally so. This suggests that somehow, with the conditional and RIPr e-variables we `almost' obtain an e-process, leading perhaps to `approximate' handling of optional stopping. Investigating this further and formalizing `approximate optional stopping' is a main avenue for further research. 

To strengthen this conjecture further, consider the sequence of e-variables $(S^{(n)}_{\haar})_{n \geq 0}$ based on the right-Haar (uniform) prior on $\meanspace_p$ and $\meanspace_q$, as in Section~\ref{sec:gausscomposite}. 
It is known that $(S^{(n)}_{\haar})_{n \geq 0}$ defines an e-process, but relative to a {\em coarser\/} filtration generated by $(V^n)_{n\geq 0}$ with $V^n = (X_{(2)}- X_1, \ldots, X_{(n)} - X_1)$, which still allows optional stopping to all practical intents and purposes \cite{HendriksenHG21,perez2024estatistics}. Theorem~\ref{thm:compositeH1gauss} implies, by the equality of $S_{\haar}$ and $S_{\cond}$, that, at least in the anti-simple case, when $\rip = \cond$, since the prior $W_0$ in (\ref{eq:W0depends}) depends on $n$, by Proposition~\ref{prop:notaprocess}, then $S_{\haar}$ does {\em not\/} define an e-process for the original, data-based filtration; but the fact that it {\em is\/} an e-process in a coarser filtration suggests that perhaps something similar holds (``being an approximate e-process in a coarser filtration''), for other conditional e-variables as well; the insights on asymptotic anytime-validity by \cite{waudby2021time} may be of use here. 

Finally, we note that  there do exist exponential family settings in which $S_{\cond}$ provides an e-process relative to the data filtration. This happens in the  settings of Proposition~\ref{prop:seqcondwins}, i.e. in Example~\ref{ex:twosample}, in which we are in the nonstrict simple and anti-simple settings at the same time.

\ownparagraph{Proof of Proposition~\ref{prop:notaprocess}}
If $S^{(n)}_{n \geq 0}$ were an e-process relative to null hypothesis $\nulhyp$, then, as is immediate from the definition of e-process, it must also be an e-process relative to the simple null hypothesis $\{P_{W_i}\}$ for any $i \geq 0$. 
    Thus, if  ${\mathbb E}_{U^n \sim P_{W_{n}}} [q(U^n)/p_{W_{n+1}}(U^n)] > 1$, then $(S^{(n)})_{n\geq 0}$ cannot form an e-process and the result is proved. Therefore, we may  assume that  ${\mathbb E}_{U^n \sim P_{W_{n}}} [q(U^n)/p_{W_{n+1}}(U^n)] \leq 1$, or equivalently,
    ${\mathbb E}_{U^n \sim Q} \left[\frac{p_{W_n}(U^n)}{p_{W_{n+1}}(U^n)}\right] \leq 1.$
    Using Jensen's inequality, which by our assumption (\ref{eq:qunequal}) is strict,  now gives 
    $${\mathbb E}_{U^n \sim Q} \left[\frac{p_{W_{n+1}}(U^n)}{p_{W_{n}}(U^n)} \right] > 1/\left( {\mathbb E}_{U^n \sim Q} \left[\frac{p_{W_n}(U^n)}{p_{W_{n+1}}(U^n)}\right] \right)  \geq 1, 
    $$ 
    or equivalently, ${\mathbb E}_{U^n \sim P_{W_{n+1}}} [q(U^n)/p_{W_{n}}(U^n)] > 1$, and the result again follows.
\end{supplement}
\end{document}